\documentclass[runningheads]{llncs}

\usepackage[T1]{fontenc}

\usepackage{graphicx}

\usepackage{color}
\usepackage{hyperref}

\usepackage{tikz}
\usetikzlibrary{patterns}
\usetikzlibrary{arrows.meta,decorations.pathreplacing,patterns,backgrounds,positioning}
\usepackage{booktabs}
\usepackage{amssymb}

\usepackage[compatibility=false]{caption}
\usepackage{subcaption}
\providecommand{\Val}{\mathsf{Val}}
\providecommand{\status}{\mathsf{status}}
\providecommand{\PENDING}{\textsc{pending}}
\providecommand{\VALID}{\textsc{valid}}
\providecommand{\INVALID}{\textsc{invalid}}
\usepackage{amsmath}

\begin{document}
\title{Security Limits of Mining Before Validation in Nakamoto Consensus}
\titlerunning{Security Limits of Mining Before Validation}

\author{Yifan Zhou\inst{1,2}\fnmsep\thanks{This work was done while Yifan Zhou was visiting the University of Warwick.}
\and Jiang Xiao\inst{1} \and Kaihua Qin\inst{2}}
\authorrunning{Y. Zhou et al.}
\institute{Huazhong University of Science and Technology
\and
University of Warwick}

\maketitle

\begin{abstract}
Mining before validation allows miners to extend a newly received block before completing its validity checks, giving them a head start in the race for the next block reward.
This head start, however, comes with a security risk: rejecting one invalid block also discards the honest work built on it, an effect missed when validation is treated as instantaneous.
We quantify this risk in a model of Nakamoto consensus with bounded network delay and a validation-time bound independent of processing load.
We establish an explicit threshold on adversarial mining power below which honest miners' fully validated chains continue to grow and agree on a stable history with high probability over any fixed observation period.
We also construct an attack that repeatedly draws honest mining onto invalid branches.
When the adversary produces more than one block on average during the allowed validation time, the attack can eventually remove a target block at any fixed initial confirmation depth.
Combined with ordinary private mining, this attack yields a matching asymptotic bound in the fully decentralized regime, where each honest miner has negligible mining power.
The results show how validation latency limits the security of mining before validation, beyond the constraint imposed by network delay.

\end{abstract}
\section{Introduction}

Nakamoto consensus, introduced in Bitcoin,\footnote{\url{https://bitcoin.org/bitcoin.pdf}} enables Byzantine fault-tolerant \emph{state machine replication} (SMR)~\cite{lamport2019time,schneider1990implementing,castro1999practical} in a permissionless setting.
It provides probabilistically secure atomic broadcast~\cite{cachin2001secure,hadzilacos1993fault} among untrusted miners that may join or leave dynamically~\cite{lewis2023permissionless}.
In essence, Nakamoto consensus combines proof-of-work (PoW) mining for Sybil resistance with a longest-chain rule for determining the ledger history.
To understand the security of Nakamoto consensus, two core security questions have been extensively studied: how much adversarial mining power it can tolerate and whether the private-chain attack, in which the adversary withholds a competing chain, is optimal.

Prior analyses~\cite{garay2015bitcoin,pass2017analysis,ren2019analysis,gavzi2020tight,dembo2020everything} address these questions by comparing the progress of honest miners' chains (\emph{honest progress}) with that of adversarial competing branches (\emph{adversarial progress}).
They identify conditions on the honest mining share under which honest progress outpaces adversarial progress in expectation, then use concentration bounds to obtain high-probability security guarantees.
Network delay complicates this comparison: a newly mined block takes time to reach other miners~\cite{decker2013information}, who may meanwhile produce competing blocks.
Only one branch can remain on the eventual canonical chain, so not all honest mining work contributes to chain progress.
Accounting for this loss yields resilience lower bounds~\cite{garay2015bitcoin,pass2017analysis,garay2024bitcoin,kiffer2018better,ren2019analysis}, including tight asymptotic bounds that match the private-chain attack threshold in their respective models~\cite{dembo2020everything,gavzi2020tight}.

These analyses largely treat block validation, the process of checking whether a block and its transactions satisfy the protocol's rules, as instantaneous.
In Bitcoin, however, validation can take substantial time.\footnote{\url{https://bip54.org}}
Miners receive rewards for producing blocks but no separate reward for validating others' blocks, an incentive mismatch underlying the verifier's dilemma~\cite{luu2015demystifying}.
To begin extending a newly received block sooner, some miners adopt \emph{simplified payment verification (SPV) mining}, or mining before validation.
We study this behavior when a block's payload has arrived but its validation is still pending.
Extending such a block may advance the honest chain if the block and its ancestry are valid, but the resulting work must be discarded otherwise.
Mining on an invalid parent caused Bitcoin's July 2015 temporary forks.\footnote{\url{https://bitcoin.org/en/alert/2015-07-04-spv-mining}}
Thus, an adversary can erase multiple honest descendants with one invalid block, without building an equally long valid fork.

\noindent\textbf{Our contributions.} We establish security properties in our model of Nakamoto consensus, where all honest miners perform SPV mining.
These comprise the standard chain growth and common prefix properties~\cite{garay2015bitcoin}, together with block liveness and mining-time persistence adapted from prior analyses~\cite{dembo2020everything}.
The guarantees concern each honest miner's longest fully validated chain and hold with high probability over any fixed finite observation horizon against every admissible adversarial strategy from genesis whenever $\lambda_a<\lambda_h/(1+\lambda_h\Delta)$ and $\lambda_a\tau<1$.
Here $\lambda_h$ and $\lambda_a$ are the honest and adversarial mining rates, $\Delta$ bounds network delay, and $\tau$ is a separate, load-independent validation bound.
For fixed model parameters and failure probability, sufficient confirmation windows and common-prefix depth grow logarithmically with the horizon.
We also derive a sufficient resilience bound for mixed honest mining policies, parameterized by the SPV share of honest mining power (Appendix~\ref{app:mixed-honest}).

Our proof distinguishes honest work on valid ancestry from work lost to invalid ancestry.
An invalid adversarial block, called a \emph{carrier}, can divert honest mining before its rejection.
We define effective honest mining time to account for mining on valid ancestry and bound usable unpublished adversarial stock uniformly over the observation horizon.
The stock bound accounts for earlier blocks that remain usable when an interval begins.
Together, these bounds compare remaining honest progress with valid competing work under arbitrary adaptive allocations of adversarial power between carriers and competing branches.
We adapt the permanent-block approach of Dembo et al.~\cite{dembo2020everything} to establish block liveness, derive mining-time persistence, and then obtain common prefix.

We also construct a carrier attack that eventually removes a target block from a shared valid chain under an admissible delivery and validation schedule.
It succeeds with probability one when $\lambda_a\tau>1$, for any prescribed finite initial confirmation depth.
Together with ordinary private mining~\cite{dembo2020everything}, this attack gives an asymptotic upper boundary matching our sufficient resilience bound in the fully decentralized regime, where individual honest mining power is negligible.

\section{Model and Security Definitions}\label{sec: model}

In this work, we consider a permissionless blockchain running Nakamoto consensus~\cite{lewis2023permissionless}.
We adopt the continuous-time model of Nakamoto consensus used in prior analyses~\cite{ren2019analysis,dembo2020everything}.
To model SPV mining, we extend this framework by bounding block-validation latency separately from network delay and allowing honest miners to extend blocks before validation completes.

\noindent\textbf{Blocks and chains.} All blocks mined by time $t$ form a tree $\mathcal{T}(t)$ rooted at the genesis block $b_0$.
Each non-genesis block cryptographically commits to its unique, previously mined parent.
All parties initially know $b_0$, and the adversary initially has no other blocks.
Let $\mathcal{T}_i(t)$ denote the subtree known to party $i$ at time $t$. We consider a block known to a party only once the party has received its header, payload, and all ancestor blocks.
The height $h(B)$ is the number of parent edges from $b_0$ to $B$.
A chain $\mathcal{C}$ is a path starting at $b_0$, and $|\mathcal{C}|$ counts its non-genesis blocks.
The last block of a chain is called its \emph{tip}.
Let $\mathcal{C}^{\lceil k}$ denote $\mathcal{C}$ with its last $k$ blocks pruned, or just genesis if $|\mathcal{C}|\leq k$.
The relation $\mathcal{C}_1\preceq \mathcal{C}_2$ means that $\mathcal{C}_1$ is a prefix of $\mathcal{C}_2$.
For blocks, $A\preceq B$ means that $A$ is an ancestor of $B$, including $A=B$.

\noindent\textbf{Block mining.} Miners are either honest or adversarial.
Honest miners follow the protocol, while adversarial miners act under the control of a single adversary.
We model mining successes as independent Poisson processes.
The total rate is $\lambda>0$, the adversarial share is $\beta\in[0,1)$, and the honest and adversarial rates are
$\lambda_h=(1-\beta)\lambda$ and $\lambda_a=\beta\lambda$.
For each execution, honest miners have fixed power fractions $w_i$ with $\sum_i w_i=1$ and individual rates $w_i\lambda_h$.
We allow any finite number of honest miners with any such power distribution, while keeping the total mining rate fixed throughout each execution.
Each mining success produces one new block whose parent was selected before that success.

\noindent\textbf{Network model.} An honest miner immediately broadcasts its new block and its ancestry.
If a block first becomes known to an honest miner at time $t$, whether by creation or receipt, every honest miner receives it by $t+\Delta$.
Receipt times refer to availability with complete ancestry, and miners do not extend orphan blocks.
Neither the mining rule nor the fork-choice rule needs to know $\Delta$.

\noindent\textbf{Adversary model.} The adversary controls miners with aggregate mining rate $\lambda_a=\beta\lambda$ and coordinates them without communication delay.
A block produced by these miners is called an \emph{adversarial block}, whether valid or invalid.
The adversary may mine on any block it knows and withhold its blocks until it chooses to release them to an honest miner.
Releasing a block supplies its previously unavailable ancestry.
The adversary schedules delivery separately for each recipient within the network delay bound and resolves ties between longest eligible chains.
All mining choices, releases, delivery schedules, and validation schedules are causal: they may depend on the execution history and private randomness, but not on future PoW successes.

\noindent\textbf{Validity and validation time.} Every modeled block passes the preliminary header and PoW checks.
Its residual validity is a fixed bit $\Val(B)\in\{0,1\}$, assigned when the block is created, with $\Val(b_0)=1$.
A chain is objectively valid exactly when every block on it has validity bit $1$.
Only adversarially produced blocks may have validity bit $0$.
An honestly produced descendant of an invalid block can therefore have validity bit $1$ while belonging to an invalid chain.

Let $r_i(B)$ be the time miner $i$ receives $B$, and let $\nu_i(B)$ be the time it learns the truthful value of $\Val(B)$.
For every received block, we assume $r_i(B)\leq\nu_i(B)\leq r_i(B)+\tau$.
Set $r_i(b_0)=\nu_i(b_0)=0$. Before receipt, these times are treated as infinite in local eligibility tests.
Validation starts on receipt and completes within $\tau$ even when other validations, including those of ancestors, are pending.
This is an abstract, load-independent validation oracle. It excludes validation queues and unavailable payloads.
The adversary may choose any causal completion schedule within the bound, separately for each miner, but cannot change a block's validity bit.
The model permits worst-case delays. It does not assert that every such schedule can be realized by a concrete payload or implementation.

\noindent\textbf{Execution status and eligible sets.} For a known block $B$, define
\[
\status_i(B,t)=
\begin{cases}
\INVALID,
 &\text{if some $A\preceq B$ satisfies $\nu_i(A)\leq t$}\\
 &\qquad\text{and $\Val(A)=0$,}\\[1mm]
\VALID,
 &\text{if every $A\preceq B$ satisfies $\nu_i(A)\leq t$}\\
 &\qquad\text{and $\Val(A)=1$,}\\[1mm]
\PENDING, &\text{otherwise.}
\end{cases}
\]
A chain remains pending until all its blocks have been validated or one has been found invalid.
Both $\VALID$ and $\INVALID$ are final states.
The optimistic and strict eligible sets are, respectively,
\begin{align*}
E_i^{\mathrm{opt}}(t)&=\{B\in\mathcal{T}_i(t):\status_i(B,t)\neq\INVALID\},\\
E_i^{\mathrm{str}}(t)&=\{B\in\mathcal{T}_i(t):\status_i(B,t)=\VALID\}.
\end{align*}

\noindent\textbf{Optimistic and strict chains.} Let $\mathcal{C}_i^{\mathrm{opt}}(t)$ and $\mathcal{C}_i^{\mathrm{str}}(t)$ be longest chains in the respective eligible sets.
All honest miners in the main analysis are optimistic: they extend $\mathcal{C}_i^{\mathrm{opt}}(t)$ while validation is pending.
Appendix~\ref{app:mixed-honest} extends the security guarantees to a fixed fraction of honest mining power following this policy, with the remainder extending fully validated chains.
When a block is rejected, the miner removes its entire descendant subtree from the eligible set and selects a longest remaining chain.
We evaluate security on $\mathcal{C}_i^{\mathrm{str}}(t)$, the longest fully validated chain known to miner $i$.
This chain is selected separately from the optimistic chain and need not be its validated prefix.
Strict-chain height is nondecreasing, although the selected chain can change when an equally long or longer valid branch arrives.

\noindent\textbf{Security properties.} All properties below concern strict chains. We fix a finite observation horizon $T$.
The true creation time of a block $B$ is denoted by $m(B)$.
Mining-time confirmation is an analytical definition, since a block header need not reveal its true creation time.

\begin{definition}[Chain growth]\label{def:chain-growth}
For a minimum growth rate $\gamma>0$, measured in blocks per unit time, and a duration $s>0$, an execution satisfies $(\gamma,s)$-chain growth through $T$ if, for every honest miner $i$ and every $t$ with $0\leq t\leq T-s$,
\[
 |\mathcal{C}_i^{\mathrm{str}}(t+s)|-|\mathcal{C}_i^{\mathrm{str}}(t)|\geq\gamma s.
\]
\end{definition}

\begin{definition}[Common prefix]\label{def:common-prefix}
For $k\in\mathbb{N}$, an execution satisfies $k$-common prefix through $T$ if
\[
 (\mathcal{C}_i^{\mathrm{str}}(t_1))^{\lceil k}\preceq \mathcal{C}_j^{\mathrm{str}}(t_2)
\]
for all honest miners $i,j$ and $0\leq t_1\leq t_2\leq T$.
\end{definition}

We adapt the persistence notion of Dembo et al.~\cite{dembo2020everything} to strict chains.
For persistence and block liveness, $T$ limits the confirmation times and block-creation windows covered simultaneously, while permanence extends to all future times.

\begin{definition}[Mining-time persistence]\label{def:persistence}
For a confirmation duration $\sigma>0$, a non-genesis block $B$ in a strict chain with tip $D$ is $\sigma$-confirmed if $m(D)-m(B)\geq\sigma$.
Persistence with observation horizon $T$ means that the prefix ending at any $\sigma$-confirmed block observed in an honest strict chain at time $t\leq T$ remains a prefix of every honest strict chain at every time $t'\geq t$.
\end{definition}

\begin{definition}[Block liveness]\label{def:liveness}
An execution satisfies block liveness with window $\sigma>0$ and observation horizon $T$ if every interval $(s,s+\sigma]\subseteq[0,T]$ contains the creation of an honest block that belongs to every honest strict chain at every time $t'\geq s+\sigma$.
\end{definition}


\section{Honest Progress and Revocation under SPV Mining}\label{sec: warmup}

The security properties defined in Section~\ref{sec: model} concern fully validated chains, but SPV miners may extend blocks whose validation is pending.
An honest block mined on such a branch is discarded if an ancestor is later found invalid.

\subsection{Honest progress}
To understand how SPV mining changes the contribution of honest work, we first consider mining on valid ancestry.
Network delay allows honest miners to work on different local views, so several honest successes may add only one level to a chain.
We use \emph{honest progress} to measure chain advancement after discounting this latency-induced concurrency, with each unit representing one additional level.
Different analyses quantify this contribution differently. Backbone counts isolated blocks~\cite{garay2024bitcoin}, whereas the block-race analysis accounts more tightly for concurrent honest work~\cite{dembo2020everything}.
Generally, a new unit of honest progress can be guaranteed by allowing $\Delta$ time for the preceding level to propagate and then waiting for the next honest success, which takes $1/\lambda_h$ time on average.
Figure~\ref{fig:warmup-counting} illustrates how concurrent blocks, even with different parents, are grouped into one counted unit per propagation window.
The block-race framework~\cite{dembo2020everything} obtains the resulting asymptotic rate $r_\star=1/(\Delta+1/\lambda_h)=\lambda_h/(1+\lambda_h\Delta)$, which is tight in the limit of negligible individual honest mining power.

\begin{figure}[t]
  \centering
  \begin{subfigure}[t]{0.43\linewidth}
    \centering
    \begin{tikzpicture}[
      x=1cm,y=1cm,font=\scriptsize,
      block/.style={draw,fill=white,minimum size=0.20cm,inner sep=0pt},
      progress/.style={draw,fill=white,text width=1.04cm,
        minimum height=0.56cm,align=center,inner sep=2pt},
      lab/.style={inner sep=1pt,align=center},
      link/.style={-{Latex[length=1.1mm]},
        shorten >=0.5pt,shorten <=0.5pt,line width=0.4pt},
      arrow/.style={-{Latex[length=1.3mm]},line width=0.4pt}
    ]
      \path[use as bounding box] (0,0) rectangle (5.16,3.35);

      \foreach \a/\b in {0.80/2.25,3.30/4.75} {
        \fill[black!7] (\a,1.60) rectangle (\b,3.00);
        \draw[decorate,decoration={brace,amplitude=2.5pt}]
          (\a,3.03) -- (\b,3.03)
          node[lab,midway,above=3pt] {$\Delta$};
      }

      \node[block] (root) at (0.15,2.25) {};
      \node[block] (a1) at (0.80,2.72) {};
      \node[block] (a2) at (1.38,2.28) {};
      \node[block] (a3) at (1.95,1.87) {};
      \foreach \n in {a1,a2,a3}
        \draw[link] (\n.west) -- (root.east);

      \node[block] (b1) at (3.30,2.72) {};
      \node[block] (b2) at (3.88,2.28) {};
      \node[block] (b3) at (4.45,1.87) {};
      \foreach \child/\parent in {b1/a3,b2/a1,b3/a2}
        \draw[link] (\child.west) -- (\parent.east);

      \draw[arrow] (0.12,1.60) -- (5.08,1.60);
      \foreach \x/\s in {
        0.80/$t_1$,2.25/$t_1+\Delta$,
        3.30/$t_2$,4.75/$t_2+\Delta$}
        \draw (\x,1.65) -- (\x,1.55)
          node[lab,below=2pt] {\s};

      \node[progress] (p1) at (1.52,0.49) {honest\\progress};
      \node[progress] (p2) at (4.02,0.49) {honest\\progress};
      \draw[arrow,densely dotted] (1.52,1.32) -- (p1.north);
      \draw[arrow,densely dotted] (4.02,1.32) -- (p2.north);
    \end{tikzpicture}
    \caption{Counting progress}
    \label{fig:warmup-counting}
  \end{subfigure}\hfill
  \begin{subfigure}[t]{0.27\linewidth}
    \centering
    \begin{tikzpicture}[
      x=1cm,y=1cm,font=\scriptsize,
      block/.style={draw,fill=white,minimum size=0.28cm,inner sep=0pt},
      progress/.style={draw,fill=white,text width=1.04cm,
        minimum height=0.56cm,align=center,inner sep=2pt},
      lab/.style={inner sep=1pt,align=center},
      arrow/.style={-{Latex[length=1.3mm]},line width=0.4pt}
    ]
      \path[use as bounding box] (0,0) rectangle (3.18,3.35);

      \node[block] (parent) at (0.30,2.48) {};
      \node[progress] (p) at (1.65,2.48) {honest\\progress};
      \draw[arrow,shorten >=1pt,shorten <=1pt]
        (p.west) -- (parent.east);
      \node[lab] at (0.30,3.03) {$\ell$};
      \node[lab] at (1.65,3.03) {$\ell+1$};

      \node[lab] at (0.43,1.20) {all:\\$\geq\ell$};
      \node[lab] (result) at (2.66,1.20) {all:\\$\geq\ell+1$};
      \draw[arrow,densely dotted] (p.south east) -- (result.north);

      \draw[arrow] (0.12,0.56) -- (3.07,0.56);
      \foreach \x/\s in {1.65/$t$,2.66/$t+\Delta$}
        \draw (\x,0.62) -- (\x,0.50)
          node[lab,below=2pt] {\s};
    \end{tikzpicture}
    \caption{Chain growth}
    \label{fig:warmup-growth}
  \end{subfigure}\hfill
  \begin{subfigure}[t]{0.27\linewidth}
    \centering
    \begin{tikzpicture}[
      x=1cm,y=1cm,font=\scriptsize,
      progress/.style={draw,fill=white,text width=1.04cm,
        minimum height=0.56cm,align=center,inner sep=2pt},
      adv/.style={draw,fill=black,text=white,
        minimum size=0.36cm,inner sep=1pt},
      root/.style={draw,circle,fill=white,
        minimum size=0.09cm,inner sep=0pt},
      link/.style={-{Latex[length=1.3mm]},
        shorten >=1pt,shorten <=1pt},
      lab/.style={inner sep=1pt}
    ]
      \path[use as bounding box] (0,0) rectangle (3.18,3.35);

      \node[lab] at (1.94,2.94) {honest branch};
      \node[root] (root) at (0.28,1.61) {};
      \node[progress] (h) at (1.94,2.32) {honest\\progress};
      \node[adv] (a) at (1.94,0.89) {$a_1$};

      \draw[link] (h.west) -- (root);
      \draw[link] (a.west) -- (root);
      \draw[densely dotted] (h.south) -- (a.north);
      \node[lab] at (1.94,0.36) {private branch};
    \end{tikzpicture}
    \caption{Competition}
    \label{fig:warmup-prefix}
  \end{subfigure}

  \caption{Honest progress.
  (a)~Concurrent forks contribute one counted unit per propagation window.
  (b)~A unit raises the common height bound.
  (c)~A private branch needs one adversarial block per unit.}
  \label{fig:warmup}
\end{figure}
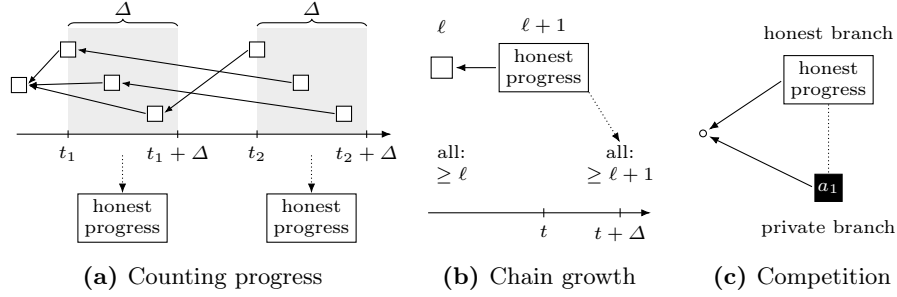

Next, we illustrate how the abstraction of honest progress supports the analysis of chain growth and common prefix. \emph{Chain growth} requires every honest miner's chain to grow at a positive rate over sufficiently long intervals.
The one-level increments in Figure~\ref{fig:warmup-growth} accumulate at rate $r_\star$, yielding any lower chain-growth rate with high probability.
The $k$-\emph{common-prefix} property requires the prefix remaining after pruning the last $k$ blocks to appear in every honest chain at the same or a later time.
A violation therefore requires a competing branch to replace a suffix longer than $k$ blocks.
Producing this suffix requires more than $k$ mining successes, so a large $k$ entails a long comparison interval with high probability.
The question is therefore whether conflicting branches can remain competitive over such intervals.
Figure~\ref{fig:warmup-prefix} illustrates this competition for a private fork, where one adversarial block matches one level of honest progress.
When all blocks are valid, $r_\star>\lambda_a$ gives honest progress a rate advantage over adversarial production.
Controlling competing branches over long intervals turns this advantage into a bound on the probability of disagreement at depth $k$.

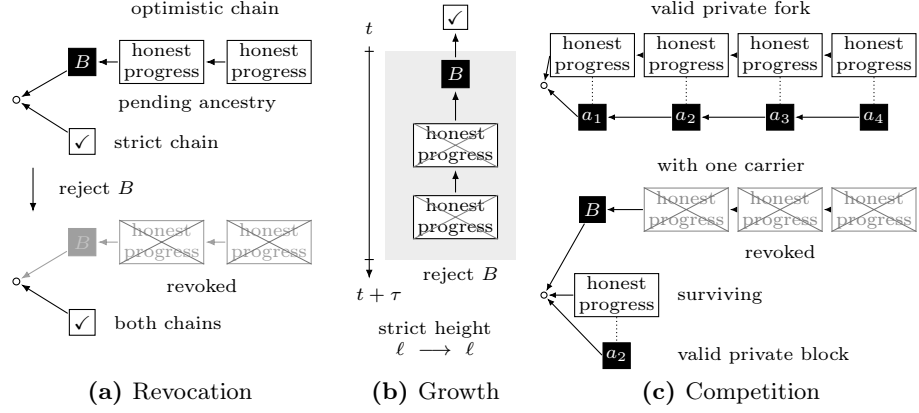
\begin{figure}[t]
  \centering
  \begin{subfigure}[t]{0.36\linewidth}
    \centering
    \begin{tikzpicture}[
      x=1cm,y=1cm,font=\scriptsize,
      progress/.style={draw,fill=white,text width=1.04cm,
        minimum height=0.56cm,align=center,inner sep=1pt},
      valid/.style={draw,fill=white,minimum size=0.34cm,inner sep=0pt},
      adv/.style={draw,fill=black,text=white,
        minimum size=0.36cm,inner sep=1pt},
      root/.style={draw,circle,fill=white,
        minimum size=0.09cm,inner sep=0pt},
      link/.style={-{Latex[length=1.2mm]},
        shorten >=0.5pt,shorten <=0.5pt},
      lab/.style={inner sep=1pt,align=center}
    ]
      \path[use as bounding box] (0,0) rectangle (4.32,5.00);

      \node[lab] at (2.63,4.78) {optimistic chain};
      \node[root] (r0) at (0.13,3.58) {};
      \node[adv] (b0) at (1.00,4.08) {$B$};
      \node[progress] (h10) at (2.05,4.08) {honest\\progress};
      \node[progress] (h20) at (3.47,4.08) {honest\\progress};
      \draw[link] (b0) -- (r0);
      \draw[link] (h10.west) -- (b0.east);
      \draw[link] (h20.west) -- (h10.east);
      \node[lab] at (2.52,3.53) {pending ancestry};
      \node[valid] (s0) at (1.00,3.03) {$\checkmark$};
      \draw[link] (s0) -- (r0);
      \node[lab,anchor=west] at (1.38,3.03) {strict chain};

      \draw[link] (0.35,2.70) -- (0.35,2.08);
      \node[lab,anchor=west] at (0.65,2.40) {reject $B$};

      \node[root] (r1) at (0.13,1.18) {};
      \begin{scope}[opacity=0.38]
        \node[adv] (b1) at (1.00,1.70) {$B$};
        \node[progress] (h11) at (2.05,1.70) {honest\\progress};
        \node[progress] (h21) at (3.47,1.70) {honest\\progress};
        \draw[link] (b1) -- (r1);
        \draw[link] (h11.west) -- (b1.east);
        \draw[link] (h21.west) -- (h11.east);
      \end{scope}
      \foreach \n in {h11,h21}
        \draw[black!55,line width=0.35pt]
          (\n.north west) -- (\n.south east)
          (\n.north east) -- (\n.south west);
      \node[lab] at (2.58,1.13) {revoked};
      \node[valid] (s1) at (1.00,0.64) {$\checkmark$};
      \draw[link] (s1) -- (r1);
      \node[lab,anchor=west] at (1.38,0.64) {both chains};
    \end{tikzpicture}
    \caption{Revocation}
    \label{fig:revoke-view}
  \end{subfigure}\hfill
  \begin{subfigure}[t]{0.195\linewidth}
    \centering
    \begin{tikzpicture}[
      x=1cm,y=1cm,font=\scriptsize,
      progress/.style={draw,fill=white,text width=1.04cm,
        minimum height=0.56cm,align=center,inner sep=1pt},
      valid/.style={draw,fill=white,minimum size=0.32cm,inner sep=0pt},
      adv/.style={draw,fill=black,text=white,
        minimum size=0.36cm,inner sep=1pt},
      link/.style={-{Latex[length=1.2mm]},
        shorten >=1pt,shorten <=1pt},
      lab/.style={inner sep=1pt,align=center}
    ]
      \path[use as bounding box] (0,0) rectangle (2.25,5.00);
      \fill[black!7] (0.46,1.47) rectangle (2.19,4.22);
      \draw[link] (0.25,4.32) -- (0.25,1.18);
      \foreach \y in {4.22,1.47}
        \draw (0.19,\y) -- (0.31,\y);
      \node[lab] at (0.25,4.52) {$t$};
      \node[lab] at (0.38,1.00) {$t+\tau$};

      \node[valid] (v) at (1.39,4.67) {$\checkmark$};
      \node[adv] (b) at (1.39,3.92) {$B$};
      \node[progress] (h1) at (1.39,2.98) {honest\\progress};
      \node[progress] (h2) at (1.39,2.03) {honest\\progress};
      \draw[link] (b.north) -- (v.south);
      \draw[link] (h1.north) -- (b.south);
      \draw[link] (h2.north) -- (h1.south);
      \foreach \n in {h1,h2}
        \draw[black!50,line width=0.35pt]
          (\n.north west) -- (\n.south east)
          (\n.north east) -- (\n.south west);
      \node[lab] at (1.46,1.26) {reject $B$};
      \node[lab] at (1.12,0.43)
        {strict height\\$\ell\ \longrightarrow\ \ell$};
    \end{tikzpicture}
    \caption{Growth}
    \label{fig:revoke-growth}
  \end{subfigure}\hfill
  \begin{subfigure}[t]{0.425\linewidth}
    \centering
    \begin{tikzpicture}[
      x=1cm,y=1cm,font=\scriptsize,
      progress/.style={draw,fill=white,text width=1.04cm,
        minimum height=0.56cm,align=center,inner sep=1pt},
      adv/.style={draw,fill=black,text=white,
        minimum size=0.34cm,inner sep=1pt},
      root/.style={draw,circle,fill=white,
        minimum size=0.08cm,inner sep=0pt},
      link/.style={-{Latex[length=1.0mm,width=0.8mm]},
        shorten >=0.3pt,shorten <=0.3pt},
      lab/.style={inner sep=1pt,align=center}
    ]
      \path[use as bounding box] (0,0) rectangle (5.06,5.00);

      \node[lab] at (2.53,4.78) {valid private fork};
      \node[root] (r0) at (0.06,3.77) {};
      \foreach \x/\i in {0.70/1,1.94/2,3.18/3,4.42/4} {
        \node[progress] (p\i) at (\x,4.17) {honest\\progress};
        \node[adv] (a\i) at (\x,3.35) {$a_{\i}$};
        \draw[densely dotted,line width=0.4pt]
          (p\i.south) -- (a\i.north);
      }
      \draw[link] (p1.west) -- (r0);
      \draw[link] (a1.west) -- (r0);
      \foreach \i/\j in {2/1,3/2,4/3} {
        \draw[link] (p\i.west) -- (p\j.east);
        \draw[link] (a\i.west) -- (a\j.east);
      }

      \node[lab] at (2.53,2.73) {with one carrier};
      \node[root] (r1) at (0.06,1.00) {};
      \node[adv] (b) at (0.70,2.13) {$B$};
      \foreach \x/\i in {1.94/1,3.18/2,4.42/3} {
        \node[progress,text=black!45,draw=black!45] (q\i)
          at (\x,2.13) {honest\\progress};
        \draw[black!55,line width=0.35pt]
          (q\i.north west) -- (q\i.south east)
          (q\i.north east) -- (q\i.south west);
      }
      \draw[link] (b) -- (r1);
      \draw[link] (q1.west) -- (b.east);
      \draw[link] (q2.west) -- (q1.east);
      \draw[link] (q3.west) -- (q2.east);
      \node[lab] at (3.18,1.57) {revoked};

      \node[progress] (survives) at (1.02,1.00) {honest\\progress};
      \node[adv] (private) at (1.02,0.20) {$a_2$};
      \draw[link] (survives.west) -- (r1);
      \draw[link] (private.west) -- (r1);
      \draw[densely dotted]
        (survives.south) -- (private.north);
      \node[lab,anchor=west] at (1.78,1.00) {surviving};
      \node[lab,anchor=west] at (1.78,0.20) {valid private block};
    \end{tikzpicture}
    \caption{Competition}
    \label{fig:revoke-prefix}
  \end{subfigure}
  \caption{SPV mining revokes honest progress. Black blocks are adversarial.
  (a)~Rejecting carrier $B$ removes its descendants.
  (b)~The two revoked units leave strict-chain height unchanged in this execution.
  (c)~One carrier revokes three units. Only the surviving unit needs a valid competitor.}
  \label{fig:revoke}
\end{figure}

\subsection{Revocation under SPV mining}

Figure~\ref{fig:revoke-view} shows how an honest block's contribution remains revocable until the block and all its ancestors have been validated.
An invalid adversarial block $B$, called a \emph{carrier}, may enter the optimistic eligible set before its validation completes.
In the upper snapshot, the optimistic chain extends $B$ and its honest descendants, while the strict chain remains on a fully validated branch.
Once $B$ is found invalid, its entire subtree is excluded from the optimistic eligible set.
After rejection, both chain selections coincide with the surviving validated branch.
Thus, one adversarial mining success can revoke several honest blocks without constructing an equally long valid adversarial branch.

SPV mining does not change the rate of honest mining successes.
However, work on invalid ancestry does not advance or secure a fully validated chain.
Figure~\ref{fig:revoke-growth} and Figure~\ref{fig:revoke-prefix} illustrate this distinction.
In Figure~\ref{fig:revoke-growth}, the strict-chain height remains unchanged despite these honest successes.
Figure~\ref{fig:revoke-prefix} illustrates how carriers revoke part of the honest progress, reducing the valid competing work needed to match the surviving progress.
We use integer block counts for ease of illustration. These counts do not imply that the allocation is optimal or that every carrier revokes three units.

\section{Security Analysis}\label{sec: analysis}

\noindent\textbf{Roadmap.}
We establish chain growth, block liveness, mining-time persistence, and common prefix as defined in Section~\ref{sec: model}.
We introduce \emph{effective honest mining time} to measure mining on objectively valid ancestry as an equivalent duration at the full honest rate.
This quantity lets us account for carrier-induced losses when bounding honest progress (Section~\ref{sec: effective-time}).

We must also account for adversarial blocks mined before the interval under consideration.
As noted in prior analyses of premining~\cite{dembo2020everything}, blocks mined before a common honest block cannot extend it.
Under SPV mining, however, an earlier carrier can still divert honest mining if its branch remains eligible and at least as long as a miner's current chain.
We therefore track the adversary's \emph{stock}, the usable unpublished blocks available at the start of an interval.

We first relate effective mining time to valid-chain height and use this relation to bound stock uniformly over the observation horizon $[0,T]$.
These bounds yield chain growth by charging every new adversarial success as a potential carrier and accounting for losses caused by the initial stock (Section~\ref{sec: stock-growth}).
We then analyze the race under arbitrary allocations of adversarial mining power between carriers and valid competing branches (Section~\ref{sec: shared-budget}).
Using the stock bound to control the losses from initial stock, we account for the further losses caused by new work allocated to carriers and compare the remaining honest progress with valid competing adversarial work (Section~\ref{sec: analysis-permanence}).
This comparison supports a construction of permanent honest blocks, which gives block liveness and, in turn, mining-time persistence.
Finally, we derive common prefix by converting the confirmation duration into a block-depth bound.

All four guarantees hold with high probability against every admissible strategy from genesis whenever $\lambda_a<\lambda_h/(1+\lambda_h\Delta)$ and $\lambda_a\tau<1$.
The protocol's resilience threshold, expressed as an adversarial mining share, is therefore at least $\min\{2/(2+\lambda\Delta+\sqrt{4+(\lambda\Delta)^2}),\,1/(\lambda\tau)\}$.
All four guarantees hold for every $\beta$ strictly below this bound, with $1/(\lambda\tau)$ interpreted as $+\infty$ when $\tau=0$.
We present the main lemmas and theorems with proof sketches and give the full proofs in Appendix~\ref{sec: proof-analysis}.

\subsection{Effective time and interval loss}\label{sec: effective-time}
Let $a(t)\in[0,1]$ be the fraction of honest mining power working on objectively valid ancestry immediately before time $t$.
More explicitly, $a(t)$ is the sum of $w_i$ over miners whose selected parent and all its ancestors are objectively valid.
We define the \emph{effective honest mining time} by
$A(t)=\int_0^t a(s)\,ds$.
Thus, $A(t)-A(u)$ is the equivalent duration of mining on valid ancestry using all honest mining power.
It has units of time, and $\lambda_h(A(t)-A(u))$ is the corresponding integrated intensity of honest successes on valid ancestry.
This random integrated intensity need not be independent of the corresponding success process.

Let $H(t)$ be the maximum height of an objectively valid chain known to at least one honest miner at time $t$.
Let $Z(u,t)$ count adversarial successes in $(u,t]$.
Of these, $Z_{\mathrm{bad}}(u,t)$ counts blocks with an invalid ancestor, including the block itself, and
$Z_{\mathrm{val}}(u,t)=Z(u,t)-Z_{\mathrm{bad}}(u,t)$ counts the remainder.
Only the latter can belong to an objectively valid competing chain.
The classified counts need not be independent Poisson processes, since the adversary chooses how to use each success.

Let $S(u)$ count adversarial blocks mined by time $u$ that have not been received by any honest miner and have height at least $H(u)$.
A block is counted only if no objectively invalid block on its ancestry, including the block itself, has been received by any honest miner by $u$.
Here ``invalid'' refers to objective validity, whether or not validation has completed.
This stock includes usable blocks on invalid branches as well as valid private branches.
It is an analytical quantity, not a state that an honest miner can observe.

The next lemma bounds the effective mining time remaining in an interval under any allocation of adversarial power.
It separates the loss caused by newly mined blocks on invalid ancestry from the loss caused by the usable stock $S(u)$ present at the interval's start.

\begin{lemma}[Interval loss with explicit stock]\label{lem:interval-charging}
Every admissible execution satisfies, for all $0\leq u\leq t$,
\[
 A(t)-A(u)\geq(t-u)-\tau Z_{\mathrm{bad}}(u,t)-\tau S(u)-(2\Delta+\tau).
\]
From genesis, the sharper inequality $A(t)\geq t-\tau Z_{\mathrm{bad}}(0,t)$ holds.
\end{lemma}

\noindent\emph{Proof sketch.}
At each miner, a first invalid ancestor can attract mining for at most $\tau$ time units after receipt.
Summing with mining-power weights gives at most $\tau$ loss per invalid block.
For an arbitrary interval, set aside its first $2\Delta+\tau$ time units for propagation and validation.
Thereafter, each remaining bad mining instant can be charged to an adversarial block at height at least $H(u)$, either in the initial stock or mined after $u$ on invalid ancestry.
Each charged block again accounts for at most $\tau$ weighted loss, giving the interval bound.

\subsection{Uniform stock and chain growth}\label{sec: stock-growth}
Lemma~\ref{lem:interval-charging} quantifies the effective time remaining after losses from the initial stock $S(u)$ and newly mined invalid work $Z_{\mathrm{bad}}$. To compare the resulting honest progress with private-chain growth from $Z_{\mathrm{val}}$, we need a uniform bound on $S(u)$. Deriving this bound exploits a tradeoff. Honest progress makes unpublished blocks below the public valid height $H$ unusable, whereas suppressing that progress with carriers consumes the adversarial stock. To quantify how quickly retained stock becomes unusable, we first establish how effective mining time translates into growth of $H$.

Let $L_i(t)=|\mathcal{C}_i^{\mathrm{str}}(t)|$ and set
$\rho=\frac{\lambda_h}{1+\lambda_h\Delta},\qquad K=\kappa+\log(2+\lambda T)$.
The next lemma bounds the public valid height in effective time.

\begin{lemma}[Effective-time height growth]\label{thm:effective-height-growth}
For every $0<r<\rho$, there is a constant $C_H>0$ depending only on the model parameters and $r$ such that, for every $T\geq0$ and $\kappa\geq1$, under every admissible strategy,
\[
 H(t)-H(u)\geq r(A(t)-A(u))-C_HK
\]
holds simultaneously for all $0\leq u\leq t\leq T$ with probability at least $1-e^{-\kappa}$.
\end{lemma}

Combining this height bound with the accounting of consumed and surviving blocks yields the following uniform bound on $S(u)$.

\begin{theorem}[Uniform stock and effective time]\label{thm:effective-time-intervals}
Suppose $\lambda_a<\lambda_h/(1+\lambda_h\Delta)$ and $\lambda_a\tau<1$.
There are constants $C_S,C_A>0$, depending only on the model parameters, such that, for every $T\geq0$ and $\kappa\geq1$, under every admissible strategy from genesis, with probability at least $1-e^{-\kappa}$,
\begin{equation}
 \begin{split}
 S(u)&\leq C_SK,\\
 A(t)-A(u)&\geq(t-u)-\tau Z_{\mathrm{bad}}(u,t)-C_AK
 \end{split}
 \label{eq:effective-time-interval-stock}
\end{equation}
simultaneously for all $0\leq u\leq t\leq T$, where $K=\kappa+\log(2+\lambda T)$.
\end{theorem}

\noindent\emph{Proof sketch.}
An old unpublished block in the stock can remain usable only if its hidden path reaches the new public valid height. The path blocks above $H(u)$ belong to the initial stock and cannot also be consumed by carriers, since that would expose invalid ancestry on the surviving path. Hence the initial stock must cover both height growth and old carrier consumption. Once their sum exceeds $S(u)$, only newly mined blocks can remain usable.

Combining this disjoint accounting with the height and production bounds gives a strict contraction under $\lambda_a<\rho$ and $\lambda_a\tau<1$, after deducting newly consumed carriers from replacement stock. Iterating from empty stock at genesis, with concentration bounds controlling fluctuations and production between comparisons, yields $S(u)\leq C_SK$ uniformly over $[0,T]$. Substituting this bound into Lemma~\ref{lem:interval-charging} gives the effective-time inequality.


\begin{theorem}[Chain growth]\label{thm:spv-backbone-chain-growth}
Suppose $\lambda_a<\lambda_h/(1+\lambda_h\Delta)$ and $\lambda_a\tau<1$.
For every $0<\gamma<\rho(1-\lambda_a\tau)$, there is a constant $C_{\mathrm{cg}}>0$, depending only on the model parameters and $\gamma$, such that, for every $T\geq0$ and $\kappa\geq1$, under every admissible strategy from genesis, with probability at least $1-e^{-\kappa}$,
\[
  L_i(t)-L_i(u)\geq\gamma(t-u)
\]
holds simultaneously for all honest miners $i$ and all $0\leq u\leq t\leq T$ satisfying $t-u\geq C_{\mathrm{cg}}K$.
\end{theorem}

\noindent\emph{Proof sketch.}
Choose $0<r<\rho$ and $\varepsilon>0$ such that
$\gamma':=r(1-\lambda_a\tau-\varepsilon)>\gamma$.
Apply the stock, production, and height bounds with confidence parameter $\kappa+\log3$, so they hold jointly with probability at least $1-e^{-\kappa}$.
Since $Z_{\mathrm{bad}}(u,t)\leq Z(u,t)$, Theorem~\ref{thm:effective-time-intervals} and the uniform production bound give
\begin{equation}
  A(t)-A(u)\geq
  (1-\lambda_a\tau-\varepsilon)(t-u)-C_\varepsilon K.
  \label{eq:effective-time-interval-rate}
\end{equation}
Every objectively valid chain is fully validated by all honest miners within $\Delta+\tau$ time units of first becoming known to an honest miner.
Accounting for this fixed delay in the height bound yields
$L_i(t)-L_i(u)\geq\gamma'(t-u)-DK$
for a constant $D>0$, uniformly over all honest miners and intervals.
Taking $C_{\mathrm{cg}}\geq D/(\gamma'-\gamma)$ absorbs the additive deficit whenever $t-u\geq C_{\mathrm{cg}}K$, proving the claim.

\subsection{Shared adversarial budget}\label{sec: shared-budget}
The chain-growth bound above charges every adversarial success as a potential carrier.
We now analyze arbitrary allocations between carriers and valid competing branches, comparing the remaining honest progress with valid competing work under the same adversarial mining budget.

\begin{proposition}[Mixed-budget accounting]\label{prop:mixed-budget}
Let $r>0$. Suppose an interval of length $\ell=t-u$ has an honest progress quantity $P(u,t)$ satisfying
$P(u,t)\geq r(A(t)-A(u))-E(u,t)$.
Then the pathwise loss bound of Lemma~\ref{lem:interval-charging} implies
\begin{align*}
 P(u,t)-Z_{\mathrm{val}}(u,t)
 &\geq r\ell-r\tau Z_{\mathrm{bad}}(u,t)-Z_{\mathrm{val}}(u,t)-F(u,t)\\
 &\geq r\ell-\max\{1,r\tau\}Z(u,t)-F(u,t),
\end{align*}
where $F(u,t)=E(u,t)+r\tau S(u)+r(2\Delta+\tau)$.
The coefficient of $\ell$ after replacing $Z(u,t)$ by its mean $\lambda_a\ell$ is positive exactly when
$\lambda_a<r\quad\text{and}\quad\lambda_a\tau<1$.

\end{proposition}

\begin{proof}
Substitute the loss inequality into the assumed progress bound and use
$Z=Z_{\mathrm{bad}}+Z_{\mathrm{val}}$.
Each classified success has coefficient at most $\max\{1,r\tau\}$.
Finally, $r>\lambda_a\max\{1,r\tau\}$ is equivalent to the two stated inequalities.
\end{proof}

\subsection{Permanent blocks and security guarantees}\label{sec: analysis-permanence}
Next, we establish the remaining properties stated in Section~\ref{sec: model} using the stock bound and mixed-budget accounting established above.
Our construction follows the permanent-block approach of Dembo et al.~\cite{dembo2020everything}.
In our SPV setting, protecting a block against all valid competing branches also requires control of invalid ancestry and usable unpublished stock.
We call a block \emph{permanent from $v$} if it belongs to every honest strict chain at every time $t'\geq v$.

\begin{theorem}[Mining-time persistence]\label{thm:spv-race-persistence}
Suppose $\lambda_a<\lambda_h/(1+\lambda_h\Delta)$ and $\lambda_a\tau<1$.
Call a non-genesis block $B\in \mathcal{C}_i^{\mathrm{str}}(t)$ $\sigma$-confirmed if the tip of $\mathcal{C}_i^{\mathrm{str}}(t)$ was mined at least $\sigma$ time units after $B$.
There is a constant $C_{\mathrm{pers}}>0$, depending only on $(\beta,\lambda,\Delta,\tau)$, such that the following holds under every admissible strategy from genesis. For every $T\geq0$ and $\kappa\geq1$, set
$K:=\kappa+\log(2+\lambda T)$.
For every $\sigma\geq C_{\mathrm{pers}}K$, with probability at least $1-e^{-\kappa}$, the following holds simultaneously for all honest miners $i$ and all $t\in[0,T]$. If $B$ is $\sigma$-confirmed in $\mathcal{C}_i^{\mathrm{str}}(t)$, then the prefix ending at $B$ is a prefix of $\mathcal{C}_j^{\mathrm{str}}(t')$ for every honest miner $j$ and every $t'\geq t$.
\end{theorem}

\begin{theorem}[Block liveness]\label{thm:spv-race-liveness}
Suppose $\lambda_a<\lambda_h/(1+\lambda_h\Delta)$ and $\lambda_a\tau<1$.
There is a constant $C_{\mathrm{live}}>0$, depending only on $(\beta,\lambda,\Delta,\tau)$, such that the following holds under every admissible strategy from genesis.
For every $T\geq0$ and $\kappa\geq1$, set $K:=\kappa+\log(2+\lambda T)$.
For every $u\in[C_{\mathrm{live}}K,T]$, with probability at least $1-e^{-\kappa}$, every interval $(s,s+u]$, where $0\leq s\leq T-u$, contains the mining time of an honest block $B$ such that
\[
  B\in \mathcal{C}_i^{\mathrm{str}}(t)
  \qquad\text{for every honest miner $i$ and every $t\geq s+u$}.
\]
\end{theorem}

\noindent\emph{Proof sketch for Theorems~\ref{thm:spv-race-persistence} and~\ref{thm:spv-race-liveness}.}  Both proofs rely on a bound on the waiting time for a new permanent honest block.
We first use the stock contraction underlying Theorem~\ref{thm:effective-time-intervals} to bring the stock below a fixed bound depending only on the model parameters.
A fixed pattern of sufficiently separated honest successes, with no new adversarial successes, then has positive probability of exhausting or outgrowing this stock and giving a new honest block a fixed lead over every valid branch omitting it, including hidden branches.

While this lead holds in every honest miner's view, honest mining on valid ancestry extends the new block, so valid competing branches gain new blocks only through adversarial mining.
The shared-budget comparison gives honest progress a positive rate advantage after accounting for both carrier losses and valid competing work.
Concentration bounds therefore imply that a sufficiently large lead has a fixed positive probability of surviving forever.
We restart whenever the comparison ceases to certify the lead.
Bounds on the durations of failed attempts and resets give an exponential waiting-time bound in terms of the starting stock.
Combining this with uniform stock control and a union bound over a grid of windows shows that every sufficiently long interval contains a new honest block that is permanent by its end, proving block liveness.

For persistence, consider a $\sigma$-confirmed block $B$ on a strict chain whose tip was mined at $w$.
Liveness supplies an honest block $N$ mined in $(w-\sigma,w]$ and permanent by $w$.
Since the chain is observed at or after $w$, it contains both $B$ and $N$.
Moreover, $N$ was mined after $B$, so $B\preceq N$.
Every future strict chain contains $N$ and hence the entire prefix ending at $B$.

Next, we derive the common-prefix property from Theorem~\ref{thm:spv-race-persistence}.

\begin{theorem}[Common prefix]\label{thm:spv-backbone-common-prefix}
Suppose $\lambda_a<\lambda_h/(1+\lambda_h\Delta)$ and $\lambda_a\tau<1$.
There is a constant $C_{\mathrm{cp}}>0$, depending only on the model parameters, such that, for every $T\geq0$ and $\kappa\geq1$, under every admissible strategy from genesis, with probability at least $1-e^{-\kappa}$,
\[
  \bigl(\mathcal{C}_i^{\mathrm{str}}(t_1)\bigr)^{\lceil k}
  \preceq \mathcal{C}_j^{\mathrm{str}}(t_2)
\]
holds simultaneously for all honest miners $i,j$, all $0\leq t_1\leq t_2\leq T$, and every integer $k\geq k_0:=\lceil C_{\mathrm{cp}}K\rceil$, where $K=\kappa+\log(2+\lambda T)$.
\end{theorem}

\noindent\emph{Proof sketch.}
We convert the confirmation duration in persistence into a block depth by bounding how many blocks can be mined during that duration.
Apply persistence and the uniform total-count bound with confidence parameter $\kappa+\log2$, and take $\sigma=C_{\mathrm{pers}}(K+\log2)$.
Choosing $C_{\mathrm{cp}}$ sufficiently large ensures that, on their common event of probability at least $1-e^{-\kappa}$, every window $((w-\sigma)^+,w]$ with $w\leq T$ contains at most $k_0$ mining successes.

Now prune $k_0$ blocks from any observed strict chain.
If a non-genesis block $B$ remains at its tip, then $B$ and its $k_0$ descendants represent $k_0+1$ distinct mining successes.
They cannot all lie in the last $\sigma$ time units before the original tip was mined.
Since every descendant was mined after $B$, this forces $B$ to precede that tip by at least $\sigma$ time units.
Thus $B$ is $\sigma$-confirmed, and persistence preserves the retained prefix in every later strict chain.
Pruning additional blocks preserves the prefix relation, so the same event covers every $k\geq k_0$.

\noindent\textbf{Observation window $T$.} The observation cutoff $T$ limits the confirmations and creation windows covered simultaneously.
Retention of their permanent blocks extends to every future time.
This distinction is consistent with the per-transaction and finite-ledger guarantees discussed by Dembo et al.~\cite{dembo2020everything}.
Our proof obtains the simultaneous guarantee directly from the stock and waiting-time bounds.
Taking $T\to\infty$ does not yield a fixed confirmation window for an infinite ledger, because the required window grows with $K$.

\section{The Carrier Attack and Resilience Upper Bound}\label{sec: analysis-carrier}
We now give an admissible execution in which carriers eventually displace a block at any prescribed finite confirmation depth when $\lambda_a\tau>1$.
The result concerns eventual success strictly above this boundary.
It does not establish success at equality or a finite-horizon error exponent.

\subsection{Strategy and queue dynamics}
Start with a shared, fully validated public chain containing a non-genesis target block $b$ at the prescribed confirmation depth.
This state is attainable from genesis by delivering and validating valid blocks immediately and waiting for the required honest blocks.
Split the adversarial power into a carrier share $\beta_c$ and a private-fork share $\beta-\beta_c$, where
$\frac{1}{\lambda\tau}<\beta_c<\beta$.
The resulting independent mining rates are $c=\beta_c\lambda$ and $\mu=(\beta-\beta_c)\lambda>0$.
The private fork extends the parent of $b$ and therefore excludes $b$.
The adversary delivers released blocks immediately, schedules immediate validation of valid blocks, and schedules rejection of each carrier exactly $\tau$ after receipt at every honest miner.

When no carrier is available, honest miners extend the valid public tip and the adversary waits for a carrier mining success on that tip.
When this success occurs, the adversary releases the carrier to all honest miners.
Let $P$ be its valid parent.
During the next $\tau$ time units, all honest miners extend the carrier's invalid ancestry, while carrier miners produce and withhold further invalid siblings with parent $P$.
At rejection, every honest descendant of the active carrier becomes ineligible.
If another carrier is available, release it immediately and start another window.
Otherwise, honest miners resume work on valid ancestry until a new carrier is mined on the then-current public tip.

A \emph{busy period} is a consecutive sequence of windows with a nonempty carrier queue.
Let $Q_n$ count the active carrier and the unused siblings at the start of window $n$.
Until the queue reaches zero,
\begin{equation}
 Q_{n+1}=Q_n-1+N_n,\qquad
 N_n\overset{\mathrm{iid}}{\sim}\operatorname{Poisson}(c\tau),\qquad Q_0=1.
 \label{eq:carrier-stock-recursion}
\end{equation}
At each rejection boundary, if the private chain is strictly longer than the valid public chain, the adversary publishes it instead of releasing another carrier.
The invalid optimistic descendants have already disappeared, so they do not obstruct adoption of this longer valid chain.

\begin{theorem}[Eventual carrier success]\label{thm:carrier-success}
For every fixed $\beta\in[0,1)$ satisfying $\beta\lambda\tau>1$ and every prescribed finite initial confirmation depth, the strategy above removes $b$ from every honest strict chain with probability one, starting from the stated shared-chain setup.
\end{theorem}

\noindent\emph{Proof sketch.}
By Eq.~\ref{eq:carrier-stock-recursion}, each window consumes one carrier and adds $N_n$ new ones.
Viewing these new carriers as offspring gives a branching process with mean $c\tau>1$.
Thus the queue has a fixed positive probability of never emptying, allowing the busy period to continue indefinitely.

When a busy period ends, the strategy waits for a new carrier on the current valid public tip and restarts.
This waiting time is finite almost surely, and fresh Poisson increments give each restart the same positive survival probability.
The probability of successive failures therefore decreases geometrically, so almost surely a busy period begins that can sustain suppression indefinitely.

Throughout such a period, honest miners extend invalid ancestry, keeping the valid public height at $h(P)$.
The private fork grows at rate $\mu>0$ and therefore eventually exceeds this finite height.
At the next rejection boundary, the prescribed publication and validation make every honest miner adopt the strictly longer valid chain, which omits $b$.
The full proof appears in Appendix~\ref{app:carrier-success}.

\begin{figure}[t]
  \centering
  \captionsetup{font=small,labelfont=bf,labelsep=period,skip=5pt}
  \captionsetup[subfigure]{font=small,labelformat=parens,labelsep=space,skip=3pt}
  \begin{subfigure}[t]{0.230\linewidth}
    \centering
    \begin{tikzpicture}[x=1cm,y=1cm,
      blk/.style={draw,fill=white,minimum width=0.32cm,minimum height=0.32cm,inner sep=0pt},
      adv/.style={blk,fill=black},
      lab/.style={font=\scriptsize,inner sep=1pt},
      note/.style={lab,align=center}]
      \scriptsize
      \path[use as bounding box] (0,0) rectangle (2.76,4.40);

      \node[lab] at (1.38,4.10) {At crossover};
      \node[note] at (0.42,3.64) {one\\carrier};
      \node[note] at (1.96,3.64) {one expected\\unit};
      \node[adv] (CE) at (0.42,2.99) {};
      \draw[white,line width=0.55pt]
        (CE.north west) -- (CE.south east)
        (CE.north east) -- (CE.south west);
      \draw[fill=white] (1.46,2.84) rectangle (2.42,3.14);
      \draw (1.48,2.86) -- (2.40,3.12)
            (1.48,3.12) -- (2.40,2.86);
      \draw[-{Latex[length=2.4pt,width=3pt]},shorten <=1pt,shorten >=1pt]
        (CE.east) -- (1.46,2.99)
        node[lab,midway,above=2pt] {revokes};

      \draw[densely dotted,black!35] (0.08,2.56) -- (2.68,2.56);
      \node[lab] at (1.38,2.28) {Just past crossover};
      \node[note] at (0.42,1.82) {one\\carrier};
      \node[note] at (1.97,1.82) {one unit, plus\\a little more};
      \node[adv] (CB) at (0.42,1.11) {};
      \draw[white,line width=0.55pt]
        (CB.north west) -- (CB.south east)
        (CB.north east) -- (CB.south west);
      \draw[fill=white] (1.46,0.96) rectangle (2.60,1.26);
      \draw (2.42,0.96) -- (2.42,1.26);
      \draw (1.48,0.98) -- (2.40,1.24)
            (1.48,1.24) -- (2.40,0.98)
            (2.44,0.98) -- (2.58,1.24)
            (2.44,1.24) -- (2.58,0.98);
      \draw[-{Latex[length=2.4pt,width=3pt]},shorten <=1pt,shorten >=1pt]
        (CB.east) -- (1.46,1.11)
        node[lab,midway,above=2pt] {revokes};
      \draw[decorate,decoration={brace,mirror,amplitude=2pt}]
        (2.42,0.84) -- (2.60,0.84);
      \node[lab] at (2.49,0.58) {extra};
    \end{tikzpicture}\unskip\par
    \caption{Crossover}
    \label{fig:revoke-crossover}
    \label{fig:mixed}
  \end{subfigure}\hfill
  \begin{subfigure}[t]{0.745\linewidth}
    \centering
    \begin{tikzpicture}[x=1cm,y=1cm,
      blk/.style={draw,fill=white,minimum width=0.30cm,minimum height=0.30cm,inner sep=0pt},
      adv/.style={blk,fill=black,text=white},
      link/.style={-{Latex[length=2.1pt,width=2.8pt]},shorten >=0.7pt,shorten <=0.7pt},
      lab/.style={font=\scriptsize,inner sep=1pt},
      note/.style={lab,align=center},
      brace/.style={decorate,decoration={brace,mirror,amplitude=2pt}}]
      \scriptsize
      \path[use as bounding box] (0,0) rectangle (9.04,4.40);

      \node[lab,anchor=west] at (0.02,4.23) {\textbf{1} Confirm};
      \node[lab,anchor=west] at (1.89,4.23) {\textbf{2} Revoke \& private mine};
      \node[lab,anchor=west] at (6.80,4.23) {\textbf{3} Reverse};
      \draw[-{Latex[length=2.4pt,width=3pt]},black!60] (5.90,4.23) -- (6.48,4.23);
      \draw[densely dotted,black!40] (1.69,0.65) -- (1.69,3.98);
      \draw[densely dotted,black!40] (6.63,0.65) -- (6.63,3.98);

      \node[note] at (0.84,3.57) {honest strict\\chain};
      \node[blk] (Sg) at (0.16,3.04) {};
      \node[blk] (Sb) at (0.65,3.04) {$b$};
      \node[blk] (St) at (1.50,3.04) {};
      \draw[link] (Sb) -- (Sg);
      \draw[link] (St) -- (Sb);
      \node[lab,fill=white,inner sep=0.3pt] at (1.12,3.04) {$\cdots$};
      \draw[brace] (0.81,2.77) -- (1.66,2.77);
      \node[lab] at (1.25,2.48) {$k$};

      \node[lab] at (4.25,3.87)
        {$\frac{1}{\lambda\tau}<\beta=\frac{1}{\lambda\tau}+\epsilon
          <\beta_{\mathrm{pa}}(\lambda\Delta)$};
      \node[blk] (Mg) at (1.99,2.29) {};
      \node[blk] (Mb) at (2.52,3.04) {$b$};
      \node[blk] (Mt) at (3.43,3.04) {};
      \node[adv] (Mc) at (4.05,3.04) {};
      \node[blk] (Mh1) at (4.64,3.04) {};
      \node[blk] (Mh2) at (5.94,3.04) {};
      \draw[link] (Mb) -- (Mg);
      \draw[link] (Mt) -- (Mb);
      \draw[link] (Mc) -- (Mt);
      \draw[link] (Mh1) -- (Mc);
      \draw[link] (Mh2) -- (Mh1);
      \node[lab,fill=white,inner sep=0.3pt] at (2.99,3.04) {$\cdots$};
      \node[lab,fill=white,inner sep=0.3pt] at (5.34,3.04) {$\cdots$};
      \node[lab] at (2.76,3.49) {current chain};
      \node[lab] at (4.10,3.49) {carrier};
      \node[lab] (Mminers) at (5.65,3.49) {honest miners};
      \draw[-{Latex[length=2.1pt,width=2.8pt]},densely dotted,shorten >=1pt]
        (Mminers.south east) -- (Mh2.north);
      \draw[white,line width=0.55pt]
        (Mc.north west) -- (Mc.south east)
        (Mc.north east) -- (Mc.south west);
      \foreach \i in {1,2}
        \draw (Mh\i.north west) -- (Mh\i.south east)
              (Mh\i.north east) -- (Mh\i.south west);
      \draw[brace] (3.89,2.76) -- (6.10,2.76);
      \node[lab] at (5.00,2.51) {revoked};

      \node[lab] at (4.67,2.22) {carriers: $1/(\lambda\tau)+\epsilon/2$};
      \draw[-{Latex[length=2.4pt,width=3pt]},black!60] (2.94,1.86) -- (6.31,1.86);
      \node[adv] (Tc1) at (3.15,1.86) {};
      \node[adv] (Tc2) at (4.30,1.86) {};
      \foreach \i in {1,2}
        \draw[white,line width=0.55pt]
          (Tc\i.north west) -- (Tc\i.south east)
          (Tc\i.north east) -- (Tc\i.south west);
      \node[lab,fill=white] at (5.51,1.86) {$\cdots$};
      \node[lab,fill=white] at (6.01,1.86) {time};
      \draw[brace] (3.15,1.59) -- (4.30,1.59)
        node[lab,midway,below=3pt] {$\tau$};
      \draw[brace] (4.30,1.59) -- (5.45,1.59)
        node[lab,midway,below=3pt] {$\tau$};

      \node[adv] (Mp1) at (2.52,1.01) {};
      \node[adv] (Mp2) at (3.54,1.01) {};
      \draw[link] (Mp1) -- (Mg);
      \draw[link] (Mp2) -- (Mp1);
      \node[lab,fill=white,inner sep=0.3pt] at (3.07,1.01) {$\cdots$};
      \node[lab,anchor=west] at (3.90,1.01) {private: $\epsilon/2$};

      \node[lab] at (7.96,3.49) {old strict chain};
      \node[blk] (Rg) at (6.96,2.29) {};
      \node[blk] (Rb) at (7.30,3.04) {$b$};
      \node[blk] (Rt) at (8.27,3.04) {};
      \draw[link] (Rb) -- (Rg);
      \draw[link] (Rt) -- (Rb);
      \node[lab,fill=white,inner sep=0.3pt] at (7.82,3.04) {$\cdots$};
      \foreach \x/\j in {7.30/1,8.27/2,8.79/3}
        \node[adv] (Rp\j) at (\x,1.57) {};
      \draw[link] (Rp1) -- (Rg);
      \draw[link] (Rp2) -- (Rp1);
      \draw[link] (Rp3) -- (Rp2);
      \node[lab,fill=white,inner sep=0.3pt] at (7.82,1.57) {$\cdots$};
      \node[note] at (8.30,2.43) {eventually\\longer};
      \draw[-{Latex[length=1.8pt,width=2.5pt]},black!60,shorten >=0.5pt]
        (8.79,2.10) -- (Rp3.north);
      \node[note] at (7.99,0.98) {confirmed $b$\\reversed};

      \draw[densely dotted,black!35] (0.04,0.53) -- (8.99,0.53);
      \node[lab] at (4.52,0.23)
        {Growth rates during suppression: valid public $0$, private $\epsilon\lambda/2>0$};
    \end{tikzpicture}\unskip\par
    \caption{Attack}
    \label{fig:realistic-scenario}
    \label{fig:mixed-execution}
    \label{fig:mixed-asymptotic}
  \end{subfigure}
  \caption{Crossover and carrier attack. White denotes honest work and black denotes adversarial work.
  (a)~Reference-rate accounting at and just beyond $\beta_{\mathrm{pa}}(\lambda\Delta)\lambda\tau=1$.
  (b)~The carrier-dominant regime, with $\beta=1/(\lambda\tau)+\epsilon<\beta_{\mathrm{pa}}(\lambda\Delta)$.
  After an infinite busy period begins under continued suppression, the valid public height is constant while the private fork keeps growing. At a rejection boundary, publishing the longer private fork removes $b$.}
  \label{fig:carrier-attack}
\end{figure}
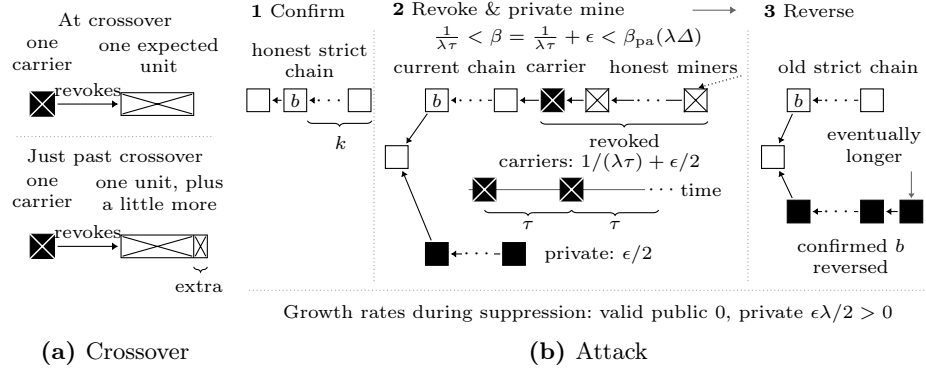

\subsection{Comparison with ordinary private mining}
Ordinary private mining remains an admissible special case when all blocks are valid and validation is immediate.
In the fully decentralized regime, its known threshold is~\cite{dembo2020everything}
$\beta_{\mathrm{pa}}(\lambda\Delta)
=\frac{2}{2+\lambda\Delta+\sqrt{4+(\lambda\Delta)^2}}$.
This expression concerns that miner-power regime, rather than every fixed distribution of honest power.
Together with Theorem~\ref{thm:carrier-success}, it gives the upper bound
$\min\{\beta_{\mathrm{pa}}(\lambda\Delta),1/(\lambda\tau)\}$
on asymptotic resilience for security guarantees covering both attacks.
For $\tau=0$, the carrier term is interpreted as $+\infty$.
The two boundaries coincide when $\beta_{\mathrm{pa}}(\lambda\Delta)\lambda\tau=1$.
At the reference progress rate $\beta_{\mathrm{pa}}(\lambda\Delta)\lambda$, a duration $\tau$ corresponds to one unit of progress at this crossover, and to more than one for larger $\tau$.
Figure~\ref{fig:carrier-attack} illustrates this rate comparison and the carrier-dominant regime.
The queue proof uses the strict inequality $\beta_c\lambda\tau>1$, rather than treating equality of average supply and consumption as sufficient for an infinite busy period.
The sufficient condition $\lambda_a<\rho$ is equivalent to
$\beta<\beta_{\mathrm{pa}}(\lambda\Delta)$, and $\lambda_a\tau<1$ is equivalent to $\beta<1/(\lambda\tau)$ when $\tau>0$.
The stock and permanence theorems therefore give security strictly below the displayed upper boundary.
The thresholds match in the fully decentralized regime. We do not settle behavior at equality.

\section{Related Work}\label{sec: related-work}
\label{sec: related-work-permissionless-consensus}
\label{sec: related-work-spv-mining}

\noindent\textbf{Consensus security.} Analyses of Nakamoto consensus establish security under synchrony~\cite{garay2015bitcoin} and bounded network delay~\cite{pass2017analysis,garay2024bitcoin,kiffer2018better,ren2019analysis}.
Later work gives tight asymptotic resilience bounds under the respective models~\cite{dembo2020everything,gavzi2020tight,blake2026rigorous}, matching the ordinary private-chain attack threshold in the fully decentralized regime~\cite{dembo2020everything}.
Finite-confirmation analyses quantify safety probabilities~\cite{li2021close,guo2022bitcoin,gavzi2023practical} and prove a bait-and-switch attack optimal in a particular bounded-delay model~\cite{cao2025beat}.

\noindent\textbf{Resources and participation.} Permissionless consensus models distinguish resource and participation assumptions~\cite{lewis2020resource,lewis2023permissionless,terner2022permissionless}.
Stake-based protocols~\cite{daian2019snow,kiayias2017ouroboros} determine eligibility, but eligible participants need not be online~\cite{lewis2023permissionless}.
Protocols for changing participation~\cite{pass2017sleepy,goyal2021instant,momose2022constant,malkhi2023towards,d2024goldfish,tang2026communication} study guarantees under varying availability.
Our model fixes aggregate mining rates to isolate delayed validation.

\noindent\textbf{Processing and validation.} Kiffer et al.~\cite{kiffer2024nakamoto} analyze bounded processing capacity, where queues and processing policies affect security and honest miners process a chain before extending it.
Tuxedo~\cite{das2021tuxedo} modifies the protocol to overlap processing and mining, allowing transaction validation to lag by a bounded number of blocks.
Our model permits extending a received block before validation completes and assumes a load-independent bound on validation latency.

\noindent\textbf{Mining before validation.} Cao et al.~\cite{cao2023leveraging} analyze a double-spending attack that divides adversarial mining between a distraction chain and a competing private chain.
Their miners may extend headers before receiving payloads, and the distraction relies on withholding those payloads.
Our model assumes that payloads have arrived and bounds the remaining validation time by $\tau$.
We analyze how rejecting invalid ancestry discards honest work.
Manifoldchain~\cite{che2024manifoldchain} uses predictive mining on multiple verified and unverified parents in a sharded protocol, whereas we retain a single-parent longest-chain rule under which rejecting an invalid ancestor excludes all its descendants from eligible chains.

\noindent\textbf{Incentives.} The verifier's dilemma~\cite{luu2015demystifying} and quantitative studies~\cite{alharby2020data,smuseva2022verifier,smuseva2025verifier,alharby2024quantitative,wadhwa2026blockchain} examine validation costs, mining policies, and rewards.
Blockchain denial-of-service (BDoS)~\cite{mirkin2020bdos} studies how withholding block data affects incentives to continue mining.
Our security analysis takes the honest mining policy as given and does not establish that it is an equilibrium.

\section{Conclusion}\label{sec: discussion}

We show that SPV mining, or mining before validation, can reduce the
adversarial mining power needed to break Nakamoto consensus below the
level predicted by network delay alone.
Specifically, when all honest miners follow this policy, an adversary that
produces more than one block on average during the allowed validation time
can repeatedly draw honest mining onto invalid branches.
Rejecting those branches erases honest work while the adversary's private
chain grows, eventually allowing it to replace the honest chain.
Under a validation-time bound independent of processing load, our security
guarantees and matching attacks give an asymptotically tight resilience
threshold as each honest miner's share of mining power becomes negligible.

We further establish security and attack bounds for fixed mixtures of
honest miners who extend only validated blocks and those who also extend
blocks awaiting validation.
These bounds depend on each policy's share of honest mining power.
We leave the tight resilience threshold for mixed policies and the effects
of validation queues and processing limits for future work.

\bibliographystyle{splncs04}
\bibliography{refs}

\appendix
\renewcommand{\theHsection}{\Alph{section}}
\section{Proofs of Theorems in Sections~\ref{sec: analysis} and~\ref{sec: analysis-carrier}}
\label{sec: proof-analysis}

\spnewtheorem*{restatedintervallemma}{Lemma~\noexpand\ref{lem:interval-charging}}{\bfseries\upshape}{\itshape}
\spnewtheorem*{restatedheightlemma}{Lemma~\noexpand\ref{thm:effective-height-growth}}{\bfseries\upshape}{\itshape}

We follow the roadmap in Section~\ref{sec: analysis}. We first prove Lemmas~\ref{lem:interval-charging} and~\ref{thm:effective-height-growth} from Section~\ref{sec: analysis}.
We then use their loss and height bounds to control unpublished stock and establish chain growth.
For the remaining properties, we return to a fixed stock bound and create an honest block with a lead over every valid branch omitting it.
The shared-budget comparison controls the survival of this lead and the cost of restarting after a failed comparison.
This gives a permanent honest block in every sufficiently long interval.
We then prove block liveness and mining-time persistence, and convert the latter into common prefix.
Finally, we prove the eventual carrier success claimed in Section~\ref{sec: analysis-carrier}.

Throughout the appendix, $\mathcal F_t$ denotes the complete execution history through $t$, including unpublished blocks, objective validity, deliveries, and completed validations.
The selected mining parents are predictable with respect to this history, and the independent Poisson mining processes retain their rates after every finite stopping time.
All constants are uniform over admissible strategies, the number of honest miners, and their fixed power shares.
Their permitted dependence on rates and confidence parameters is specified where they are introduced.

\noindent\textbf{Roadmap for the stock bound.}
We prove Theorem~\ref{thm:effective-time-intervals} in five steps. Fix a starting time $u$, calling blocks mined by $u$ old and later blocks new.

\noindent\emph{1. Account for consumption and height growth.}
The proof of Lemma~\ref{lem:interval-charging} sets aside a fixed boundary allowance and charges the remaining loss of effective time to distinct adversarial blocks that become public.
Each charge accounts for at most $\tau$ loss, and old and new consumption are counted separately.
Lemma~\ref{thm:effective-height-growth} converts the remaining effective time into growth of $H$. Its proof also supplies the conditional estimate needed at finite stopping times.
Lemma~\ref{lem:strict-height-bridge} transfers this height to strict chains for the later chain-growth argument.

\noindent\emph{2. Determine when old stock becomes unusable.}
Lemma~\ref{lem:effective-stock-clearance} quantifies the tradeoff between keeping old blocks usable and spending them on carriers.
An old block usable at a later time $v$ needs at least $H(v)-H(u)$ old unpublished blocks on its path above $H(u)$.
None can be an old consumed block, as publishing such a block would expose invalid ancestry on the surviving path.
The initial stock must therefore cover both the height increase and old consumption.
If their sum exceeds $S(u)$, every block still usable must be new.

\noindent\emph{3. Bound the stock that replaces it.}
Lemma~\ref{lem:effective-poisson-stopping} bounds new production after such stopping times.
Lemma~\ref{lem:effective-stock-contraction} combines this with the preceding bounds to choose a later comparison time when no old stock remains.
Choose an effective-height rate $r$ with $\lambda_a<r<\rho$.
Waiting for height growth permits replacement at rate $\lambda_a$, below the effective-height rate $r$. Consuming old carriers instead adds at most $\tau$ per block to the elapsed-time bound, with $\lambda_a\tau<1$.
Since the two uses share the same initial stock, the larger replacement ratio $\max\{\lambda_a/r,\lambda_a\tau\}<1$ controls any mixture.
New carriers used to prolong the interval must also be deducted from new production: each adds at most $\tau$ to the time bound but removes one replacement block.
With small rate margins for concentration, these deductions leave at most a fixed fraction of $S(u)$, plus an additive fluctuation allowance.

\noindent\emph{4. Include stock peaks between comparisons.}
Lemma~\ref{lem:effective-stock-uniform} iterates this contraction from empty stock at genesis.
The accumulated allowances form a geometric sum, bounding stock at every comparison time.
Stock can increase only through new production. Lemma~\ref{lem:effective-stock-contraction} bounds both the elapsed time and production before each comparison, thereby covering intervening stock peaks.

\noindent\emph{5. Obtain one event covering $[0,T]$.}
To complete Lemma~\ref{lem:effective-stock-uniform}, a positive minimum duration between comparisons bounds their number before $T$ by a quantity linear in $T$.
Their conditional failure probabilities can then be summed, even though the starts depend on the execution.
Since each failure probability decreases exponentially with the allowance, choosing an allowance of order $K=\kappa+\log(2+\lambda T)$ gives $S(v)\leq C_SK$ uniformly over $[0,T]$ with the required probability.
Substituting this bound into Lemma~\ref{lem:interval-charging} yields the effective-time inequality and completes Theorem~\ref{thm:effective-time-intervals}.

\subsection{Interval loss with explicit stock}
\label{app:effective-time-accounting}

\begin{restatedintervallemma}[Interval loss with explicit stock]
Every admissible execution satisfies, for all $0\leq u\leq t$,
\[
 A(t)-A(u)\geq(t-u)-\tau Z_{\mathrm{bad}}(u,t)-\tau S(u)-(2\Delta+\tau).
\]
From genesis, the sharper inequality $A(t)\geq t-\tau Z_{\mathrm{bad}}(0,t)$ holds.
\end{restatedintervallemma}

\begin{proof}[Lemma~\ref{lem:interval-charging}]
Recall that $a(t)$ is the honest power share mining on valid ancestry and $A(t)=\int_0^t a(v)\,dv$.
Thus $1-a(t)$ is the share wasted on invalid ancestry, and $A$ is nondecreasing and $1$-Lipschitz.
To identify which earlier blocks can cause a later loss, write $\mathcal P(t)=\bigcup_{i\text{ honest}}\mathcal T_i(t)$ for the blocks known to at least one honest miner.
The height $H(t)$ is the greatest height of a block in $\mathcal P(t)$ whose ancestry, including itself, is objectively valid.
It is nondecreasing because blocks with valid ancestry remain known and eligible for optimistic fork choice.

The stock in Section~\ref{sec: analysis} consists of blocks hidden from every honest miner that can still reach this public valid height.
More precisely, write $S(u)=|\mathcal U(u)|$, where
\begin{equation}
\begin{split}
 &\mathcal U(u):=\{F\in\mathcal T(u)\setminus\mathcal P(u):F\text{ is adversarial},\\
 &\qquad h(F)\geq H(u),\quad
 \nexists B\preceq F\text{ with }B\in\mathcal P(u),\ \Val(B)=0\}.
\end{split}
\label{eq:effective-stock-definition}
\end{equation}
The height condition includes ties because tie-breaking is adversarial.
Once an old block fails one of these conditions, it cannot satisfy it later, as publication is permanent. $H$ is nondecreasing, and a public invalid ancestor remains public.
Consequently, stock can increase only through new adversarial production, giving the pathwise bound
\begin{equation}
 S(v)\leq S(u)+Z(u,v),\qquad v\geq u.
 \label{eq:effective-stock-increase}
\end{equation}

We group wasted honest work by the invalid block whose rejection excludes the subtree being mined. For a chain containing an invalid block, its \emph{first-invalid root} is the first block $B$ on the path from genesis
with $\Val(B)=0$. The unique-parent rule makes this root unique. Distinct first-invalid roots have disjoint descendant subtrees, so every unit of wasted work belongs to exactly one such group. The validation deadline bounds how long each miner can work in a group's subtree.

We first show that a single first-invalid root $B$ causes at most $\tau$ units of lost effective honest mining time.
Let $\mathcal I_i(B)$ be the set of times when miner $i$'s selected parent has $B$ as its first-invalid root. If the miner never receives $B$, ancestor closure makes $\mathcal I_i(B)$ empty. Otherwise, the miner can select this ancestry only after receipt. At $\nu_i(B)$, the model sets $B$ and its entire subtree to \INVALID, permanently removing them from $E_i^{\mathrm{opt}}$. Thus $\mathcal I_i(B)\subseteq[r_i(B),\nu_i(B))$ and its duration is at most $\tau$, including any switches away from this subtree and back. Weighting by the fixed honest shares gives
\[
  \sum_iw_i\int_0^\infty\mathbf 1\{v\in \mathcal I_i(B)\}\,dv
  \leq\sum_iw_i\tau=\tau.
\]

The single-root loss bound does not yet identify which stock blocks pay for that loss. To make this connection, we assign each relevant first-invalid root to an adversarial descendant that becomes public at height at least $H(u)$. We call this assignment a \emph{charge}. A charged block is already public and therefore cannot remain in unpublished stock.

Roots known to an honest miner by $u+\Delta$ are covered by a fixed boundary allowance. For each remaining root, charge the first adversarial descendant (possibly the root itself) at height at least $H(u)$ that becomes public, using a fixed deterministic rule for simultaneous publications. The charge occurs at that publication time, whether or not the subtree attracts honest mining. We use $R_0(u,v)$ and $R_1(u,v)$ to count charges made by $v$ whose blocks were mined by $u$ or after $u$, respectively. Thus, mining time determines whether a charge is old or new, and publication determines when it enters the count. We now show that these charged blocks are distinct and bound the resulting loss, including the boundary allowance.

Set $b_\partial:=2\Delta+\tau$. The counts $R_0(u,v)$ and $R_1(u,v)$ defined above are nondecreasing and adapted to $\mathcal F_v$. We will prove that, for all $v\geq u$,
\begin{equation}
\begin{split}
  R_0(u,v)&\leq S(u),\qquad R_1(u,v)\leq Z_{\mathrm{bad}}(u,v),\\
  A(v)-A(u)&\geq(v-u)-b_\partial
                    -\tau[R_0(u,v)+R_1(u,v)].
\end{split}
\label{eq:effective-root-charge-budget}
\end{equation}
Each unit counted by $R_0$ or $R_1$ is a distinct adversarial block published by $v$, mined respectively by $u$ or after $u$. In particular,
\begin{equation}
  A(v)-A(u)\geq(v-u)-b_\partial
                       -\tau S(u)-\tau Z_{\mathrm{bad}}(u,v).
  \label{eq:effective-time-existing-stock}
\end{equation}

First, consider roots known to an honest miner by $u+\Delta$. Every honest miner receives them by $u+2\Delta$ and rejects them by $u+2\Delta+\tau$. Their total weighted loss after $u$ is at most $b_\partial$, since $1-a\leq1$, irrespective of their number or overlap. 

For each remaining root, the charged publication occurs after $u+\Delta$. There are only finitely many block publications in any bounded interval, so the first qualifying publication, when one occurs, is well defined. It is determined by the current history. Distinct roots give distinct charged blocks because their subtrees are disjoint.

We show that every remaining root attracting honest mining by $v$ has been charged by then. Every honest miner knows a valid chain of height $H(u)$ after $u+\Delta$, so any selected tip has height at least $H(u)$. If a selected tip in the root's subtree is adversarial, it is already a qualifying public descendant. Otherwise, follow its path back to the first honestly mined descendant of the root. That honest block was mined after $u+\Delta$, because its miner already knew the root. Its selected parent was adversarial and had height at least $H(u)$. This parent was public before that success, so a qualifying publication also occurred in this case.

If the charged block was mined by $u$, it was still hidden then. Otherwise, ancestor closure would have made its root public by $u$. It has no public invalid ancestor at $u$, since such an ancestor would also expose its first-invalid root. Together with its height, these facts place it in $\mathcal U(u)$. Count this charge in $R_0$. If the block was mined after $u$, it contributes to $Z_{\mathrm{bad}}(u,v)$, and we count the charge in $R_1$. The first qualifying publication determines each charge, so both counts are adapted and nondecreasing. Every charged block is public when it is counted.

Every charged root accounts for at most $\tau$ units of loss by the single-root loss bound proved above. Adding these losses to the initial boundary cost proves \eqref{eq:effective-root-charge-budget}. The bounds on $R_0$ and $R_1$ then give \eqref{eq:effective-time-existing-stock}.

The charge counts give the interval inequality of the first main-text lemma.
Starting from genesis removes the boundary cost because every invalid root is created during the interval itself.

Substitute $R_0(u,t)\leq S(u)$ and $R_1(u,t)\leq Z_{\mathrm{bad}}(u,t)$ into \eqref{eq:effective-root-charge-budget}.
For the genesis bound, every first-invalid root encountered by time $t$ is an adversarial block mined in $(0,t]$ on invalid ancestry.
Distinct roots number at most $Z_{\mathrm{bad}}(0,t)$, and the same single-root bound assigns at most $\tau$ loss to each.
Hence $t-A(t)\leq\tau Z_{\mathrm{bad}}(0,t)$.
\qed
\end{proof}

\subsection{Height growth in effective time}
\label{app:effective-time-height}

\begin{restatedheightlemma}[Effective-time height growth]
For every $0<r<\rho$, there is a constant $C_H>0$ depending only on the model parameters and $r$ such that, for every $T\geq0$ and $\kappa\geq1$, under every admissible strategy,
\[
 H(t)-H(u)\geq r(A(t)-A(u))-C_HK
\]
holds simultaneously for all $0\leq u\leq t\leq T$ with probability at least $1-e^{-\kappa}$.
\end{restatedheightlemma}

\begin{proof}[Lemma~\ref{thm:effective-height-growth}]
The loss bound measures the honest work that remains on valid ancestry.
We next show how this work raises public valid height.
Successes separated by $\Delta$ give successive height increases, because every selected block reaches every miner before the next selection.
The stock proof will restart this comparison at times determined by the execution, so we first establish a conditional bound after an arbitrary finite stopping time.

Fix $0<r<\rho=\lambda_h/(1+\lambda_h\Delta)$.
There is $C_H>0$ such that, after every almost surely finite stopping time $u$ and for every $\mathcal F_u$-measurable $\eta\geq1$,
\begin{equation}
\begin{split}
 \Pr\bigl[\,&H(t)-H(u)\geq r[A(t)-A(u)]-C_H\eta\\
 &\text{for every }t\geq u\mid\mathcal F_u\bigr]\geq1-e^{-\eta}.
\end{split}
\label{eq:effective-height-input-conditional}
\end{equation}
One choice of coefficient is obtained by setting
\begin{equation}
 \xi:=\frac12\left(\frac{\lambda_h}{r}-1-\lambda_h\Delta\right),
 \qquad \chi:=\lambda_h(1-e^{-\xi}),
 \qquad C_H:=\frac{1+\chi\Delta}{\xi+\chi\Delta}.
 \label{eq:effective-height-coefficient}
\end{equation}

Starting after $u+\Delta$, select the first effective honest success.
After each selected success, wait $\Delta$ physical time units before selecting the next effective success.
Let $Q_u(t)$ count selected successes by $t$.
Before the first selection, a valid chain attaining $H(u)$ has reached every honest miner.
Since an effective success extends a longest optimistic chain with valid ancestry, its block has height at least $H(u)+1$.
Each selected block reaches every miner before the next selection, so induction gives
\begin{equation}
 H(t)-H(u)\geq Q_u(t),\qquad t\geq u.
 \label{eq:effective-selected-height}
\end{equation}

To account for effective time skipped during these waits, let $J_u(v)$ indicate whether selection is permitted immediately before $v$ and set $B_u(t):=\int_u^t J_u(v)a(v)\,dv$.
The initial wait and the waits after selected successes exclude at most $\Delta[Q_u(t)+1]$ effective time, hence
\begin{equation}
 B_u(t)\geq A(t)-A(u)-\Delta[Q_u(t)+1].
 \label{eq:effective-selected-clock}
\end{equation}
The selected count has predictable intensity $\lambda_hJ_u(v)a(v)$.
With $\xi,\chi$ from \eqref{eq:effective-height-coefficient}, the process
\[
 \mathcal M_u(t):=\exp\{-\xi Q_u(t)+\chi B_u(t)\}
\]
is therefore a nonnegative local martingale starting from one, and hence a supermartingale.
Indeed, its drift coefficient is $J_u(v)a(v)[\chi+\lambda_h(e^{-\xi}-1)]=0$.
The conditional maximal inequality gives $\Pr(\sup_{t\geq u}\mathcal M_u(t)>e^\eta\mid\mathcal F_u)\leq e^{-\eta}$.

On the complementary event, $\chi B_u(t)-\xi Q_u(t)\leq\eta$ for every $t\geq u$.
Together with \eqref{eq:effective-selected-clock}, this gives
\begin{equation}
 Q_u(t)\geq\frac{\chi}{\xi+\chi\Delta}[A(t)-A(u)]
             -\frac{\chi\Delta+\eta}{\xi+\chi\Delta}.
 \label{eq:effective-selected-rate}
\end{equation}
The choice $r<\rho$ makes $\xi>0$.
Since $e^\xi\geq1+\xi$, we have $\xi/(1-e^{-\xi})\leq1+\xi$ and consequently
\[
 \frac{\chi}{\xi+\chi\Delta}
 \geq\frac{\lambda_h}{1+\xi+\lambda_h\Delta}>r.
\]
The deficit in \eqref{eq:effective-selected-rate} is at most $C_H\eta$ because $\eta\geq1$.
Combining this estimate with \eqref{eq:effective-selected-height} proves the conditional bound.

We now make this estimate hold for all intervals in $[0,T]$.
A finite grid controls their starting points, while the conditional bound already covers every later endpoint.
Rounding a start to the next grid point loses at most one grid spacing of effective time.

Fix $0<r<\rho$, and let $C_0$ be a coefficient supplied by
the conditional bound \eqref{eq:effective-height-input-conditional}.
Use the deterministic grid
\[
 \mathcal G_T:=\{j/\lambda:0\leq j\leq\lfloor\lambda T\rfloor\}.
\]
For each $g\in\mathcal G_T$, the conditional bound with reliability parameter $K$
gives
\[
 H(v)-H(g)\geq r[A(v)-A(g)]-C_0K
 \qquad\text{for all }v\geq g,
\]
except with probability at most $e^{-K}$.
A union bound makes these estimates hold together with failure
probability at most
\[
 |\mathcal G_T|e^{-K}
 \leq(1+\lambda T)e^{-K}<e^{-\kappa}.
\]

Work on this event and fix any $0\leq u\leq t\leq T$.
Set $g:=\lceil\lambda u\rceil/\lambda$.
If $g\leq t$, then $g\in\mathcal G_T$.
Monotonicity of $H$ and the $1$-Lipschitz property of $A$ give
\begin{align*}
 H(t)-H(u)
 &\geq H(t)-H(g)\\
 &\geq r[A(t)-A(u)]-C_0K-r/\lambda.
\end{align*}
If $g>t$, then $A(t)-A(u)\leq t-u<1/\lambda$, and the same lower
bound follows from $H(t)-H(u)\geq0$.
Since $K\geq1$, taking $C_H:=C_0+r/\lambda$ proves the result
simultaneously for all interval endpoints.
\qed
\end{proof}

The main-text height lemma is now proved. To use it for chain growth and agreement, we must transfer public valid height to the fully validated chains selected by individual miners. Propagation and validation impose the fixed lag $D:=\Delta+\tau$.

\begin{lemma}[Transfer to individual strict heights]\label{lem:strict-height-bridge}
Let $L_i(t)=|\mathcal{C}_i^{\mathrm{str}}(t)|$ and $D=\Delta+\tau$.
For every execution,
\[
 L_i(u)\leq H(u)\leq\min_j L_j(u+D).
\]
Consequently, on the event of Lemma~\ref{thm:effective-height-growth},
\begin{equation}
 L_i(t)-L_i(u)\geq r(A(t)-A(u))-C_HK-rD
 \label{eq:strict-effective-growth}
\end{equation}
simultaneously for every honest miner and every $0\leq u\leq t\leq T$.
\end{lemma}

\begin{proof}
A strictly eligible chain is objectively valid, which gives the first inequality.
A valid chain known to one honest miner at $u$ reaches all honest miners by $u+\Delta$.
Every block on it completes validation by $u+\Delta+\tau$, under the load-independent validation assumption, so the entire chain is strictly eligible everywhere by $u+D$.
If $t-u\geq D$, apply Lemma~\ref{thm:effective-height-growth} between $u$ and $t-D$ and use
$A(t)-A(t-D)\leq D$.
If $t-u<D$, the claimed lower bound follows from the monotonicity of each strict height and $A(t)-A(u)\leq D$.
\end{proof}

\subsection{Uniform stock and effective time}
\label{app:effective-time-constant}

The interval and height lemmas describe how old carriers can delay honest progress.
To bound stock, we also need to account for old blocks that remain hidden.
A surviving hidden path must reach the new public valid height, while each consumed carrier has already left the stock.
These two requirements occupy disjoint parts of the initial stock.

The loss bound counts old blocks spent on attacks, but unspent old blocks may still remain hidden. We next determine when none of them can remain in stock. An old block, meaning one mined by $u$, has a fixed height. For it to remain in stock at $v$, its path must reach $H(v)$. We will show that the blocks on this path above $H(u)$ all belong to the initial stock. They must also be distinct from the old blocks counted by $R_0(u,v)$. Each charged block is public by $v$ and has invalid ancestry, so placing it on the path would expose an invalid ancestor of the
surviving block.

Thus, a surviving old block requires the initial stock to cover both the public height increase and the old blocks already consumed. Once their sum exceeds $S(u)$, only newly mined blocks can remain in stock. The following lemma also bounds this remaining stock. It subtracts the new blocks counted by $R_1(u,v)$ from total new adversarial production, since these blocks are already public and cannot remain in stock.

\begin{lemma}[Elimination of initial stock]
\label{lem:effective-stock-clearance}
Fix $u\leq v$ and use the charges in
the proof of Lemma~\ref{lem:interval-charging}. If
\begin{equation}
  H(v)-H(u)+R_0(u,v)>S(u),
  \label{eq:effective-old-stock-clearance-test}
\end{equation}
then no block mined by $u$ belongs to $\mathcal U(v)$, and
\begin{equation}
  S(v)\leq Z(u,v)-R_1(u,v).
  \label{eq:effective-old-stock-clearance}
\end{equation}
\end{lemma}
\begin{proof}
Suppose that a block $F$ mined by $u$ remains in $\mathcal U(v)$. It must already belong to $\mathcal U(u)$, because a previously mined block outside stock cannot enter stock later. Since parents are fixed when blocks are mined, every block on the path to $F$ was also mined by $u$. The path contains exactly $h(F)-H(u)$ blocks strictly above height $H(u)$. We first show that all these blocks belong to $\mathcal U(u)$.

Consider any such path block. If it were public at $u$ with valid ancestry, its height would contradict the definition of $H(u)$. If it were public at $u$ with invalid ancestry, ancestor closure would make an invalid ancestor of $F$ public. That ancestor would remain public at $v$, contradicting $F\in\mathcal U(v)$. The path block is therefore hidden at $u$ and hence adversarial, since honest blocks are public when mined. It also has no public invalid ancestor at $u$, since any such ancestor would be an exposed invalid ancestor of $F$ at $v$. Together with its height, these properties place the path block in $\mathcal U(u)$.

The charge construction in the proof of Lemma~\ref{lem:interval-charging} places the $R_0(u,v)$ distinct charged old blocks in $\mathcal U(u)$. Each block has a first-invalid root on its ancestry and is public by the time it is charged. If a charged block lay on the path to $F$, ancestor closure would make that root public by $v$, contradicting $F\in\mathcal U(v)$. The path blocks and the charged old blocks therefore form disjoint subsets
of the initial stock, giving
\[
  h(F)-H(u)+R_0(u,v)\leq S(u).
\]
Membership in $\mathcal U(v)$ also requires $h(F)\geq H(v)$. Substituting this lower bound contradicts
\eqref{eq:effective-old-stock-clearance-test}, so no old block remains in stock at $v$.

Every block in $\mathcal U(v)$ must therefore come from the $Z(u,v)$ adversarial blocks mined in $(u,v]$. Among these blocks, the $R_1(u,v)$ charged new blocks are distinct and already public by $v$. They cannot belong to $\mathcal U(v)$. Removing them from the possible sources of stock proves
\eqref{eq:effective-old-stock-clearance}.
\qed
\end{proof}

The clearance lemma identifies when all remaining stock must be newly mined.
We next choose a check time that reaches this condition before new production replaces too much old stock.
Assume first that $\lambda_a>0$.
To bound adversarial production uniformly over future times, we introduce an
auxiliary rate $\lambda_a^+$ above the actual rate $\lambda_a$, leaving
slack for upward fluctuations.
We choose it below the rate $r<\rho$ used in the effective-height lower bound
and require $\lambda_a^+\tau<1$, preserving both strict margins needed for
stock contraction.
Fix such rates and define the positive exponential parameter $z_a$, which
will determine the additive allowance in the production bound, by
\begin{equation}
 \lambda_a<\lambda_a^+<r<\rho,\qquad \lambda_a^+\tau<1,
 \qquad z_a:=\log(\lambda_a^+/\lambda_a)>0.
 \label{eq:explicit-rate-margin}
\end{equation}
Such choices are possible under $\lambda_a<\rho$ and $\lambda_a\tau<1$.
The strict margins allow concentration estimates for both height growth and adversarial production.
The next estimate holds conditionally after the stopping times used for stock checks.

\begin{lemma}[Adversarial production after a stopping time]
\label{lem:effective-poisson-stopping}
For every almost surely finite stopping time $u$ and every $\mathcal F_u$-measurable $\eta\geq1$,
conditionally on $\mathcal F_u$, with probability at least $1-e^{-\eta}$,
\begin{equation}
  Z(u,v)\leq\lambda_a^+(v-u)+\eta/z_a
  \quad\text{for all }v\geq u.
  \label{eq:explicit-poisson-cone}
\end{equation}
\end{lemma}
\begin{proof}
Conditionally on $\mathcal F_u$, future adversarial successes form a
Poisson process of rate $\lambda_a$.
To control excess production at all later times, we apply the exponential
moment bound to the excess count $Z(u,v)-\lambda_a^+(v-u)$.
The choice of $z_a$ gives
\[
  \lambda_a(e^{z_a}-1)=\lambda_a^+-\lambda_a<z_a\lambda_a^+.
\]
Consequently, the process
$\exp\{z_a[Z(u,v)-\lambda_a^+(v-u)]\}$ is a nonnegative
supermartingale starting from one at $u$.
Exceeding the count bound in \eqref{eq:explicit-poisson-cone} would make
this process exceed $e^\eta$.
Its conditional maximal inequality bounds the probability of such a
crossing at any later time by $e^{-\eta}$, proving the lemma.
\qed
\end{proof}

We now choose when to compare the stocks.
The clearance condition in Lemma~\ref{lem:effective-stock-clearance} requires
$H(v)-H(u)+R_0(u,v)>S(u)$.
It combines public height growth with the number of old blocks already
consumed. Each consumed old block contributes one unit to this sum.
The loss bound charges it at most $\tau$ units of effective time, reducing
the height lower bound by at most $r\tau$.
Its net coefficient in the combined lower bound is therefore $1-r\tau$.
If this coefficient is nonnegative, we may discard the old contribution.
If it is negative, we bound the number of old charges by $S(u)$.

We use $M$ to cover both cases when choosing the required progress.
It accounts for the greater of the height needed to overtake unspent
old stock and the height allowance for delay if old stock is consumed.
We also use $d$ for the positive margin in the delay condition and
$\rho_S$ for the coefficient that will bound replacement stock relative
to initial stock. These constants are
\begin{align}
  d&:=1-\lambda_a^+\tau,& M&:=\max\{1,r\tau\},\notag\\
  \rho_S&:=\frac{\lambda_a^+M}{r}
       =\max\{\lambda_a^+/r,\lambda_a^+\tau\}<1.
  \label{eq:explicit-contraction-parameters}
\end{align}
The two terms in $\rho_S$ correspond to replacement production while
honest height overtakes old stock and while consumed blocks delay that
growth. Both are strictly below one by our choice of $\lambda_a^+$.

The comparison must also allow for the boundary loss $b_\partial$ and
the height deficit $C_H\eta$.
We use $C_0\eta$ to cover their combined effect on height, choosing
$C_0\geq1$ to ensure a positive minimum wait.
The coefficient $C_1$ converts this extra wait and the production
allowance $\eta/z_a$ into an additive stock bound.
Define
\begin{equation}
  C_0:=r b_\partial+\max\{C_H,1\},\qquad
  C_1:=\frac{\lambda_a^+(C_0+1)}r+\frac{1}{z_a}.
  \label{eq:explicit-check-constants}
\end{equation}

Starting at $u$ with $S_0=S(u)$, we convert elapsed time into the height
contribution $r(v-u)$ and deduct $r\tau R_1(u,v)$ for newly consumed blocks.
We call the first time their difference reaches $M S_0+C_0\eta+1$ a
\emph{check time}.
The threshold covers initial stock, the worst contribution of old
charges, and the boundary and fluctuation allowances.
The extra unit makes the clearance inequality strict.
The next lemma proves that the chosen check eliminates old stock and
bounds the new stock remaining there whenever the estimates hold.
It also bounds the wait and the stock accumulated during it, which we
will need when repeating the comparison.

\begin{lemma}[Contraction of unpublished stock]
\label{lem:effective-stock-contraction}
Fix a finite stopping time $u$ and an $\mathcal F_u$-measurable $\eta\geq1$, and set $S_0:=S(u)$.
Let $\sigma$ be the first $v\geq u$ such that
\begin{equation}
  r(v-u)-r\tau R_1(u,v)\geq M S_0+C_0\eta+1.
  \label{eq:explicit-stock-check}
\end{equation}
The check time $\sigma$ is an almost surely finite stopping time and
satisfies $\sigma-u\geq\eta/r$.
Conditionally on $\mathcal F_u$, with probability at least $1-2e^{-\eta}$,
\begin{align}
  S(\sigma)&\leq\rho_S S_0+C_1\eta,
  \label{eq:explicit-stock-contraction}\\
  \sigma-u&\leq\frac{M S_0+(C_0+1+r\tau/z_a)\eta}{rd},
  \label{eq:explicit-check-duration}\\
  \sup_{u\leq v\leq\sigma}S(v)
       &\leq\left(1+\frac{\rho_S}{d}\right)S_0+\frac{C_1}d\eta.
  \label{eq:explicit-stock-between-checks}
\end{align}
\end{lemma}
\begin{proof}
We first verify that the check time is well defined, independently of
the events used for the quantitative bounds.
The counts defining $\sigma$ are adapted, so the threshold can be tested
using the complete history $\mathcal F_v$ at each time $v$.
Since $R_1(u,v)\leq Z(u,v)$, the expression on the left of
\eqref{eq:explicit-stock-check} is at least
$r(v-u)-r\tau Z(u,v)$.
The Poisson strong law and $\lambda_a\tau<1$ make this lower bound tend
to infinity almost surely. Thus, $\sigma$ is a finite stopping time.

The expression defining the check grows continuously at rate $r$ between
downward jumps. It therefore reaches the threshold without an upward
overshoot, and equality holds in \eqref{eq:explicit-stock-check} at
$\sigma$.
Also, $R_1\geq0$ bounds the expression above by $r(v-u)$.
The threshold is at least $C_0\eta\geq\eta$, so reaching it requires
$\sigma-u\geq\eta/r$.

We next show that the check eliminates the initial stock.
Intersect the height event in
\eqref{eq:effective-height-input-conditional} with the production event
in \eqref{eq:explicit-poisson-cone}.
Each has conditional failure probability at most $e^{-\eta}$ at $u$,
so their intersection has conditional probability at least
$1-2e^{-\eta}$.
Both estimates hold for every $v\geq u$ on this event, including the
random endpoint $\sigma$.

On this intersection, substitute the loss bound from
the proof of Lemma~\ref{lem:interval-charging} into the height bound and
add the old consumption $R_0(u,v)$.
The resulting lower bound on the clearance quantity is
\begin{align}
  &H(v)-H(u)+R_0(u,v)\notag\\
  &\quad\geq r(v-u)-r\tau R_1(u,v)
      +(1-r\tau)R_0(u,v)-r b_\partial-C_H\eta\notag\\
  &\quad\geq r(v-u)-r\tau R_1(u,v)
      -(M-1)S_0-C_0\eta.
  \label{eq:explicit-combined-progress}
\end{align}
For the second inequality, if $r\tau\leq1$, the old-charge term is
nonnegative and $M=1$.
If $r\tau>1$, its negative coefficient and $R_0(u,v)\leq S_0$ give the
lower bound $-(M-1)S_0$.
Finally, $\eta\geq1$ and the definition of $C_0$ give
$r b_\partial+C_H\eta\leq C_0\eta$.
At $\sigma$, the last line of \eqref{eq:explicit-combined-progress}
equals $S_0+1$.
The clearance condition is therefore satisfied, and
Lemma~\ref{lem:effective-stock-clearance} excludes every old block from
the remaining stock.

It remains to bound the new blocks that replace this stock.
The same clearance lemma gives
$S(\sigma)\leq Z(u,\sigma)-R_1(u,\sigma)$.
The subtraction is essential when accounting for the wait.
In the equality defining $\sigma$, each newly consumed block adds $\tau$
units to the waiting time.
At the enlarged rate, this permits $\lambda_a^+\tau$ additional blocks
in the production bound, but the consumed block itself cannot remain in stock.
Its net contribution is therefore $\lambda_a^+\tau-1=-d<0$.
Writing $R_1:=R_1(u,\sigma)$ for this calculation, the production bound
and the equality at the check give
\begin{align}
  S(\sigma)
  &\leq\lambda_a^+(\sigma-u)+\eta/z_a-R_1\notag\\
  &=\rho_S S_0-dR_1+
          \frac{\lambda_a^+}{r}(C_0\eta+1)+\eta/z_a\notag\\
  &\leq\rho_S S_0+C_1\eta.
  \label{eq:explicit-contraction-calculation}
\end{align}
The last inequality discards the nonpositive term $-dR_1$ and uses
$1\leq\eta$.
Since $\rho_S<1$, the dependence on initial stock contracts, apart from
the additive allowance $C_1\eta$.

To control stock before the check, we also need an upper bound on the wait.
On the same production event, $R_1(u,v)\leq Z(u,v)$ gives
$r(v-u)-r\tau R_1(u,v)\geq rd(v-u)-r\tau\eta/z_a$.
Because $d>0$, this lower bound increases with elapsed time.
The check must occur by the time it reaches $M S_0+C_0\eta+1$.
Solving for that time and replacing $1$ by $\eta$ yields
\eqref{eq:explicit-check-duration}.

Finally, stock can increase after $u$ only through newly mined
adversarial blocks, as stated in \eqref{eq:effective-stock-increase}.
For every $u\leq v\leq\sigma$, bound that production by
$\lambda_a^+(\sigma-u)+\eta/z_a$ and apply the upper bound on the wait.
We obtain
\[
  S(v)\leq\left(1+\frac{\rho_S}{d}\right)S_0+
  \left[\frac{\lambda_a^+(C_0+1+r\tau/z_a)}{rd}+\frac{1}{z_a}\right]\eta.
\]
The bracket equals $C_1/d$, since $d+\lambda_a^+\tau=1$.
Hence \eqref{eq:explicit-stock-between-checks} holds on the same event
as the contraction and duration bounds.
\qed
\end{proof}

One check contracts the initial stock up to an additive fluctuation allowance.
We repeat these checks from empty stock at genesis, using the same allowance at every start.
The contraction makes the accumulated allowances a geometric sum, and the bound between checks covers intervening stock peaks.
The minimum duration of a check bounds their number before $T$, so conditional failure probabilities can be summed over these random starts.

For $K=\kappa+\log(2+\lambda T)$, the following constants convert the check count and the accumulated allowance into a uniform stock bound.

\begin{lemma}[Uniform stock bound]
\label{lem:effective-stock-uniform}
For $T\geq0$ and $\kappa\geq1$, set $K:=\kappa+\log(2+\lambda T)$ and
\begin{equation}
  G:=1+\log\bigl(8\max\{1,r/\lambda\}\bigr),
  \qquad C_S:=G\frac{(1+d)C_1}{d(1-\rho_S)}.
  \label{eq:explicit-uniform-stock-constant}
\end{equation}
Under every admissible strategy from genesis,
\begin{equation}
  \Pr\!\left[\sup_{0\leq v\leq T}S(v)>C_S K\right]
       \leq\tfrac14e^{-\kappa}.
  \label{eq:explicit-uniform-stock}
\end{equation}
\end{lemma}
\begin{proof}
Use the fixed value $\eta:=\kappa+\log(8(2+rT))$ for every comparison.
The logarithmic term will pay for the number of comparisons before $T$.
Let $u_0=0$, and let $u_{j+1}$ be the check time from
Lemma~\ref{lem:effective-stock-contraction} starting at $u_j$.
These are finite stopping times with
$u_{j+1}-u_j\geq\eta/r$.
Thus, the check times cannot accumulate in a bounded interval, and their
successive intervals cover $[0,T]$.
Since $\eta\geq1$, at most $1+rT$ comparisons start by $T$.

For each $j$, let $E_j$ be the event that the height and production
bounds used in the comparison starting at $u_j$ both hold.
We first bound stock on executions where every $E_j$ with $u_j\leq T$
holds.
On such an execution, Lemma~\ref{lem:effective-stock-contraction} gives
$S(u_{j+1})\leq\rho_S S(u_j)+C_1\eta$ at each of these starts.
Genesis has $S(u_0)=0$, so iterating gives
$S(u_j)\leq C_1\eta\sum_{\ell=0}^{j-1}\rho_S^\ell
\leq C_1\eta/(1-\rho_S)$ at every check start up to $T$.

To include a time between checks, apply
\eqref{eq:explicit-stock-between-checks} to the comparison interval
containing it and substitute the bound on that interval's initial stock.
The estimate also covers the interval containing $T$, even if its next
check occurs after $T$.
Hence
\[
  \sup_{0\leq v\leq T}S(v)
  \leq\left[\frac{1+\rho_S/d}{1-\rho_S}+\frac1d\right]C_1\eta
  =\frac{(1+d)C_1}{d(1-\rho_S)}\eta.
\]
To express this bound on the theorem's scale, use
$2+rT\leq\max\{1,r/\lambda\}(2+\lambda T)$.
It gives $\eta\leq K+\log(8\max\{1,r/\lambda\})\leq GK$, since
$K\geq1$.
The definition of $C_S$ now yields the required stock bound on every
execution under consideration.

It remains to bound the probability that one of these comparisons fails.
At each actual start, the conditional estimates give
$\Pr(E_j^c\mid\mathcal F_{u_j})\leq2e^{-\eta}$.
The event $\{u_j\leq T\}$ is measurable at $u_j$, so we may condition
on the history there before summing over the starts.
The deterministic bound on their number gives
\[
  \Pr[\exists j:u_j\leq T,\ E_j^c]
  \leq\sum_j\mathbb E\!\left[
       \mathbf 1_{\{u_j\leq T\}}
       \Pr(E_j^c\mid\mathcal F_{u_j})\right]
  \leq2(1+rT)e^{-\eta}\leq\tfrac14e^{-\kappa}.
\]
The conditioning accounts for the dependence of each start on the past.
On the complementary event, the preceding stock bound holds throughout
$[0,T]$, proving the lemma.
\qed
\end{proof}

The stock bound is uniform over the interval start, and the loss accounting is pathwise.
We can therefore substitute the former into the latter on one event to obtain both conclusions of the first theorem.

\begin{proof}[Theorem~\ref{thm:effective-time-intervals}]
If $\lambda_a=0$, then $S(u)=0$ and $A(t)=t$ almost surely, so any positive $C_S,C_A$ suffice.
Otherwise, choose $r=(\lambda_a+\rho)/2$ and an enlarged adversarial rate as in \eqref{eq:explicit-rate-margin}.
The conditional estimate \eqref{eq:effective-height-input-conditional} supplies the height bound at this rate.
Use $C_S$ from Lemma~\ref{lem:effective-stock-uniform} and take $C_A\geq b_\partial+\tau C_S$.
On the stock event, Lemma~\ref{lem:interval-charging} gives, for every $0\leq u\leq t\leq T$,
\[
 A(t)-A(u)\geq(t-u)-\tau Z_{\mathrm{bad}}(u,t)-b_\partial-\tau C_SK
 \geq(t-u)-\tau Z_{\mathrm{bad}}(u,t)-C_AK.
\]
Here $K\geq1$ absorbs the fixed boundary cost.
The stock event fails with probability at most $e^{-\kappa}/4$, which is within the stated bound.
The construction also includes $\tau=0$, when $d=M=1$ and $\rho_S=\lambda_a^+/r<1$.
\qed
\end{proof}

\subsection{Chain growth on arbitrary intervals}
\label{app:effective-time-all-intervals}

We have bounded loss using the realized count $Z_{\mathrm{bad}}$.
To express chain growth as a rate in physical time, we bound this count by total adversarial production and control that production uniformly over interval endpoints.
The enlarged rate in the next lemma can be chosen independently of the rates used for stock contraction.

To cover all interval endpoints, we enclose each interval in one whose
endpoints lie on a finite time grid.
We bound production for all grid intervals together and account for the
extra time added at their boundaries.
The constant $C_Z$ below collects the production allowance and the cost
of extending the endpoints.

\begin{lemma}[Uniform bound on adversarial block counts]
\label{lem:effective-poisson-uniform}
Fix any $\lambda_a^+>\lambda_a>0$ and set
$z_a:=\log(\lambda_a^+/\lambda_a)$.
For $T\geq0$, $\kappa\geq1$, and $K:=\kappa+\log(2+\lambda T)$, let
\begin{equation}
  C_Z:=\frac{2\lambda_a^+}{\lambda}+\frac{4}{z_a}.
  \label{eq:explicit-uniform-count-constant}
\end{equation}
With probability at least $1-e^{-\kappa}/4$,
\begin{equation}
  Z(u,t)\leq\lambda_a^+(t-u)+C_Z K
  \quad\text{for all }0\leq u\leq t\leq T.
  \label{eq:explicit-uniform-adversarial-count}
\end{equation}
\end{lemma}
\begin{proof}
Use grid points $j/\lambda$ for $0\leq j\leq\lceil\lambda T\rceil$.
The spacing $1/\lambda$ keeps the number of points proportional to
$2+\lambda T$ and adds at most $2/\lambda$ to the length of an enclosed
interval.
The execution is defined beyond $T$ if the last grid point exceeds the
observation horizon.
There are at most $(2+\lambda T)^2$ ordered grid pairs.

For a fixed grid pair $p\leq q$, we bound production above rate
$\lambda_a^+$ using the Poisson exponential moment.
With $z_a=\log(\lambda_a^+/\lambda_a)$, it gives
\[
  \mathbb E\exp\{z_a[Z(p,q)-\lambda_a^+(q-p)]\}
  =\exp\{(q-p)[\lambda_a(e^{z_a}-1)-z_a\lambda_a^+]\}\leq1.
\]
Markov's inequality therefore bounds the probability that this count
exceeds $\lambda_a^+(q-p)+4K/z_a$ by $e^{-4K}$.
Summing over all grid pairs bounds the probability of any such failure
by $(2+\lambda T)^2e^{-4K}\leq e^{-\kappa}/4$.

On the event where all grid bounds hold, consider any
$[u,t]\subseteq[0,T]$.
Round $u$ down to a grid point $p$ and $t$ up to a grid point $q$.
The interval inclusion gives $Z(u,t)\leq Z(p,q)$, and the extra length
satisfies $q-p\leq t-u+2/\lambda$.
Applying the grid bound yields
$Z(u,t)\leq\lambda_a^+(t-u)+2\lambda_a^+/\lambda+4K/z_a$.
Since $K\geq1$, the definition of $C_Z$ absorbs the fixed rounding cost
and proves \eqref{eq:explicit-uniform-adversarial-count} for every
interval on the same event.
\qed
\end{proof}

The uniform count bound now gives the real-time rate in \eqref{eq:effective-time-interval-rate}.
The strict-height transfer then proves growth relative to any starting height, as required by the theorem.

\begin{proof}[Theorem~\ref{thm:spv-backbone-chain-growth}]
Fix $0<\gamma<\rho(1-\lambda_a\tau)$.
Choose $0<r<\rho$ and $\varepsilon\in(0,1-\lambda_a\tau)$ so that $\gamma':=r(1-\lambda_a\tau-\varepsilon)>\gamma$.
Apply the stock, count, and uniform height estimates with confidence parameter $\widehat\kappa:=\kappa+\log3$.
Writing $\widehat K:=K+\log3$, their joint failure probability is at most $e^{-\kappa}$.
All constants below are independent of $T$ and $\kappa$.

If $\lambda_a\tau>0$, use Lemma~\ref{lem:effective-poisson-uniform} at rate $\lambda_a^+=\lambda_a+\varepsilon/\tau$ and denote its allowance coefficient by $C_Z$.
On the joint event, Theorem~\ref{thm:effective-time-intervals} gives
\begin{align*}
 A(t)-A(u)
 &\geq(t-u)-\tau Z(u,t)-C_A\widehat K\\
 &\geq(1-\lambda_a\tau-\varepsilon)(t-u)
              -(C_A+\tau C_Z)\widehat K
\end{align*}
for all $0\leq u\leq t\leq T$.
If $\lambda_a=0$ or $\tau=0$, then $A(t)=t$ almost surely and the same conclusion holds with $C_Z=0$.
Since $\widehat K\leq(1+\log3)K$, this proves \eqref{eq:effective-time-interval-rate} with an adjusted coefficient.

By Lemma~\ref{lem:strict-height-bridge}, the same event gives
\[
 L_i(t)-L_i(u)\geq\gamma'(t-u)-C_*K
\]
simultaneously for all honest miners and interval endpoints, where one may take
\[
 C_*:=(1+\log3)\{r(C_A+\tau C_Z)+C_H\}+rD.
\]
Taking $C_{\mathrm{cg}}\geq C_*/(\gamma'-\gamma)$ yields
$L_i(t)-L_i(u)\geq\gamma(t-u)$ whenever $t-u\geq C_{\mathrm{cg}}K$.
\qed
\end{proof}

Chain growth is now established.
For agreement, honest progress must also stay ahead of valid competing branches.
Proposition~\ref{prop:mixed-budget}, proved in the main text, accounts for both carrier losses and valid competing work using the same total adversarial production.
The proof below applies that comparison to a particular honest block after giving it a lead over every valid branch omitting it.

The construction needs two fixed time allowances: $D=\Delta+\tau$ for a public valid chain to become strictly eligible everywhere, and $b:=b_\partial=2\Delta+\tau$ for the boundary loss in Lemma~\ref{lem:interval-charging}.
We continue to assume $\lambda_a<\rho$ and $\lambda_a\tau<1$.

\subsection{Return to bounded stock}\label{app:framework-stock-return}
We now derive the security properties from the stock bound while allowing arbitrary allocation of adversarial power.
The proof constructs permanent honest blocks within sufficiently long intervals, as outlined below.

\noindent\textbf{Roadmap.}
Starting from a finite stopping time $u$, we construct a new honest block that belongs to every strict chain from some later time onward.
The final step turns this construction into liveness, persistence, and common prefix.

\noindent\emph{1. Return to a fixed stock bound.}
We first reduce usable stock to a fixed bound $m$, so that a fixed mining pattern can succeed with probability bounded away from zero uniformly over starting histories.
Lemma~\ref{lem:framework-stock-return} obtains this return by repeating the stock contraction and controlling the time and additional stock even when a comparison fails.
It bounds the exponential moment of the return time in terms of the starting stock.

\noindent\emph{2. Give a new honest block a clean lead.}
For any prescribed integer $L\geq1$, Lemma~\ref{lem:framework-clean-reset} uses sufficiently separated honest successes with no new adversarial successes.
After an initial wait for propagation and rejection, each invalid root can spoil at most one selected success, consuming a distinct block from the stock present at the pattern's start.
Enough remaining successes therefore raise the common valid height above all old hidden tips.
After exposed invalid branches have been rejected, a further run of $L$ successes extends one chain and gives its first block $B$ a lead over every valid branch omitting it, including hidden branches.
Usable stock is then zero. A final interval of length $D$ without mining completes propagation and validation, giving the clean lead of Definition~\ref{def:framework-clean-lead}.
Returning to bounded stock after unsuccessful patterns and repeating gives the lemma's waiting-time bound.

\noindent\emph{3. Show that the lead can survive forever.}
Lemma~\ref{lem:framework-clean-race} uses the initial public valid height minus $L$, plus all new valid adversarial production, to bound every valid branch omitting $B$.
This bound applies until the first effective honest success that omits $B$. While the bound stays below the valid height available to every honest miner, such a success is impossible and every strict chain retains $B$.
Reaching the available height marks a failure of the comparison.
To maintain this inequality, the shared-budget comparison charges each new adversarial block either one level of competing height or at most $r\tau$ levels of lost honest progress, with $\lambda_a<r<\rho$.
Both charges use the same total production, so the two rate conditions leave a positive advantage for honest progress.
Concentration then gives a sufficiently large fixed lead a probability of at least $3/4$ of surviving forever.
For finite failures of the comparison, the lemma also bounds both the duration and the remaining stock with exponential tails.

\noindent\emph{4. Bound the waiting time for a permanent block.}
Lemma~\ref{lem:framework-permanent-wait} repeats the clean-lead construction after each finite failure of the comparison.
The stock left by a failure determines the next reset's cost.
Choosing its exponent $\alpha>0$ in the reset bound small enough lets the failure tail control this cost together with the elapsed time.
Together with Lemma~\ref{lem:framework-clean-reset}, it makes the exponential-moment contribution of each failed cycle less than one at a sufficiently small exponent.
Summing over failed cycles gives an exponential waiting-time bound with an initial-stock factor $e^{\alpha S(u)}$.
The first clean lead whose comparison never fails supplies the new permanent block. All retries use conditional estimates at finite reset and crossing times.

\noindent\emph{5. Derive the three properties.}
Lemma~\ref{lem:framework-permanent-windows} combines this waiting-time bound with uniform stock control and a finite grid of windows of length proportional to $K$.
With probability at least $1-e^{-\kappa}$, every sufficiently long interval in $[0,T]$ then contains a newly mined honest block permanent by its end.
This gives block liveness, Theorem~\ref{thm:spv-race-liveness}.
For persistence, use the confirmation window ending when the observed chain's tip was mined.
Its permanent block was mined after the confirmed block and already lies on the observed chain, so its ancestry preserves the confirmed prefix, proving Theorem~\ref{thm:spv-race-persistence}.
Finally, Lemma~\ref{lem:framework-total-count} bounds total mining successes during a confirmation window.
After sufficient pruning, any non-genesis retained tip is therefore confirmed. Persistence then gives common prefix in Theorem~\ref{thm:spv-backbone-common-prefix}.

We first show that the usable unpublished stock returns to a fixed bound.
This prepares the construction of a block that belongs to every honest strict chain from some time onward.

Fix the rates used in the conditional height and production estimates.
The lower honest rate $r$ must exceed the upper adversarial rate $\lambda_a^+$, and the latter must still satisfy the validation condition.
Choose
\begin{equation}
 \lambda_a<\lambda_a^+<r<\frac{\lambda_h}{1+\lambda_h\Delta},\qquad
 \lambda_a^+\tau<1.
 \label{eq:framework-race-rates}
\end{equation}
The two strict assumptions leave room for these choices. To charge a new block on invalid ancestry in units of the height lower bound, write $c:=r\tau$. Set $M:=\max\{1,c\}$ for the larger of this charge and the one-block
charge for valid competing work. 

We need the height and production estimates after starting times chosen from the execution history.
For a confidence parameter $\eta\geq1$, use $C_H\eta$ as the height deficit and $C_a\eta$ as the production allowance. Here $C_H$ is given by \eqref{eq:effective-height-coefficient}, and $C_a=1/\log(\lambda_a^+/\lambda_a)$ when $\lambda_a>0$. When $\lambda_a=0$, take $C_a=0$ and use $Z=0$. The conditional height bound \eqref{eq:effective-height-input-conditional} and Lemma~\ref{lem:effective-poisson-stopping} imply that, after every finite stopping time $u$, the following bounds hold conditionally on $\mathcal F_u$ with failure probability at most $2e^{-\eta}$,
\begin{equation}
\begin{split}
 H(t)-H(u)&\geq r[A(t)-A(u)]-C_H\eta,\\
 Z(u,t)&\leq\lambda_a^+(t-u)+C_a\eta
 \qquad(t\geq u).
\end{split}
\label{eq:framework-conditional-cones}
\end{equation}
Both estimates cover all later endpoints on the same event. The height bound already accounts for the propagation gaps between selected effective successes in its proof.

The finite-horizon stock estimate alone allows stock of order $K$. We need a return to a fixed stock size before attempting the prescribed mining pattern in the next lemma. The following exponential-moment bound controls the time required for this return. Its parameter $\alpha$ determines how strongly the initial stock enters the bound.
Allowing $\alpha$ to be arbitrarily small will let us combine the return cost with the random stock left after later failures.

\begin{lemma}[Return time to bounded stock]
\label{lem:framework-stock-return}
For every sufficiently small $\alpha>0$, there are an integer $m\geq1$ and $q_r>0$ with the following property. After every finite stopping time $u$, there is a finite stopping time $v\geq u$ at which $S(v)\leq m$ and
\begin{equation}
 \mathbb E[e^{q_r(v-u)}\mid\mathcal F_u]\leq e^{\alpha S(u)}.
 \label{eq:framework-return-moment}
\end{equation}
The constants are uniform over admissible strategies and histories.
\end{lemma}

\begin{proof}
When $\lambda_a=0$, the stock remains zero from genesis, so take $v=u$. Assume henceforth that $\lambda_a>0$.
We repeat the stock comparison at later times chosen from the current history, using a fluctuation allowance proportional to the starting stock. The allowance makes failure exponentially unlikely when the stock is large, while preserving contraction when the estimates hold.

Within this proof, let $\theta>0$ be an extra threshold allowance per initial stock block. Choose it small enough that
$\rho_{R,0}:=\lambda_a^+(M+\theta)/r<1$. Next choose $\epsilon>0$ small enough that $C_H\epsilon<\theta/2$ and $\rho_R:=\rho_{R,0}+C_a\epsilon<1$. The parameter $\epsilon$ sets the confidence level per unit of stock, while $\rho_R$ will bound the fraction left at a successful checkpoint. Recall that $R_0(u,t)$ and $R_1(u,t)$ count the old and new blocks published and consumed in the proof of Lemma~\ref{lem:interval-charging}. For initial stock $s:=S(u)$, define the checkpoint $\sigma$ as the first $t\geq u$ at which
\begin{equation}
 r(t-u)-cR_1(u,t)\geq(M+\theta)s.
 \label{eq:framework-return-check}
\end{equation}
The left side is determined by the current history and has only downward jumps. It reaches its threshold without an upward overshoot. Also, $R_1\leq Z$ and $\lambda_a\tau<1$ make it tend to infinity almost surely. Thus $\sigma$ is a finite stopping time, and equality holds at the checkpoint.

For sufficiently large $s$, use \eqref{eq:framework-conditional-cones} with $\eta=\epsilon s\geq1$.
Call the event where both estimates hold the good event. Combine the height estimate with the loss bound in terms of $R_0+R_1$ from the proof of Lemma~\ref{lem:interval-charging} to obtain
\begin{align*}
 H(\sigma)-H(u)+R_0(u,\sigma)
 &\geq r(\sigma-u)-cR_1(u,\sigma)
       +(1-c)R_0(u,\sigma)-rb-C_H\epsilon s\\
 &\geq (1+\theta-C_H\epsilon)s-rb>s.
\end{align*}
The old-charge term satisfies $(1-c)R_0\geq-(M-1)s$, since $R_0\leq s$.
The extra allowance $\theta s$ exceeds the height fluctuation and the fixed boundary cost once $s$ is large enough.
Lemma~\ref{lem:effective-stock-clearance} therefore clears all old stock. Only new adversarial blocks can remain, and the newly charged blocks must be subtracted because they are already public. Using the equality at the checkpoint, we obtain
\begin{align}
 S(\sigma)
 &\leq Z(u,\sigma)-R_1(u,\sigma)\notag\\
 &\leq\left[\frac{\lambda_a^+(M+\theta)}r+C_a\epsilon\right]s
       -(1-\lambda_a^+\tau)R_1(u,\sigma)
 \leq\rho_R s.
 \label{eq:framework-return-contraction}
\end{align}
The production estimate also bounds the time needed to reach the threshold. Substitute $R_1\leq Z$ into the checkpoint condition and solve the resulting linear lower bound for elapsed time. Writing $C_t$ for the resulting time coefficient gives
\begin{equation}
 \sigma-u\leq C_t s,\qquad
 C_t:=\frac{M+\theta+cC_a\epsilon}{r(1-\lambda_a^+\tau)}.
 \label{eq:framework-return-duration}
\end{equation}
The contraction and duration bounds thus hold except with conditional probability at most $2e^{-\epsilon s}$.

For the exponential moment, a rare failure cannot be allowed to have an uncontrolled cost.
We therefore bound the checkpoint duration and its final stock even outside the good event.
First suppose $\tau>0$ and write $Q:=(M+\theta)s$ for the threshold. Let $\zeta_Q$ be the first $t\geq0$ at which
$rt-cZ(u,u+t)\geq Q$. Replacing the new-charge count by all adversarial production can only delay the crossing, so $\sigma-u\leq\zeta_Q$. 

To bound this longer wait, take a sufficiently small parameter $\xi>0$. The associated time exponent and moment bound are
\begin{equation}
 \omega:=r\xi-\lambda_a(e^{c\xi}-1)>0,\qquad
 \mathbb E[e^{\omega\zeta_Q}\mid\mathcal F_u]\leq e^{\xi Q}.
 \label{eq:framework-busy-moment}
\end{equation}
Positivity holds for small $\xi$ because the derivative at zero is $r-\lambda_a c=r(1-\lambda_a\tau)>0$.
For the moment bound, stop the nonnegative Poisson martingale $\exp\{-\xi[rt-cZ(u,u+t)]+\omega t\}$ at $\zeta_Q\wedge n$. The positive drift makes $\zeta_Q$ finite, and the crossing has no upward overshoot. Fatou's lemma then gives the displayed bound. Stock can increase only through new production, so the same crossing
also gives
\[
 S(\sigma)\leq s+Z(u,u+\zeta_Q)\leq s+\zeta_Q/\tau.
\]

We now balance contraction on the good event against the cost of its complement, which we call the bad event.
Choose $\xi$ with $\xi(M+\theta)<\epsilon$.
Then choose $\alpha,q>0$ with
$2\alpha/\tau+2q\leq\omega$ and
$qC_t<\alpha(1-\rho_R)$.
Cauchy--Schwarz combines the bad-event probability
$2e^{-\epsilon s}$ with \eqref{eq:framework-busy-moment}.
Its contribution to $\mathbb E[e^{\alpha S(\sigma)+q(\sigma-u)}\mid\mathcal F_u]$
is at most
\[
 \sqrt{2}\exp\!\left\{
 \alpha s-\frac{\epsilon-\xi(M+\theta)}2s\right\}.
\]
The good-event contribution is at most $\exp\{(\alpha\rho_R+qC_t)s\}$.
Relative to $e^{\alpha s}$, both terms decrease exponentially in $s$ by the choices of $\xi$ and $q$.

When $\tau=0$, the checkpoint instead has deterministic duration $Q/r$. Use $S(\sigma)\leq s+Z(u,u+Q/r)$ and the Poisson exponential moment in the same Cauchy--Schwarz argument. The bad-event contribution is at most
\[
 \sqrt{2}\exp\!\left\{\alpha s-
 \frac12\left[\epsilon-\frac{M+\theta}{r}
   \{\lambda_a(e^{2\alpha}-1)+2q\}\right]s\right\}.
\]
Taking $\alpha,q$ sufficiently small makes the square bracket positive, and the same good-event estimate applies.
In either case, $\alpha$ may be chosen arbitrarily small, followed by a sufficiently small $q$.

Choose an integer $m\geq1$ large enough that the preceding estimates apply and the two relative factors sum to at most one for all $s>m$. Set $q_r=q$.
We have proved the conditional drift inequality
\begin{equation}
 \mathbb E[e^{\alpha S(\sigma)+q_r(\sigma-u)}\mid\mathcal F_u]
 \leq e^{\alpha S(u)}\qquad(S(u)>m).
 \label{eq:framework-return-drift}
\end{equation}
To turn it into a return-time bound, let $u_0=u$ and repeat the checkpoint construction whenever $S(u_j)>m$.
After the first checkpoint with stock at most $m$, keep the checkpoint times fixed. If the initial stock is already at most $m$, take $v=u$.

Along this stopped sequence, $e^{\alpha S(u_j)+q_r(u_j-u)}$ is a nonnegative supermartingale by
\eqref{eq:framework-return-drift}. Each checkpoint before return lasts at least $(M+\theta)m/r>0$.
On an execution that never returns, its time factor would therefore tend to infinity, contrary to the finite expectation bound and Fatou's lemma. Hence the return time $v$ is finite almost surely. Applying Fatou's lemma at return and using $e^{\alpha S(v)}\geq1$ gives \eqref{eq:framework-return-moment}.
\qed
\end{proof}

\subsection{Permanent honest blocks}
\label{app:framework-permanent}

A return to small stock lets us attempt a favorable mining pattern, but does not yet identify a block that will remain forever. We first seek an honest block whose chain is ahead of every valid branch omitting it, including hidden branches. The starting state must also have zero usable stock and enough time for its public height to become fully validated everywhere. We call these conditions a clean lead and measure its size in blocks.

\begin{definition}[Clean lead]
\label{def:framework-clean-lead}
An honest block $B$ has a \emph{clean lead of $L$ blocks} at time $z\geq D$ if $S(z)=0$, $H$ is constant on $[z-D,z]$, and every objectively valid chain in $\mathcal T(z)$ that omits $B$ has height at most $H(z)-L$. Here $L\geq1$ is an integer and the comparison includes unpublished chains. The height plateau implies that every strict chain has height $H(z)$ at $z$, so each contains $B$.
\end{definition}

The zero-stock condition removes the initial-stock term from future loss accounting.
The comparison with all valid chains includes branches that an adversary has withheld.
Finally, the plateau and Lemma~\ref{lem:strict-height-bridge} give strict height exactly $H(z)$ at every honest miner. Since a chain omitting $B$ is lower, every strict chain contains $B$ at the starting time.

We next bound how long it takes to produce this state with a prescribed lead $L$.
The return-time lemma first reduces the stock to a fixed size. A fixed mining pattern can then remove the old stock and produce the lead with a probability independent of the preceding execution. If that pattern fails to occur, we return to small stock and try again.

\begin{lemma}[Waiting time for a clean lead]
\label{lem:framework-clean-reset}
Fix $L\geq1$ and a sufficiently small $\alpha>0$. There are constants $q_R,C_R>0$ such that, after every finite stopping time $u$, one can find a finite stopping time $z\geq u$ and an honest block $B$ mined in $(u,z]$ that has a clean lead of $L$ at $z$, with
\begin{equation}
 \mathbb E[e^{q_R(z-u)}\mid\mathcal F_u]\leq C_Re^{\alpha S(u)}.
 \label{eq:framework-reset-moment}
\end{equation}
\end{lemma}

\begin{proof}
Apply Lemma~\ref{lem:framework-stock-return} and let $v$ be the resulting stopping time with stock at most $m$.
We prescribe a mining pattern after $v$ that works under every admissible release and validation schedule.
It begins with enough honest successes to outgrow the old hidden branches, then adds $L$ further successes to create the required lead.

To specify this pattern, let $w:=1/\lambda_h$ be the length of a window in which we request exactly one honest success. Use slots of length $g:=D+2w$, with the requested success in each slot's first $w$ units and none in its remainder. Successes in consecutive slots are then more than $D$ apart. Use $2m+2$ slots to overcome the initial stock and $L$ further slots for the lead, giving $n:=2m+2+L$ successes in total. Place a silent interval of length $b$ before the first group and another between the groups, followed by a final silent interval of length $D$.
The total duration is $d_*:=2b+ng+D$.

Require no adversarial success throughout these $d_*$ time units. Require honest successes only in the $n$ specified windows, exactly one in each. Conditional on $\mathcal F_v$, independent Poisson increments give this event probability
\begin{equation}
 p:=e^{-\lambda d_*}(\lambda_h w)^n>0.
 \label{eq:framework-clean-probability}
\end{equation}
The value of $p$ depends only on the fixed constants, so the same lower bound applies whenever an attempt starts with stock at most $m$.

We first verify what the initial group achieves on this event. By time $v+b$, every honest miner knows a fully validated chain of height $H(v)$. Every first-invalid root exposed by $v+\Delta$ has also been rejected
everywhere. A later root that spoils a specified success must therefore consume an old tip from the stock $S(v)$, by the accounting in the proof of Lemma~\ref{lem:interval-charging}. There are at most $m$ such tips.
The same root cannot spoil two specified successes, because it is known to an honest miner by the first and rejected everywhere within $D$ time, before the next specified success. Hence at least $m+2$ of the first $2m+2$ successes are effective. Each increases the common available valid height by one before the next success, starting from $H(v)$.

To see why this growth removes the old threat, consider an old hidden tip with no public invalid ancestor at $v$.
The path argument in Lemma~\ref{lem:effective-stock-clearance} puts each of its path blocks above $H(v)$ in $\mathcal U(v)$. There are at most $m$ such blocks, so the tip's height is at most $H(v)+m$.
The effective successes raise the common valid height above every such tip.
Roots already exposed are rejected everywhere by the end of the following silent interval of length $b$.
Thus an old tip is either too low to attract mining or has exposed invalid ancestry. No new adversarial blocks are available during the prescribed event.

After the first group and its following silent interval, all remaining eligible invalid branches are below the common valid height, and every previously exposed invalid subtree has been rejected. Each further honest success is therefore effective. Since successive successes are separated by more than $D$, each new block becomes the unique highest eligible tip everywhere before the next success. The final $L$ specified successes consequently extend a single honest chain above all prior branches.
Let $B$ be the first block in this run and set $z=v+d_*$. Every valid chain omitting $B$ ends at or below the height of its parent, while the run adds $L$ levels. All remaining hidden tips are below the run, giving $S(z)=0$.
The last $D$ units contain no mining, and old hidden tips are too low to raise public height when released.
Hence $H$ is constant throughout this final interval, and propagation and validation finish everywhere.
The block $B$ has a clean lead of $L$ at $z$. 

If the prescribed event does not occur, end the attempt after $d_*$ time, return to stock at most $m$, and repeat.
At the attempt's end, stock is bounded by $m+Z(v,v+d_*)$. The return-time moment in \eqref{eq:framework-return-moment}, followed by the Poisson exponential moment for this count, therefore controls the whole cost of an unsuccessful attempt and its return. Writing $V$ for that duration, a uniform bound at exponent $q_r$ is
\[
 B_r:=\exp\{q_r d_*+\alpha m+
                   \lambda_a d_*(e^\alpha-1)\}.
\]

We need the contribution of a failed attempt to be less than one after weighting by its duration.
For $0<q\leq q_r$, convexity gives $e^{qV}-1\leq(q/q_r)(e^{q_rV}-1)$. Since the probability of failure is at most $1-p$, its weighted contribution is at most $1-p+(q/q_r)B_r$. Choose $q\leq pq_r/(2B_r)$ to make this at most $1-p/2$. Conditioning at each finite return time, the contributions of successive failed attempts sum to at most
$\sum_{j\geq0}(1-p/2)^j=2/p$. The successful attempt contributes at most $e^{qd_*}$. Including the initial return bound proves \eqref{eq:framework-reset-moment} with $q_R=q$ and $C_R=2e^{qd_*}/p$.
All attempt ends and return times are stopping times, so the same conditional estimates remain valid at every restart.
\qed
\end{proof}

A clean lead supplies the initial advantage, but new adversarial work continues after it is formed.
We now bound the probability that valid competing work catches the honest height available to every honest miner.
The comparison will count all new valid adversarial production as potential work on a competing branch and charge new invalid ancestry through lost effective time. Its first tie with the initial lead will be a stopping time at which we can restart the construction if needed. A tie in this upper-bound comparison need not be an actual removal
of the honest block. Absence of such a tie will suffice for permanence.

\begin{lemma}[Persistence after a clean lead]
\label{lem:framework-clean-race}
There are an integer $L\geq1$ and constants $c_f,C_f>0$ such that the following holds at every finite stopping time $z$ with a clean lead $(B,L)$.
There is a stopping time $\nu_z\in[z,\infty]$ such that
\begin{equation}
 \Pr(\nu_z<\infty\mid\mathcal F_z)\leq\tfrac14,
 \label{eq:framework-race-failure}
\end{equation}
and $B$ belongs to every strict chain at every time $t\geq z$ on $\{\nu_z=\infty\}$. Moreover, for every $v\geq0$,
\begin{equation}
 \Pr\!\left(\nu_z<\infty,\
  \nu_z-z+S(\nu_z)/\lambda>v\mid\mathcal F_z\right)
 \leq C_f e^{-c_f v}.
 \label{eq:framework-failed-race-tail}
\end{equation}
\end{lemma}

\begin{proof}
At time $z$, every valid branch omitting $B$ starts at height at most $H(z)-L$.
Until an effective honest success first extends such a branch, all new blocks on it are valid adversarial blocks.
Its growth is consequently bounded by $Z_{\mathrm{val}}(z,t)$. Every honest miner meanwhile has valid height at least $H(t-D)$ available for strict selection. To compare the two increases, define
\begin{equation}
 G_z(t):=Z_{\mathrm{val}}(z,t)-[H(t-D)-H(z)],\qquad
 \nu_z:=\inf\{t\geq z:G_z(t)\geq L\}.
 \label{eq:framework-race-deficit}
\end{equation}
The quantity $G_z$ is the competing work minus the increase in common available height.
Reaching $L$ uses up the initial lead, so a tie is treated as failure because the adversary controls tie breaking.
All quantities are determined by the complete history at time $t$, making $\nu_z$ a stopping time.

The rate difference left after charging each new adversarial success at cost at most $M$ is $\delta:=r-M\lambda_a^+>0$. Positivity follows from \eqref{eq:framework-race-rates}. We also need an allowance for the lag $D$, boundary loss $b$, and the two random deficits. Write $B_0:=r(D+b)$ for the fixed part and
$B_1:=C_H+MC_a$ for the coefficient of the confidence parameter. On the event in \eqref{eq:framework-conditional-cones} starting at $z$, use the zero initial stock in Lemma~\ref{lem:interval-charging}. For $t\geq z+D$, apply the height estimate at $t-D$ and enlarge the bad-work count to the full interval ending at $t$.
The resulting upper bound is
\begin{align}
 G_z(t)
 &\leq Z_{\mathrm{val}}(z,t)+cZ_{\mathrm{bad}}(z,t)
       -r(t-z)+r(D+b)+C_H\eta\notag\\
 &\leq-\delta(t-z)+B_0+B_1\eta.
 \label{eq:framework-race-cone}
\end{align}
The second line uses $Z_{\mathrm{val}}+cZ_{\mathrm{bad}}\leq MZ$ and the production bound.
For $z\leq t\leq z+D$, the height plateau gives $G_z(t)=Z_{\mathrm{val}}(z,t)$.
The same upper bound follows from the production estimate, since $\lambda_a^++\delta\leq r$ and $B_1\geq C_a$.

We now measure the largest upward deviation from this decreasing linear bound by
\[
 E_z:=\sup_{t\geq z}\{G_z(t)+\delta(t-z)\}.
\]
The preceding estimates give $\Pr(E_z>B_0+B_1\eta\mid\mathcal F_z)\leq2e^{-\eta}$ for every
$\eta\geq1$.
Choose a fixed integer $L>B_0+B_1\log8$.
If the comparison fails after more than $v$ time, its crossing
requires $E_z\geq L+\delta v$. Use the deviation bound with
$\eta=\log8+\delta v/(2B_1)$ to obtain 
\begin{equation}
 \Pr(\nu_z<\infty,\ \nu_z-z>v\mid\mathcal F_z)
 \leq\tfrac14 e^{-\delta v/(2B_1)},\qquad v\geq0.
 \label{eq:framework-failed-time-tail}
\end{equation}
The plateau gives $G_z(z)=0<L$, so $\nu_z>z$. Taking $v=0$ proves \eqref{eq:framework-race-failure}.

It remains to justify why no crossing makes $B$ permanent. Before a first effective honest success omitting $B$, every valid branch omitting it has height at most
\begin{equation}
 H(z)-L+Z_{\mathrm{val}}(z,t)
       =H(t-D)+G_z(t)-L.
 \label{eq:framework-competing-height}
\end{equation}
While $G_z(t)<L$, this height is strictly below the valid height available to every honest miner.
Such a branch cannot be a longest strict chain or the valid parent selected for an effective honest success.
For a mining event, apply this comparison immediately before the success, to the parent then selected.
In particular, if $D=0$, use $H(t-)$ so that the block being mined does not contribute to the information available before its own mining. Induction over mining events rules out the first effective success omitting $B$.
The height bound then continues to apply to every competing valid branch and excludes every strict-chain selection omitting $B$.

For later restarts, we also need to control the stock accumulated when a comparison does fail.
Since $S(z)=0$, a finite failure satisfies $S(\nu_z)\leq Z(z,\nu_z)$. Fix a stock threshold $x>0$ and split failures at time $z+x/[2(\lambda_a+\lambda)]$. The time-tail estimate bounds the probability of a later failure by
$\tfrac14\exp\{-\delta x/[4B_1(\lambda_a+\lambda)]\}$. An earlier failure with stock above $x$ requires more than $x$ adversarial successes before this cutoff. Their conditional Poisson mean is at most $x/2$, so applying Markov's inequality to the exponential moment with parameter $\log2$ bounds that probability by $e^{-(\log2-1/2)x}$.
Together, these bounds give an exponential tail for stock on finite failure. Apply the time bound at $v/2$ and the stock bound at $\lambda v/2$ to obtain \eqref{eq:framework-failed-race-tail}.
\qed
\end{proof}

A clean lead can fail, so one attempt does not yet give a waiting-time bound for a permanent block.
We repeat the construction after each finite crossing. Lemma~\ref{lem:framework-clean-race} controls both the time spent on that failed comparison and the stock present when the next reset begins. Together with the reset bound, these estimates control the total time spent on unsuccessful attempts.

\begin{lemma}[Waiting time for a permanent honest block]
\label{lem:framework-permanent-wait}
There are constants $q_p,C_p,\alpha>0$ such that, after every finite stopping time $u$, there is a random time $\Theta_u\geq u$ and an honest block mined in $(u,\Theta_u]$ that belongs to every strict chain at all times $t\geq\Theta_u$, with
\begin{equation}
 \mathbb E[e^{q_p(\Theta_u-u)}\mid\mathcal F_u]
 \leq C_p e^{\alpha S(u)}.
 \label{eq:framework-permanent-moment}
\end{equation}
\end{lemma}

\begin{proof}
Fix the lead $L$ from Lemma~\ref{lem:framework-clean-race}. Its failed-race tail gives a uniform exponential moment for both $\nu_z-z$ and $S(\nu_z)$ on $\{\nu_z<\infty\}$. Choose $\alpha>0$ small enough that this bound controls $e^{\alpha S(\nu_z)}$, and apply Lemma~\ref{lem:framework-clean-reset} with this $\alpha$ and $L$.
The choice ensures that the stock left by a failure has an integrable cost in the next reset's waiting-time bound.

Starting at $u$, let $z_0$ be the first clean state produced by the reset construction. Run the comparison of
Lemma~\ref{lem:framework-clean-race} from $z_0$. If it crosses at the finite time $\nu_{z_0}$, begin the next reset there. Repeat only after a finite crossing has been observed. Every reset start and completion used in this construction is therefore a stopping time, so the conditional bounds remain available after each
failure.

To sum the durations of failed attempts, let $V$ denote the time from one clean state to the next on a finite failed race, including the intervening reset.
Condition first at the crossing time. The reset bound contributes $C_Re^{\alpha S(\nu_z)}$, while the failed-race bound controls this factor together with the time already spent before the crossing. Equations~\eqref{eq:framework-reset-moment} and~\eqref{eq:framework-failed-race-tail} thus give some $q_1>0$ and finite $B_f\geq1$ such that
\[
 \mathbb E[e^{q_1V};\ \nu_z<\infty\mid\mathcal F_z]\leq B_f.
\]
We also need the weighted contribution of a failure to be less than one, so that repeated failures form a summable series. At exponent zero this contribution is the failure probability, at most $1/4$ by \eqref{eq:framework-race-failure}. For $0<q\leq q_1$, convexity bounds the contribution at exponent $q$ by $1/4+(q/q_1)B_f$.
Choose $q_p\leq\min\{q_R,q_1/(4B_f)\}$. Each failed cycle then contributes at most $1/2$ to the exponential
moment, and the first reset remains covered by its own moment bound.

Define $\Theta_u$ as the completion time of the first clean state whose comparison never crosses.
Such a state occurs almost surely, because the conditional probability of another failure is always at most $1/4$.
The definition of $\Theta_u$ depends on the future and does not make it a stopping time. We use it only to record the beginning of permanence. All conditioning below is at the finite reset and crossing times.

Sum over the number of failures before this final clean state. The first reset contributes at most $C_Re^{\alpha S(u)}$, every failed cycle contributes at most $1/2$, and the final success indicator is at
most one. Repeated conditional expectation gives 
\[
 \mathbb E[e^{q_p(\Theta_u-u)}\mid\mathcal F_u]
 \leq C_Re^{\alpha S(u)}\sum_{j\geq0}2^{-j}
 =2C_Re^{\alpha S(u)}.
\]
The honest block in the final clean state was mined after $u$. Its comparison never crosses, so
Lemma~\ref{lem:framework-clean-race} places it in every strict chain at every time from $\Theta_u$ onward.
Taking $C_p=2C_R$ proves the claimed waiting-time bound.
\qed
\end{proof}

\subsection{Block liveness and mining-time persistence}
\label{app:framework-race-properties}

The preceding lemma starts from one chosen time.
Persistence and liveness require a guarantee for every relevant interval in the observation period.
We obtain one event on which every sufficiently long interval contains a newly mined honest block that is permanent by the interval's end. The stock bound controls the cost at each starting time, and the exponential waiting-time bound lets us cover a finite grid of windows.
Every longer interval will contain one complete grid window.

\begin{lemma}[Permanent blocks in all long windows]
\label{lem:framework-permanent-windows}
There is $C_W>0$ such that, for $W=C_WK$, with probability at least $1-e^{-\kappa}$, every interval $(s,s+u]\subseteq[0,T]$ of length $u\geq2W$ contains the mining time of an honest block that belongs to
every strict chain at every time $t\geq s+u$.
\end{lemma}

\begin{proof}
First, bound the stock that can be present at a window's start.
Lemma~\ref{lem:effective-stock-uniform} gives
\[
 \Pr\!\left(\sup_{0\leq t\leq T}S(t)>C_SK\right)
 \leq\tfrac14 e^{-\kappa}.
\]
For $\tau=0$, use the stock construction with $d=M=1$ in Eq.~\ref{eq:explicit-contraction-parameters}. For $\lambda_a=0$, take $C_S=0$.

We now choose a common window length $W=C_WK$ and consider the deterministic windows $(jW,(j+1)W]$ contained in $[0,T]$.
The event $S(jW)\leq C_SK$ is known at the window's start.
We can therefore apply the conditional moment bound \eqref{eq:framework-permanent-moment} on this event.
Markov's inequality gives
\begin{align*}
 &\Pr(\Theta_{jW}>(j+1)W,\ S(jW)\leq C_SK)\\
 &\hspace{30mm}\leq C_p\exp\{\alpha C_SK-q_pW\}.
\end{align*}
Choose $C_W$ large enough for the waiting-time exponent to absorb both the initial-stock cost and the number of windows.
We also require $W\geq1/\lambda$, which bounds that number by $1+\lambda T$. Both requirements hold with
\begin{equation}
 C_W:=\max\left\{\frac1\lambda,
       \frac{\alpha C_S+3+\log(4\max\{1,C_p\})}{q_p}\right\}.
 \label{eq:framework-window-constant}
\end{equation}
The total failure probability over the grid is then at most $(1+\lambda T)e^{-3K}/4\leq e^{-\kappa}/4$.
Adding the stock failure probability keeps the total below $e^{-\kappa}$.

On the resulting event, every complete grid window contains an honest block mined after its start and permanent by its end. Fix an arbitrary $(s,s+u]\subseteq[0,T]$ with $u\geq2W$, and take $j=\lceil s/W\rceil$.
The grid window $(jW,(j+1)W]$ lies inside this interval. Its block is mined strictly after $jW\geq s$ and is permanent by $(j+1)W\leq s+u$.
The same grid event therefore covers all the required intervals at once.
\qed
\end{proof}

The window lemma directly gives block liveness.
It also puts a later permanent honest block on every chain containing a sufficiently old confirmation, which is the step needed for persistence.

\begin{proof}[Theorem~\ref{thm:spv-race-liveness}]
Take $C_{\mathrm{live}}=2C_W$ and work on the event of Lemma~\ref{lem:framework-permanent-windows}.
For every $u\geq C_{\mathrm{live}}K$ and $0\leq s\leq T-u$, the interval $(s,s+u]$ has length at least $2W$.
The lemma supplies an honest block mined in this interval that belongs to every strict chain at every time $t\geq s+u$.
The same event covers all these intervals and has probability at least $1-e^{-\kappa}$, proving the theorem.
\qed
\end{proof}

To obtain persistence, we use the window ending when the observed chain's tip was mined.
A permanent block mined within this window is later than the confirmed block and is already present when the chain is observed.
The unique-parent relation then protects the entire confirmed prefix.

\begin{proof}[Theorem~\ref{thm:spv-race-persistence}]
Take $C_{\mathrm{pers}}=2C_W$ and work on the same window event.
Fix $\sigma\geq C_{\mathrm{pers}}K$ and a non-genesis block $B$ that is $\sigma$-confirmed in $\mathcal{C}_i^{\mathrm{str}}(t)$ for some $t\leq T$.
Let $w\leq t$ be the mining time of the chain's tip and $m(B)$ the mining time of $B$.
Confirmation means $m(B)\leq w-\sigma$.
In particular, $w-\sigma\geq0$, so $(w-\sigma,w]$ lies within $[0,T]$.

Lemma~\ref{lem:framework-permanent-windows} gives an honest block $N$ mined in $(w-\sigma,w]$ and present in every strict chain at every time from $w$ onward.
Since $t\geq w$, the observed chain contains both $B$ and $N$.
Moreover, $m(N)>w-\sigma\geq m(B)$.
Along a chain, parents are mined before their descendants, so $B\preceq N$.
Every strict chain at every $t'\geq t$ contains $N$ and therefore the entire prefix ending at $B$.
The window event covers every observed chain and every confirmed block within the horizon, including their retention at all later times.
\qed
\end{proof}

\subsection{From mining-time persistence to common prefix}
\label{app:framework-common-prefix}

Persistence uses the time between the mining of a block and its chain tip.
Common prefix instead specifies how many blocks are removed from the tip.
To connect these notions, we bound how many blocks can be mined during a confirmation interval.
Let $N_{\mathrm{all}}(t)$ count all mining successes through $t$, including honest successes on invalid ancestry, and let $N_{\mathrm{all}}(u,t):=N_{\mathrm{all}}(t)-N_{\mathrm{all}}(u)$.
This is a Poisson process of rate $\lambda$, and every chain segment mined in $(u,t]$ contains at most $N_{\mathrm{all}}(u,t)$ blocks.

\begin{lemma}[Uniform bound on total mining successes]
\label{lem:framework-total-count}
For every $T\geq0$ and $\kappa\geq1$, with $K=\kappa+\log(2+\lambda T)$ and
$C_N:=4+6/\log2$, the bound
\begin{equation}
 N_{\mathrm{all}}(u,t)\leq2\lambda(t-u)+C_NK
 \quad(0\leq u\leq t\leq T)
 \label{eq:framework-total-count}
\end{equation}
holds simultaneously for all interval endpoints with probability at least $1-e^{-\kappa}/4$.
\end{lemma}

\begin{proof}
Use the grid $j/\lambda$, $0\leq j\leq\lceil\lambda T\rceil$.
For a grid interval $(p,q]$, the Poisson exponential moment gives
\[
 \mathbb E\exp\{(\log2)[N_{\mathrm{all}}(p,q)-2\lambda(q-p)]\}
 =\exp\{\lambda(q-p)(1-2\log2)\}\leq1.
\]
Thus the probability that a grid count exceeds $2\lambda(q-p)+6K/\log2$ is at most $e^{-6K}$.
There are at most $(2+\lambda T)^2$ ordered grid pairs, so their total failure probability is at most $(2+\lambda T)^2e^{-6K}\leq e^{-\kappa}/4$.
Round an arbitrary start down and its endpoint up to this grid.
The containing interval adds at most $2/\lambda$ time, hence at most four to the bound at rate $2\lambda$.
Since $K\geq1$, $C_NK$ absorbs both this rounding cost and the grid allowance.
\qed
\end{proof}

The count bound implies that a sufficiently deep retained block cannot have been mined within the last confirmation interval before its original tip.
We can therefore apply the persistence theorem to that block.

\begin{proof}[Theorem~\ref{thm:spv-backbone-common-prefix}]
Apply Theorem~\ref{thm:spv-race-persistence} and Lemma~\ref{lem:framework-total-count} with confidence parameter $\widehat\kappa:=\kappa+\log2$.
Set $\widehat K:=K+\log2$ and $\sigma:=C_{\mathrm{pers}}\widehat K$.
Their conclusions hold together with probability at least $1-e^{-\kappa}$.
Choose
\begin{equation}
 C_{\mathrm{cp}}:=(1+\log2)(2\lambda C_{\mathrm{pers}}+C_N),
 \qquad k_0:=\lceil C_{\mathrm{cp}}K\rceil.
 \label{eq:framework-cp-constant}
\end{equation}
Since $\widehat K\leq(1+\log2)K$, the count event gives, for every $w\leq T$,
\[
 N_{\mathrm{all}}(\max\{0,w-\sigma\},w)
 \leq2\lambda\sigma+C_N\widehat K\leq k_0.
\]

Fix any observed strict chain $\mathcal{C}_i^{\mathrm{str}}(t_1)$ with $t_1\leq T$, and prune its last $k_0$ blocks.
If only genesis remains, the prefix claim is immediate.
Otherwise, let $B$ be the retained tip and let $w\leq t_1$ be the mining time of the original tip.
The block $B$ and its $k_0$ descendants correspond to $k_0+1$ distinct mining successes.
If $m(B)>w-\sigma$, all these successes lie in $(\max\{0,w-\sigma\},w]$, contradicting the count bound.
Hence $m(B)\leq w-\sigma$, so $B$ is $\sigma$-confirmed in the observed chain.

Persistence places the entire prefix ending at $B$ in every $\mathcal{C}_j^{\mathrm{str}}(t_2)$ for every $t_2\geq t_1$.
This proves the required relation for $k=k_0$ and in particular for $t_2\leq T$.
Pruning any additional blocks produces a prefix of this retained chain, so the same event proves the result for every integer $k\geq k_0$.
\qed
\end{proof}

\subsection{Eventual success of the carrier attack}
\label{app:carrier-success}

We now prove the carrier-attack result of Section~\ref{sec: analysis-carrier}, under its assumption $\lambda_a\tau>1$.
Choose $\beta_c$ as in that section, giving carrier and private-fork rates
$c=\beta_c\lambda$ and $\mu=(\beta-\beta_c)\lambda$ with $c\tau>1$ and $\mu>0$.
We first establish a positive probability that one busy period continues indefinitely and freezes the valid public height.
We then use restarts to obtain such a period almost surely and show that the private fork forces publication at a finite rejection boundary.

To analyze one busy period, keep the private fork unpublished and continue the carrier rule while the queue is nonempty.
The next lemma accounts for both the queue and the honest mining that it suppresses.

\begin{lemma}[Survival of a carrier busy period]
\label{lem:carrier-busy-survival}
Suppose $c\tau>1$ and a busy period under the strategy of Section~\ref{sec: analysis-carrier} starts at a finite stopping time $v$, with valid parent $P$ and $Q_0=1$.
Under the continuation that keeps the private fork unpublished, the carrier queue never empties with probability $1-q>0$ conditional on $\mathcal F_v$, where $q\in(0,1)$ is the smallest solution in $[0,1]$ of
\begin{equation}
 q=\exp\{c\tau(q-1)\}.
 \label{eq:carrier-extinction}
\end{equation}
Throughout the busy period, the valid public height remains $h(P)$.
\end{lemma}

\begin{proof}
At the start of a window, the active carrier extends $P$ by one block and is delivered immediately to every honest miner.
Its pending chain is longer than the valid public chain, so all honest miners extend its ancestry.
Honest descendants remain on invalid ancestry, the unused carrier siblings remain hidden, and the private fork is unpublished.
Thus the window produces no new public block with valid ancestry.

At the end of the window, every honest miner rejects the active carrier and its descendants.
If another sibling is queued, releasing it immediately again gives a pending chain of height $h(P)+1$.
A withheld sibling's validation begins on its own receipt, so its window lasts the full $\tau$ even if it was mined earlier.
Induction over windows therefore keeps the valid public height at $h(P)$ until the queue empties.
Immediate delivery, immediate validation of valid blocks, and rejection of each carrier exactly $\tau$ after receipt all satisfy the model's bounds.

During a window, the carrier miners work on the fixed parent $P$.
After the stopping time $v$, the carrier process has independent Poisson increments of rate $c$.
Its successive length-$\tau$ windows consequently give the counts $N_n$ and queue recursion in \eqref{eq:carrier-stock-recursion}.
Each processed carrier consumes one queue entry and contributes $N_n$ new entries.

To count these entries, view the queue recursion as a breadth-first exploration of an auxiliary family tree: the $n$th processed vertex has $N_n$ children.
Its vertices count carriers, whose protocol parent remains $P$.
The offspring counts are independent and have probability generating function
\[
 f(s):=\mathbb E[s^{N_n}]=\exp\{c\tau(s-1)\},\qquad 0\leq s\leq1.
\]
Let $q_j$ be the probability that this family tree has become extinct by generation $j$.
Starting from one carrier gives $q_0=0$.
Conditioning on its offspring count and using independence of the descendant families gives $q_{j+1}=f(q_j)$.
Hence $q_j$ increases to the extinction probability $q$, which satisfies $q=f(q)$ by continuity.
Iteration from zero also makes $q$ no larger than any other fixed point in $[0,1]$.

It remains to show that $q<1$.
Since $f(1)=1$ and $f'(1)=c\tau>1$, there is $s_0\in(0,1)$ with $f(s_0)<s_0$.
Monotonicity of $f$ then gives $q_j\leq s_0$ for every $j$, so $q\leq s_0<1$.
Also, $q\geq q_1=f(0)>0$.
Every offspring count is finite almost surely, so extinction is equivalent to a finite family tree.
Processing its entries empties the carrier queue exactly in this case.
Otherwise there are infinitely many windows, each of length $\tau>0$, and the busy period continues indefinitely.
This proves the survival probability and the height conclusion.
\qed
\end{proof}

We have shown that one busy period has a fixed positive survival probability.
When a period ends, a new carrier arrives in finite time and starts another trial from one queue entry.
These repeated trials will freeze the valid public height long enough for the continuously growing private fork to overtake it.

\begin{proof}[Theorem~\ref{thm:carrier-success}]
Start from the stated shared-chain setup at time $t_0$.
The fixed adversarial power split gives independent Poisson processes of rates $c$ and $\mu$, whose sum is $\lambda_a$.
Carrier miners select the current valid public tip while idle and the fixed parent $P$ during a busy period.
Private-fork miners always select the current tip of their valid private chain, initially the parent of $b$.
These choices are made before each corresponding success, and all releases and validations use the causal schedule in Section~\ref{sec: analysis-carrier}.

For the probability argument, consider an auxiliary execution that keeps the private fork hidden and continues the carrier rule even at a boundary where private publication would finish the attack.
Use the same mining events for this execution and the actual strategy.
They coincide until the first rejection boundary at which the private fork is strictly longer than the valid public chain.
We show that such a boundary is almost surely finite in the auxiliary execution.

Whenever its queue empties, the next carrier success arrives after an exponential waiting time of rate $c$.
During this wait, immediate delivery lets the carrier miners track the current valid public tip.
The next success therefore starts a fresh busy period with $Q_0=1$ in finite time almost surely.
A finite busy period ends at a stopping time, and its next start is also a stopping time.
Lemma~\ref{lem:carrier-busy-survival} consequently gives the same conditional extinction probability $q<1$ at every such start.

Let $E_n$ be the event that the first $n$ busy periods all end.
Conditioning successively at their starts gives
\begin{equation}
 \Pr(E_n\mid\mathcal F_{t_0})=q^n\longrightarrow0.
 \label{eq:carrier-repeated-extinction}
\end{equation}
Each next start exists almost surely on the event that the preceding periods end.
It follows that, almost surely, some busy period never ends.
Its start $v_*$ is finite because it is preceded by only finitely many finite busy periods and finite idle waits.

Let $P_*$ be this period's valid parent.
Only finitely many blocks are mined before the finite time $v_*$, so $h(P_*)$ is finite.
By Lemma~\ref{lem:carrier-busy-survival}, the valid public height stays at $h(P_*)$ from $v_*$ onward in the auxiliary execution.
Meanwhile, the private chain gains one valid block at every private-fork success.
Since $\mu>0$, the Poisson strong law makes its height tend to infinity almost surely.
On the intersection of these probability-one events, it therefore exceeds the finite height $h(P_*)$ at some finite time.
The next rejection boundary is at most $\tau$ later and satisfies the prescribed publication condition.
Thus the first qualifying boundary in the auxiliary execution is finite almost surely.

The actual strategy agrees with this execution until that first qualifying boundary and publishes its private fork there.
Every honest miner has just rejected the active carrier and all its descendants.
The released private chain and its ancestry are delivered and validated immediately, making a valid chain strictly longer than the public chain available everywhere.
Every honest strict chain therefore switches to the private fork, which extends the parent of $b$ along a branch omitting $b$.
This removes $b$ from every honest strict chain.

The argument uses only that the initial shared chain has finite height, so it applies to every prescribed finite confirmation depth.
Such a setup is attainable from genesis by initially withholding all adversarial blocks and delivering and validating honest blocks immediately.
Because $\lambda_h>0$, an honest block and any prescribed finite number of honest descendants appear in finite time almost surely.
\qed
\end{proof}

\section{Mixed Honest Mining Policies}
\label{app:mixed-honest}

We extend the model to a fixed mixture of optimistic and strict honest
miners. Optimistic miners extend a longest chain in
$E_i^{\mathrm{opt}}(t)$, while strict miners extend a longest chain in
$E_i^{\mathrm{str}}(t)$. A strict miner continues mining on its fully
validated chain while other blocks are pending. All other assumptions
remain unchanged, including the fixed mining rates and the concurrent,
load-independent validation bound. Miner types remain fixed throughout
each execution. Let
\[
 p:=\sum_{i\text{ optimistic}}w_i\in[0,1]
\]
be the fraction of honest mining power using optimistic mining. The two
classes have aggregate rates $p\lambda_h$ and $(1-p)\lambda_h$.
Security still concerns every honest miner's strict chain, regardless
of its mining policy.

Keep the definitions of $A,H,S,Z,Z_{\mathrm{bad}}$, and
$Z_{\mathrm{val}}$. Strict miners always select objectively valid
ancestry, so $1-p\leq a(t)\leq1$. Define
\begin{equation}
\begin{split}
 d_p&:=\Delta+(1-p)\tau,\qquad h_p:=p\tau,\\
 b_p&:=p(2\Delta+\tau),\qquad D:=\Delta+\tau,\qquad
 \rho_p:=\frac{\lambda_h}{1+\lambda_h d_p}.
\end{split}
\label{mixed-parameters}
\end{equation}
Here $h_p$ and $b_p$ bound weighted losses, while $d_p$ bounds the
weighted time excluded in the height argument below. The physical
propagation and validation allowance remains $D$.

\begin{theorem}[Security with mixed honest mining policies]
\label{thm:mixed-honest-security}
Suppose
\begin{equation}
 \lambda_a<\rho_p,\qquad \lambda_a p\tau<1.
 \label{mixed-security-conditions}
\end{equation}
For every $0<\gamma<\rho_p(1-\lambda_a p\tau)$, there are positive
constants $C_{\mathrm{cg}},C_{\mathrm{live}},C_{\mathrm{pers}}$, and
$C_{\mathrm{cp}}$ such that the following holds. For every $T\geq0$
and $\kappa\geq1$, put $K=\kappa+\log(2+\lambda T)$. Under every
admissible strategy from genesis, with probability at least
$1-e^{-\kappa}$, simultaneously:
\begin{enumerate}
\item $L_i(t)-L_i(u)\geq\gamma(t-u)$ for every honest miner $i$ and
every $0\leq u\leq t\leq T$ with $t-u\geq C_{\mathrm{cg}}K$.
\item Block liveness holds for every window
$\sigma\geq C_{\mathrm{live}}K$ within $[0,T]$.
\item Mining-time persistence holds with observation horizon $T$ for every
confirmation duration $\sigma\geq C_{\mathrm{pers}}K$, with retention at all future times.
\item Common prefix holds through $T$ for every integer
$k\geq\lceil C_{\mathrm{cp}}K\rceil$.
\end{enumerate}
The constants depend only on $(\beta,\lambda,\Delta,\tau,p,\gamma)$,
and are uniform over finite miner counts, power distributions realizing
$p$, and admissible strategies.
\end{theorem}

Equivalently, a sufficient condition is $\beta<B_p$, where
\begin{equation}
 B_p:=\min\left\{
 \frac{2}{2+x_p+\sqrt{4+x_p^2}},\,
 \frac{1}{p\lambda\tau}\right\},
 \qquad x_p:=\lambda[\Delta+(1-p)\tau],
 \label{mixed-share-bound}
\end{equation}
where the second term is $+\infty$ when $p\tau=0$.
This is a sufficient bound. We do not establish a matching attack
threshold for $0<p<1$.

\subsection{Weighted loss and height growth}

\begin{lemma}[Weighted carrier loss]
\label{mixed-weighted-loss}
The adapted charge counts $R_0,R_1$ from the proof of
Lemma~\ref{lem:interval-charging} satisfy, for every $u\leq t$,
\begin{equation}
\begin{split}
 A(t)-A(u)&\geq(t-u)-b_p-h_p[R_0(u,t)+R_1(u,t)],\\
 R_0(u,t)&\leq S(u),\qquad R_1(u,t)\leq Z_{\mathrm{bad}}(u,t).
\end{split}
\label{mixed-charged-loss}
\end{equation}
From genesis, $A(t)\geq t-h_pZ_{\mathrm{bad}}(0,t)$.
\end{lemma}
\begin{proof}
Only optimistic miners work on invalid ancestry. For a first-invalid
root $B$, each such miner can mine in its subtree only between receipt
and rejection of $B$, for at most $\tau$ time units in total. Summing
their power weights bounds this root's loss by
$\sum_{i\text{ optimistic}}w_i\tau=p\tau=h_p$.

Roots known to an honest miner by $u+\Delta$ are rejected everywhere
by $u+2\Delta+\tau$. Since the power share on invalid ancestry is at
most $p$, their combined boundary loss is at most $b_p$.
For every other root, use the original charge: its first public
adversarial descendant at height at least $H(u)$.
Any miner wasting work after $u+\Delta$ is optimistic and already
knows a valid chain of height $H(u)$. Its selected tip has at least
that height. If the tip is adversarial, it qualifies for the charge.
Otherwise, the first honest descendant of this root was mined after
$u+\Delta$, and its adversarial parent was already public at height
at least $H(u)$. Thus every root attracting wasted work is charged.

Distinct roots give distinct charged blocks. As in the original proof,
old charged blocks belong to $S(u)$ and new charged blocks contribute
to $Z_{\mathrm{bad}}(u,t)$. The counts are zero through $u+\Delta$.
Thereafter, the boundary classification and each qualifying publication
are determined by the current history, so both counts are adapted.
The per-root and boundary costs prove \eqref{mixed-charged-loss}.
At genesis, every encountered root is a new adversarial block, giving
the last claim without a boundary allowance.
\qed
\end{proof}

\begin{lemma}[Height growth with mixed mining policies]
\label{mixed-height-growth}
For every $0<r<\rho_p$, there is $C_H^{\mathrm{cond}}>0$ such that,
after every finite stopping time $u$, for every
$\mathcal F_u$-measurable $\eta\geq1$,
\begin{equation}
\begin{split}
 \Pr\bigl[\,&H(t)-H(u)\geq r[A(t)-A(u)]-C_H^{\mathrm{cond}}\eta\\
 &\text{for every }t\geq u\mid\mathcal F_u\bigr]\geq1-e^{-\eta}.
\end{split}
\label{mixed-conditional-height}
\end{equation}
For every $T\geq0$ and $\kappa\geq1$, the same height inequality
holds simultaneously for all $0\leq u\leq t\leq T$ with deficit
$C_HK$ and probability at least $1-e^{-\kappa}$.
\end{lemma}
\begin{proof}
Begin an epoch at $u$ and a new epoch after each selected effective
honest success. In an epoch beginning at $s$, permit optimistic
successes after $s+\Delta$ and strict successes after $s+D$.
Select the first permitted success on valid ancestry. The permitted
miner then has the preceding selected block eligible under its own
mining rule. In the first epoch, use a valid chain attaining $H(u)$.
Consequently, if $Q(t)$ counts selections after $u$, then
$H(t)-H(u)\geq Q(t)$.

Let $e_i(t)$ indicate that miner $i$'s pre-success parent has valid
ancestry, and let $J_i(t)$ indicate that its selection is permitted.
Use the epoch immediately preceding $t$ at a selected success and
strict inequalities at permission deadlines. These indicators are
predictable. Define
\[
 B(t):=\int_u^t\sum_iw_i e_i(v)J_i(v)\,dv.
\]
An epoch excludes at most $p\Delta+(1-p)D=d_p$ weighted effective
time, including when another selection ends it before a strict
miner's wait expires. Hence
\begin{equation}
 B(t)\geq A(t)-A(u)-d_p[Q(t)+1].
 \label{mixed-selected-clock}
\end{equation}
The count $Q$ has predictable intensity
$\lambda_h\sum_iw_i e_i(t)J_i(t)$. Put
\[
 \xi:=\frac12\left(\frac{\lambda_h}{r}-1-\lambda_h d_p\right)>0,
 \qquad \chi:=\lambda_h(1-e^{-\xi}).
\]
Then $\exp\{-\xi Q(t)+\chi B(t)\}$ is a nonnegative local
martingale starting from one at $u$, and thus a supermartingale.
The conditional maximal inequality and
\eqref{mixed-selected-clock} give, except with conditional
probability $e^{-\eta}$,
\[
 Q(t)\geq\frac{\chi}{\xi+\chi d_p}[A(t)-A(u)]
             -\frac{\chi d_p+\eta}{\xi+\chi d_p}
 \qquad(t\geq u).
\]
Since $\xi/(1-e^{-\xi})\leq1+\xi$, the coefficient of effective
time exceeds $r$. Taking
$C_H^{\mathrm{cond}}=(1+\chi d_p)/(\xi+\chi d_p)$ proves
\eqref{mixed-conditional-height}.

For uniformity over $[0,T]$, apply the conditional estimate with
$\eta=K$ at grid starts $j/\lambda\leq T$. Their joint failure
probability is at most $(1+\lambda T)e^{-K}<e^{-\kappa}$.
Round any other start upward to the next grid point. Monotonicity
of $H$ and the one-Lipschitz property of $A$ lose at most
$r/\lambda$, including intervals shorter than the rounding gap.
Thus $C_H=C_H^{\mathrm{cond}}+r/\lambda$ suffices.
\qed
\end{proof}

\subsection{Stock, permanent blocks, and security}

\begin{proof}[Theorem~\ref{thm:mixed-honest-security}]
We adapt the arguments in Appendix~\ref{sec: proof-analysis} using
Lemmas~\ref{mixed-weighted-loss} and~\ref{mixed-height-growth}.
The stock definition is unchanged, and stock can increase only through
new adversarial production. The hidden path above $H(u)$ and the old
blocks consumed by carriers remain disjoint subsets of initial stock.
Thus Lemma~\ref{lem:effective-stock-clearance} gives
\begin{equation}
 H(v)-H(u)+R_0(u,v)>S(u)
 \quad\Longrightarrow\quad S(v)\leq Z(u,v)-R_1(u,v).
 \label{mixed-stock-clearance}
\end{equation}
This argument does not require strict miners to extend pending blocks.
Likewise, Lemma~\ref{lem:strict-height-bridge} retains the full physical
propagation and validation lag:
\begin{equation}
 L_i(u)\leq H(u)\leq\min_jL_j(u+D).
 \label{mixed-strict-bridge}
\end{equation}

\noindent\emph{Stock contraction.}
Choose auxiliary rates and a production coefficient satisfying
\begin{equation}
\begin{split}
 &\lambda_a<\lambda_a^+<r<\rho_p,\qquad \lambda_a^+h_p<1,\\
 &C_a:=
 \begin{cases}
  1/\log(\lambda_a^+/\lambda_a),&\lambda_a>0,\\
  0,&\lambda_a=0.
 \end{cases}
\end{split}
\label{mixed-auxiliary-rates}
\end{equation}
By Lemma~\ref{mixed-height-growth} and the unchanged Poisson estimate
of Lemma~\ref{lem:effective-poisson-stopping}, after every finite
stopping time $u$, with conditional probability at least
$1-2e^{-\eta}$,
\begin{equation}
\begin{split}
 H(t)-H(u)&\geq r[A(t)-A(u)]-C_H^{\mathrm{cond}}\eta,\\
 Z(u,t)&\leq\lambda_a^+(t-u)+C_a\eta
 \qquad(t\geq u).
\end{split}
\label{mixed-conditional-cones}
\end{equation}
Here $\eta\geq1$ may be $\mathcal F_u$-measurable. For
$\lambda_a=0$, use $Z=0$ and the same auxiliary rate choices.
Put
\begin{equation}
\begin{split}
 M&:=\max\{1,rh_p\},\qquad d:=1-\lambda_a^+h_p,\qquad
 \rho_S:=\frac{\lambda_a^+M}{r}<1,\\
 C_0&:=rb_p+\max\{C_H^{\mathrm{cond}},1\},\qquad
 C_1:=\frac{\lambda_a^+(C_0+1)}r+C_a.
\end{split}
\label{mixed-stock-constants}
\end{equation}
Starting with $S_0=S(u)$, let $\sigma$ be the first $v\geq u$
such that
\begin{equation}
 r(v-u)-rh_pR_1(u,v)\geq MS_0+C_0\eta+1.
 \label{mixed-stock-check}
\end{equation}
This is a finite stopping time. Its left side dominates
$r(v-u)-rh_pZ(u,v)$, which has positive asymptotic drift because
$\lambda_a h_p<1$. The crossing has no upward overshoot, and
$\sigma-u\geq\eta/r$.

On the event in \eqref{mixed-conditional-cones}, the weighted loss
bound gives
\begin{align*}
 H(v)-H(u)+R_0(u,v)
 &\geq r(v-u)-rh_pR_1(u,v)+(1-rh_p)R_0(u,v)\\
 &\qquad-rb_p-C_H^{\mathrm{cond}}\eta\\
 &\geq r(v-u)-rh_pR_1(u,v)-(M-1)S_0-C_0\eta.
\end{align*}
At $\sigma$, this is at least $S_0+1$, so
\eqref{mixed-stock-clearance} applies. The equality at the crossing
and the production cone yield
\begin{align}
 S(\sigma)&\leq\rho_SS_0-dR_1(u,\sigma)+C_1\eta
              \leq\rho_SS_0+C_1\eta,\label{mixed-stock-contraction}\\
 \sigma-u&\leq
 \frac{MS_0+(C_0+1+rh_pC_a)\eta}{rd},\notag\\
 \sup_{u\leq v\leq\sigma}S(v)&\leq
 \left(1+\frac{\rho_S}{d}\right)S_0+\frac{C_1}{d}\eta.\notag
\end{align}
Each newly consumed block contributes $\lambda_a^+h_p-1=-d$ to
replacement stock. Carrier consumption and valid competing work
therefore use the same adversarial production budget.

\noindent\emph{Uniform stock and chain growth.}
Repeat the checkpoints from empty stock with
$\eta=\kappa+\log[8(2+rT)]$. There are at most $1+rT$ starts by
$T$, even if a conditional cone fails, because each checkpoint
lasts at least $\eta/r$. Conditional union bounds give total
failure probability at most
$2(1+rT)e^{-\eta}\leq e^{-\kappa}/4$.
The contraction bounds stock at every checkpoint by
$C_1\eta/(1-\rho_S)$. The last inequality in
\eqref{mixed-stock-contraction} then gives
\[
 \sup_{0\leq u\leq T}S(u)
 \leq\frac{(1+d)C_1}{d(1-\rho_S)}\eta\leq C_SK
\]
for a fixed constant $C_S$, since $\eta\leq CK$ for a fixed $C$.
Substituting into \eqref{mixed-charged-loss} proves, with the same
probability, the simultaneous estimates
\begin{equation}
\begin{split}
 S(u)&\leq C_SK,\\
 A(t)-A(u)&\geq(t-u)-h_pZ_{\mathrm{bad}}(u,t)-C_AK
 \qquad(0\leq u\leq t\leq T).
\end{split}
\label{mixed-uniform-stock}
\end{equation}
When $\lambda_a=0$, stock is zero and $A(t)=t$ directly.
When $h_p=0$, the construction has $d=M=1$ and $A(t)=t$, even
if the physical validation delay is positive.

For chain growth, choose $r<\rho_p$ and a small $\varepsilon>0$
such that $r(1-\lambda_a h_p-\varepsilon)>\gamma$. If $\lambda_a h_p>0$,
Lemma~\ref{lem:effective-poisson-uniform} bounds $Z(u,t)$ by
$(\lambda_a+\varepsilon/h_p)(t-u)+C_\varepsilon K$.
Apply the stock, production, and height estimates with confidence
$\kappa+\log3$. Combining them with \eqref{mixed-strict-bridge}
gives, simultaneously over miners and interval endpoints,
\[
 L_i(t)-L_i(u)\geq
 r(1-\lambda_a h_p-\varepsilon)(t-u)-CK.
\]
The confidence shift only changes the constant $C$. For $\lambda_a=0$
or $h_p=0$, use $A(t)=t$ directly. A sufficiently large $C_{\mathrm{cg}}$
absorbs the deficit on intervals of length at least
$C_{\mathrm{cg}}K$. The physical lag remains $D$, not
$\Delta+h_p$.

\noindent\emph{Return to a fixed stock bound.}
We next prove the mixed counterpart of
Lemma~\ref{lem:framework-stock-return}. For every sufficiently
small $\alpha>0$, there are an integer $m\geq1$ and $q_r>0$
such that, after any finite stopping time $u$, there is a finite
stopping time $v\geq u$ with $S(v)\leq m$ and
\begin{equation}
 \mathbb E[e^{q_r(v-u)}\mid\mathcal F_u]
 \leq e^{\alpha S(u)}.
 \label{mixed-return-moment}
\end{equation}
For $\lambda_a=0$, take $v=u$. Otherwise, put $c=rh_p$ and
$s=S(u)$. Choose $\theta>0$ and then $\epsilon>0$ small enough
that
\[
 \rho_R:=\frac{\lambda_a^+(M+\theta)}r+C_a\epsilon<1,
 \qquad C_H^{\mathrm{cond}}\epsilon<\theta/2.
\]
Use a new checkpoint, the first $t\geq u$ at which
\begin{equation}
 r(t-u)-cR_1(u,t)\geq(M+\theta)s.
 \label{mixed-return-check}
\end{equation}
Write $\sigma$ for this finite stopping time. For sufficiently
large $s$, apply \eqref{mixed-conditional-cones} with
$\eta=\epsilon s\geq1$. The same clearance calculation gives
\[
 H(\sigma)-H(u)+R_0(u,\sigma)
 \geq(1+\theta-C_H^{\mathrm{cond}}\epsilon)s-rb_p>s.
\]
Thus, on this good event,
\begin{equation}
 S(\sigma)\leq\rho_Rs,\qquad
 \sigma-u\leq C_ts,\qquad
 C_t:=\frac{M+\theta+cC_a\epsilon}{r(1-\lambda_a^+h_p)}.
 \label{mixed-return-good}
\end{equation}
The bad event has conditional probability at most
$2e^{-\epsilon s}$.

If $h_p>0$, let $\zeta_Q$ be the first $t\geq0$ with
$rt-cZ(u,u+t)\geq Q$, where $Q=(M+\theta)s$. It dominates
$\sigma-u$. For sufficiently small $\xi>0$,
\begin{equation}
\begin{split}
 \omega&:=r\xi-\lambda_a(e^{c\xi}-1)>0,\\
 \mathbb E[e^{\omega\zeta_Q}\mid\mathcal F_u]
 &\leq e^{\xi Q},\qquad S(\sigma)\leq s+\zeta_Q/h_p.
\end{split}
\label{mixed-return-busy-moment}
\end{equation}
Indeed, $r-\lambda_a c>0$ makes $\zeta_Q$ finite. Stop the
Poisson martingale
$\exp\{-\xi[rt-cZ(u,u+t)]+\omega t\}$ at
$\zeta_Q\wedge n$. There is no upward overshoot, so Fatou's
lemma proves the moment bound. The last bound follows from
$S(\sigma)\leq s+Z(u,u+\zeta_Q)$ and the crossing equality.

Choose $\xi(M+\theta)<\epsilon$, followed by small
$\alpha,q>0$ with $2\alpha/h_p+2q\leq\omega$ and
$qC_t<\alpha(1-\rho_R)$. Cauchy--Schwarz bounds the bad-event
contribution to
$\mathbb E[e^{\alpha S(\sigma)+q(\sigma-u)}\mid\mathcal F_u]$
by
\[
 \sqrt2\exp\!\left\{\alpha s-
 \frac{\epsilon-\xi(M+\theta)}2s\right\}.
\]
The good-event contribution is at most
$e^{(\alpha\rho_R+qC_t)s}$.
If $h_p=0$, the checkpoint instead has deterministic duration
$Q/r$. Using $S(\sigma)\leq s+Z(u,u+Q/r)$ and the Poisson
exponential moment bounds the bad contribution by
\[
 \sqrt2\exp\!\left\{\alpha s-
 \frac12\left[\epsilon-\frac{M+\theta}{r}
 \{\lambda_a(e^{2\alpha}-1)+2q\}\right]s\right\}.
\]
This case includes $p=0$ with $\tau>0$.

In both cases, $\alpha$ can be arbitrarily small. Choose $q$
with $qC_t<\alpha(1-\rho_R)$, and $m$ sufficiently large that the two
contributions sum to at most $e^{\alpha s}$ for every $s>m$.
With $q_r=q$, this gives the conditional drift bound
\[
 \mathbb E[e^{\alpha S(\sigma)+q_r(\sigma-u)}
           \mid\mathcal F_u]\leq e^{\alpha S(u)}.
\]
Repeat the checkpoint whenever stock exceeds $m$, and stop the
sequence at the first return. The resulting exponentials form
a nonnegative supermartingale. Before return, checkpoint
durations are at least $(M+\theta)m/r>0$. Nonreturn would make
the time factor diverge, contradicting Fatou's lemma and the
bounded expectations. A final application of Fatou's lemma
proves \eqref{mixed-return-moment}.

\noindent\emph{Creating a clean lead.}
Use Definition~\ref{def:framework-clean-lead}, which requires
zero usable stock, a public-height plateau of length $D$, and
a lead of $L$ blocks over every valid branch omitting a new
honest block $B$, including hidden branches. At a return time
$v$ with $S(v)\leq m$, prescribe the mining pattern used in
Lemma~\ref{lem:framework-clean-reset}. Its physical durations
are
\begin{equation}
\begin{split}
 b_{\mathrm{phys}}&:=2\Delta+\tau,\qquad w:=1/\lambda_h,
 \qquad g:=D+2w,\\
 n&:=2m+2+L,\qquad d_*:=2b_{\mathrm{phys}}+ng+D.
\end{split}
\label{mixed-clean-pattern}
\end{equation}
Require exactly one honest success in the first $w$ time
units of each slot, no other honest successes, and no
adversarial successes. Silent intervals of length
$b_{\mathrm{phys}}$ precede the first $2m+2$ slots and separate
them from the last $L$ slots. A final silent interval of
length $D$ completes validation. The conditional probability
is
\begin{equation}
 \pi:=e^{-\lambda d_*}(\lambda_h w)^n>0.
 \label{mixed-clean-probability}
\end{equation}
It does not depend on the types of the honest successes.

Before the first slot, every miner has fully validated height
$H(v)$, and roots exposed by $v+\Delta$ have been rejected.
An invalid root can spoil only an optimistic success. It
spoils at most one specified success, because successive
successes are more than $D$ apart. Each such root consumes a
distinct old-stock charge, so at least $m+2$ of the first
$2m+2$ successes are effective. Whether optimistic or strict,
each effective success extends at least the preceding common
fully validated height and advances that height before the
next success. All old hidden usable tips have height at most
$H(v)+m$, by the path argument behind
\eqref{mixed-stock-clearance}. The first group outgrows them,
and the middle physical wait clears exposed invalid branches.
The last $L$ successes therefore form one valid chain. Every
remaining hidden tip either lies below this chain or has a public
invalid ancestor. Thus stock is zero, and no valid hidden tip can
raise public height during the final silent interval.
Their first block has a clean lead after this physical wait.

After a failed pattern, stock is at most $m+Z(v,v+d_*)$.
Let $V$ be the pattern duration plus the next return time.
The return estimate and the Poisson moment give the uniform
bound
\[
 \mathbb E[e^{q_rV}\mid\mathcal F_v]
 \leq B_r:=\exp\{q_rd_*+\alpha m
                    +\lambda_a d_*(e^\alpha-1)\}.
\]
For $0<q\leq q_r$, convexity bounds its weighted failure
contribution by $1-\pi+(q/q_r)B_r$. Choose $q$ small enough
that this is at most $1-\pi/2$. Geometric summation, including
the initial return, proves that a clean lead of the prescribed
size is completed at a finite stopping time $z\geq u$, with
\begin{equation}
 \mathbb E[e^{q_R(z-u)}\mid\mathcal F_u]
 \leq C_Re^{\alpha S(u)}
 \label{mixed-reset-moment}
\end{equation}
for positive $q_R,C_R$. As in
Lemma~\ref{lem:framework-clean-reset}, the block is mined
after $u$ and $\alpha$ can be chosen sufficiently small.

\noindent\emph{Preserving a clean lead.}
At a clean state $(B,L)$ at time $z$, define
\begin{equation}
\begin{split}
 G_z(t)&:=Z_{\mathrm{val}}(z,t)-[H(t-D)-H(z)],\\
 \nu_z&:=\inf\{t\geq z:G_z(t)\geq L\}.
\end{split}
\label{mixed-clean-race}
\end{equation}
Until the first effective honest success omitting $B$, a
competing valid branch grows only through valid adversarial
successes. Both honest mining policies have a valid chain of
height at least $H(t-D)$ eligible for selection. While
$G_z(t)<L$, a competing valid parent is strictly shorter and
cannot be selected by either type. Mining decisions use
pre-event heights, including $H(t-)$ when $D=0$, as in
Lemma~\ref{lem:framework-clean-race}.

Put
\[
 \delta:=r-M\lambda_a^+>0,\qquad
 B_0:=r(D+b_p),\qquad B_1:=C_H^{\mathrm{cond}}+MC_a.
\]
Using zero starting stock, the mixed loss bound and
\eqref{mixed-conditional-cones} give, for $t\geq z+D$,
\begin{align*}
 G_z(t)&\leq Z_{\mathrm{val}}(z,t)+rh_pZ_{\mathrm{bad}}(z,t)
               -r(t-z)+B_0+C_H^{\mathrm{cond}}\eta\\
       &\leq-\delta(t-z)+B_0+B_1\eta.
\end{align*}
For $z\leq t\leq z+D$, the clean plateau gives
$G_z(t)=Z_{\mathrm{val}}(z,t)$. The same cone holds because
$\lambda_a^++\delta\leq r$ and $B_1\geq C_a$.
Choose an integer $L>B_0+B_1\log8$. Applying the cone with
$\eta=\log8+\delta v/(2B_1)$ gives
\begin{equation}
 \Pr(\nu_z<\infty,\ \nu_z-z>v\mid\mathcal F_z)
 \leq\tfrac14e^{-\delta v/(2B_1)}
 \qquad(v\geq0).
 \label{mixed-race-failure-tail}
\end{equation}
On no crossing, $B$ is permanent.

On a finite crossing, $S(\nu_z)\leq Z(z,\nu_z)$.
For $x\geq0$, split according to whether
$\nu_z-z>x/[2(\lambda_a+\lambda)]$. The failure-time bound
and the Poisson exponential moment with parameter $\log2$
give
\begin{align*}
 &\Pr(\nu_z<\infty,\ S(\nu_z)>x\mid\mathcal F_z)\\
 &\qquad\leq\tfrac14
 e^{-\delta x/[4B_1(\lambda_a+\lambda)]}
 +e^{-(\log2-1/2)x}.
\end{align*}
Thus there are positive $c_f,C_f$ such that
\begin{equation}
 \Pr\left(\nu_z<\infty,\
 \nu_z-z+\frac{S(\nu_z)}\lambda>x\mid\mathcal F_z\right)
 \leq C_fe^{-c_fx}.
 \label{mixed-failed-race-stock}
\end{equation}

\noindent\emph{Retrying after observed failures.}
Choose $\alpha$ in \eqref{mixed-reset-moment} small enough
to absorb the failed-stock tail in
\eqref{mixed-failed-race-stock}. On a finite crossing,
apply the reset construction at $\nu_z$. The duration $V$
of this failed race and the following reset satisfies, for some
$q_1,B_f>0$,
\[
 \mathbb E[e^{q_1V}\mathbf1_{\{\nu_z<\infty\}}
 \mid\mathcal F_z]\leq B_f.
\]
Set $V=0$ when no crossing occurs. The probability of a finite
crossing is at most $1/4$. For
$0<q\leq q_1$, convexity therefore gives
\[
 \mathbb E[e^{qV}\mathbf1_{\{\nu_z<\infty\}}
             \mid\mathcal F_z]
 \leq\tfrac14+(q/q_1)B_f\leq\tfrac12
\]
after decreasing $q$ if needed.
Sum over the number of observed failures, conditioning only
at finite reset and crossing times. The first successful
clean lead supplies an honest block mined after $u$ and
permanent from a random time $\Theta_u$, with
\begin{equation}
 \mathbb E[e^{q_*(\Theta_u-u)}\mid\mathcal F_u]
 \leq C_*e^{\alpha S(u)}
 \label{mixed-permanent-moment}
\end{equation}
for positive $q_*,C_*$. The final time $\Theta_u$ need not
be a stopping time. No conditional estimate is applied at
that time.

\noindent\emph{From waiting times to security.}
Combine \eqref{mixed-uniform-stock} and
\eqref{mixed-permanent-moment} on a deterministic grid with
$W=C_WK$, where
\[
 C_W\geq\max\left\{\frac1\lambda,
 \frac{\alpha C_S+3+\log(4\max\{1,C_*\})}{q_*}\right\}.
\]
At a grid start $jW$, conditional Markov inequality bounds
the probability that $\Theta_{jW}>jW+W$ together with
$S(jW)\leq C_SK$ by
$C_*e^{\alpha C_SK-q_*W}\leq e^{-3K}/4$.
There are at most $1+\lambda T$ complete grid windows by
$T$. Adding their failure probabilities and the stock
failure probability $e^{-\kappa}/4$ gives at most
$e^{-\kappa}$. Every interval in $[0,T]$ of length at least
$2W$ contains a complete grid window. Its new honest block
is permanent by the interval's end. This proves block
liveness, uniformly over all such intervals. Take
$C_{\mathrm{live}}=C_{\mathrm{pers}}=2C_W$.

For persistence, let an observed strict chain have tip mined
at $w$, and let $B$ be confirmed by a mining-time gap
$\sigma\geq2W$. The interval $(w-\sigma,w]$ contains a new
honest block $N$ permanent from a time at most $w$.
Both $B$ and $N$ lie on the observed strict chain, and
$m(B)\leq w-\sigma<m(N)$. Hence $B$ is an ancestor of $N$,
so its prefix is retained in every honest strict chain
at every later time.

For common prefix, the total mining process still has rate
$\lambda$. Apply persistence and
Lemma~\ref{lem:framework-total-count} with confidence
$\kappa+\log2$, and write $\widehat K=K+\log2$.
Set $\sigma=C_{\mathrm{pers}}\widehat K$ and choose
$C_{\mathrm{cp}}\geq
(1+\log2)(2\lambda C_{\mathrm{pers}}+C_N)$, where $C_N$
is the total-count constant. Every interval
$((w-\sigma)^+,w]$ with $w\leq T$ then contains at most
$k_0:=\lceil C_{\mathrm{cp}}K\rceil$ mining successes.
A non-genesis block retained after pruning $k_0$ blocks
and its $k_0$ descendants cannot all lie in that interval.
The retained block is therefore $\sigma$-confirmed, and
persistence preserves its prefix. Pruning any larger number
of blocks also satisfies common prefix on the same event.

Finally, apply the four component guarantees with confidence
$\kappa+\log4$, retaining the confidence splits within
the chain-growth and common-prefix arguments. Enlarge the
constants to absorb the additive $\log4$ in $K$.
Their intersection has probability at least $1-e^{-\kappa}$,
as required.
\qed
\end{proof}

\subsection{Attack scope and unresolved tightness}
\label{mixed-attack-scope}

\noindent\textbf{Why the carrier attack does not directly extend.}
The attack in Section~\ref{sec: analysis-carrier} queues invalid
siblings of a fixed valid parent. It relies on all honest miners
extending invalid ancestry while the public valid height stays fixed.
With strict mining power $(1-p)\lambda_h>0$ and immediate validation
of valid blocks, strict miners keep raising the valid height.
Queued siblings of the old parent can therefore become too short to
attract optimistic miners. A nonempty queue no longer guarantees that
another diversion window can begin. Thus the queue-survival argument
does not establish a mixed-population attack threshold by replacing
$\tau$ with $p\tau$. The rejection window still lasts $\tau$, while
$p$ determines the honest power exposed during that window.

\noindent\textbf{An attack using delayed validation of valid blocks.}
\label{mixed-validation-attack}
Consider a finite shared fully validated chain containing a non-genesis
target block $b$ at any prescribed finite confirmation depth. Such a state
is attainable from genesis by initially delivering and validating
honest blocks immediately and withholding all adversarial blocks.
Measure elapsed time from the start of the attack. Deliver every
newly mined honest block immediately and complete its validation exactly
$\tau$ after creation at every honest miner, including its creator.
Keep all adversarial mining on a valid private fork from the parent
of $b$. This is an admissible schedule under the validation oracle.

Analyze the auxiliary execution in which the adversary never
publishes its fork. All public ancestry is valid and extends the
initial shared chain. Optimistic miners mine at
height $H(t)$, while strict miners mine at height $H(t-\tau)$,
where $H(t)$ is set to the initial chain's height for $t<0$.
Let $T_k$ be the first time the public tree reaches a new height $k$.
No height-$k$ block can be validated before $T_k+\tau$.
Thus only optimistic successes, of aggregate rate $p\lambda_h$,
can produce height $k+1$ before that time. If none does, the first
height-$k$ block and all its ancestry are validated everywhere at
$T_k+\tau$, and the next success of either honest type raises
the height. Later siblings cannot validate before the first record.

The record times are stopping times. After the first height increase
above the initial chain, independent Poisson increments give
independent, identically distributed record spacings $W$, with
\begin{equation}
 \Pr(W>t)=
 \begin{cases}
 e^{-p\lambda_h t},&0\leq t<\tau,\\
 e^{-p\lambda_h\tau}e^{-\lambda_h(t-\tau)},&t\geq\tau.
 \end{cases}
 \label{mixed-record-spacing}
\end{equation}
For $p>0$, integration gives
\[
 \mathbb E W
 =\frac{1-e^{-p\lambda_h\tau}}{p\lambda_h}
  +\frac{e^{-p\lambda_h\tau}}{\lambda_h}
 =\frac{1-(1-p)e^{-p\lambda_h\tau}}{p\lambda_h}.
\]
The renewal strong law therefore gives the almost-sure public-height
growth rate
\begin{equation}
 r_{\mathrm{val}}(p)=
 \begin{cases}
 \displaystyle\frac{p\lambda_h}
 {1-(1-p)e^{-p\lambda_h\tau}},&p>0,\\[2mm]
 \displaystyle\frac{\lambda_h}{1+\lambda_h\tau},&p=0.
 \end{cases}
 \label{mixed-validation-rate}
\end{equation}
If $\lambda_a>r_{\mathrm{val}}(p)$, the private fork eventually
overtakes the public height almost surely, despite any finite initial
lead. The actual attack agrees with this execution until the fork is
strictly longer, then publishes and validates it immediately at every
honest miner. Every strict chain adopts the longer fork, removing $b$.
The calculation holds for every finite distribution of
honest mining power and every permitted network bound $\Delta$,
because immediate delivery satisfies that bound. The strict rate
comparison gives no conclusion at equality.

This attack uses the oracle's permission to delay validation even
at a block's creator. Requiring immediate self-validation would
change the renewal calculation. The sufficient-security proof uses
only the upper validation bound and does not require this particular
completion schedule.

For $\Delta=0$, $0<p<1$, and $\tau>0$,
\begin{equation}
\begin{split}
 \mathbb E W
 &=\frac{1}{\lambda_h}
   +\frac{1-p}{p\lambda_h}(1-e^{-p\lambda_h\tau})\\
 &<\frac{1}{\lambda_h}+(1-p)\tau
  =\frac{1}{\rho_p}.
\end{split}
\label{mixed-rate-gap}
\end{equation}
The strict inequality follows from $1-e^{-x}<x$ for $x>0$.
Hence this attack does not match the mixed height lower bound.
This comparison does not prove that the sufficient security bound is
loose, since another attack might match it. Tightness for an interior
mixture remains unresolved.

\noindent\textbf{Why scaling only the carrier loss is insufficient.}
The same schedule gives a counterexample to retaining the original
height rate $\rho=\lambda_h/(1+\lambda_h\Delta)$ while replacing
only the carrier-loss coefficient $\tau$ by $p\tau$.
Take
\[
 \Delta=0,\qquad \lambda_h=1,\qquad \tau=10,\qquad
 p=0.01,\qquad \lambda_a=0.2.
\]
Both proposed conditions hold: $\lambda_a<\rho=1$ and
$\lambda_a p\tau=0.02<1$. Nevertheless,
\eqref{mixed-validation-rate} gives
$r_{\mathrm{val}}(0.01)\approx0.09596<0.2$, so private mining succeeds.
A pathwise bound gives the same conclusion without the renewal formula.
For each block mined by a strict miner, if its parent was mined after
the attack began, the two mining times differ by at least $\tau$.
Let $N_{\mathrm{SPV}}(t)$ count optimistic honest successes during
the attack. Any public chain therefore satisfies
\[
 H(t)\leq N_{\mathrm{SPV}}(t)+t/\tau+O(1),\qquad
 \limsup_{t\to\infty}\frac{H(t)}t
 \leq0.01+0.1<0.2.
\]
The constant absorbs the initial prefix and its first outgoing edge.
Thus the revised height
argument is necessary even when the optimistic power fraction is
positive.

\subsection{Incorporating network delay into the attack}
\label{mixed-network-attack}

The preceding attack delivers honest blocks immediately. To exploit the
network bound as well, start from the same shared fully validated chain
and deliver each new honest block to every other honest miner exactly
$\Delta$ after its creation. Its creator receives it immediately.
Complete validation exactly $\tau$ after receipt at every honest miner,
including the creator. This schedule respects ancestor-complete receipt,
because each creator already knows its block's ancestry, and every
other recipient receives each ancestor no later than its descendant.
Again, all adversarial power mines a valid
private fork from the target block's parent. Until publication, all public
blocks are valid.

Define the reference rate
\begin{equation}
 r_{\mathrm{net}}(p)
 =\left(\Delta+\frac{1}{r_{\mathrm{val}}(p)}\right)^{-1}
 =\begin{cases}
 \displaystyle\left(\Delta+
 \frac{1-(1-p)e^{-p\lambda_h\tau}}{p\lambda_h}\right)^{-1},&p>0,\\[2mm]
 \displaystyle\frac{\lambda_h}{1+\lambda_h(\Delta+\tau)},&p=0.
 \end{cases}
 \label{mixed-network-rate}
\end{equation}
This is the public-height rate in the limit of negligible individual
mining power. For a finite population, a miner can extend its own blocks
before other miners receive them, requiring the following correction.

\begin{proposition}[A network-dependent private-fork attack]
\label{prop:mixed-network-attack}
Let $\lambda_i=w_i\lambda_h$ be honest miner $i$'s mining rate,
let $w_{\max}:=\max_iw_i$, and set
\begin{equation}
\begin{split}
 \varepsilon_\Delta
 &:=\sum_i\lambda_i(1-e^{-\lambda_i\Delta})\\
 &\leq\Delta\lambda_h^2\sum_iw_i^2
 \leq\Delta\lambda_h^2w_{\max}.
\end{split}
\label{mixed-network-correction}
\end{equation}
In the auxiliary execution in which the private fork is never published,
the public height $H(t)$ satisfies, almost surely,
\begin{equation}
 r_{\mathrm{net}}(p)
 \leq\liminf_{t\to\infty}\frac{H(t)}t
 \leq\limsup_{t\to\infty}\frac{H(t)}t
 \leq r_{\mathrm{net}}(p)+\varepsilon_\Delta.
 \label{mixed-network-height-bounds}
\end{equation}
Consequently, if
$\lambda_a>r_{\mathrm{net}}(p)+\varepsilon_\Delta$,
the private-fork attack removes a non-genesis target block at any
prescribed finite initial confirmation depth with probability one.
\end{proposition}

\begin{proof}
Let $h_0$ be the initial public height. On the same honest mining
processes, define a virtual height $H_0$ that ignores the creator's
earlier access to its own blocks. Set $H_0(t)=h_0$ for $t\leq0$.
At a success of miner $i$ at time $s>0$, update
\begin{equation}
 H_0(s)=\max\{H_0(s-),\,1+H_0((s-d_i)-)\},\qquad
 d_i=\begin{cases}
 \Delta,&i\text{ is optimistic},\\
 D,&i\text{ is strict}.
 \end{cases}
 \label{mixed-network-auxiliary-height}
\end{equation}
Keep $H_0$ constant between successes. The left limits exclude the
current success when $d_i=0$.

Mark an optimistic success at $s$ as exceptional if the same miner
has an earlier post-attack success in $(s-\Delta,s)$.
For a strict success, use the window $(s-D,s-\tau)$ instead.
Both windows have length $\Delta$. Let $E(t)$ count exceptional
successes by time $t$. Positive fixed-lag ties between Poisson successes
have probability zero. We claim that, pathwise outside this null event,
\begin{equation}
 H_0(t)\leq H(t)\leq H_0(t)+E(t).
 \label{mixed-network-coupling}
\end{equation}

For the lower bound, at a success time $s$, an optimistic miner knows
every public block created before $s-\Delta$, while a strict miner has
validated every public block created before $s-D$. Each therefore
extends a chain at least as high as the corresponding delayed global
height. Induction over successes gives $H_0\leq H$.
For the upper bound, consider a nonexceptional success of miner $i$.
Every eligible post-attack block received from another miner was created
no later than $s-d_i$. There is also no eligible block mined by $i$ after that
time: for an optimistic miner it would mark the success as exceptional,
and for a strict miner it would either do so or still await validation.
Thus, by the induction hypothesis, the new block's height is at most
$1+H_0((s-d_i)-)+E(s-)$, which is at most $H_0(s)+E(s)$.
The initial shared prefix also satisfies this bound. At an exceptional
success, $H$ increases by at most one, $E$ increases by one, and $H_0$
does not decrease. This proves~\eqref{mixed-network-coupling}.

After a record of $H_0$ at time $T$, no new record is possible before
$T+\Delta$. From $T+\Delta$ to $T+D$, only optimistic successes can
raise the virtual height. After $T+D$, either type can do so.
Consequently, after the first height increase, the record spacings of
$H_0$ are independent copies of $\Delta+W$, where $W$ has the law
in~\eqref{mixed-record-spacing}. The strong Markov property at the
record times and the renewal strong law yield
$H_0(t)/t\to r_{\mathrm{net}}(p)$ almost surely.

The exceptional indicators need not be independent, but their count
also obeys a strong law. If $\Delta=0$, then $E=0$.
Otherwise, partition time into cells $[kD,(k+1)D)$.
For $k\geq1$, the exceptional count in a cell depends only on Poisson
successes in $[(k-1)D,(k+1)D)$. Counts in even-indexed cells are
independent and identically distributed, as are those in odd-indexed
cells. At a success of miner $i$, the probability of an earlier own
success in the relevant window is $1-e^{-\lambda_i\Delta}$.
The Poisson intensity formula therefore gives mean
$D\varepsilon_\Delta$ per cell. Applying the strong law separately
to the two parity classes, then using monotonicity to include partial
cells, gives $E(t)/t\to\varepsilon_\Delta$ almost surely.
Together with~\eqref{mixed-network-coupling}, this proves the height
bounds. The inequalities in~\eqref{mixed-network-correction} follow
from $1-e^{-x}\leq x$ and $\sum_iw_i=1$.

The private fork grows at almost-sure rate $\lambda_a$.
If this rate exceeds the public-height upper bound, it eventually
overtakes the public tree despite any finite initial lead.
The actual attack agrees with the auxiliary execution until this time.
The adversary then publishes its strictly longer fork, delivering and
validating it immediately everywhere. Every strict chain adopts it
and removes the target block.
\qed
\end{proof}

For a fixed finite honest power distribution, choosing between the two
delivery schedules gives a successful attack whenever
\[
 \lambda_a>
 \min\{r_{\mathrm{val}}(p),\,
        r_{\mathrm{net}}(p)+\varepsilon_\Delta\}.
\]
For every fixed strict inequality $\lambda_a>r_{\mathrm{net}}(p)$,
the network-delay attack succeeds for all sufficiently decentralized
finite distributions realizing $p$. Indeed, when $\Delta>0$, it suffices
that
\[
 w_{\max}<\frac{\lambda_a-r_{\mathrm{net}}(p)}
 {\Delta\lambda_h^2}.
\]
Thus the fully decentralized resilience threshold is at most the unique
$U_p\in(0,1)$ satisfying
\begin{equation}
 U_p\lambda=r_{\mathrm{net}}(p),\qquad
 \lambda_h=(1-U_p)\lambda.
 \label{mixed-network-share-bound}
\end{equation}
For fixed $p$, $\lambda$, $\Delta$, and $\tau$, the left side strictly
increases with the adversarial share, while the right side continuously
decreases to zero, giving uniqueness.

When $\Delta>0$, $r_{\mathrm{net}}(p)<r_{\mathrm{val}}(p)$,
so this asymptotic upper bound is strictly stronger than the
immediate-delivery attack bound. When $\Delta=0$, the correction
vanishes and the original attack rate is recovered exactly.
At $p=0$, $r_{\mathrm{net}}(0)=\rho_0$, matching the sufficient
bound in the fully decentralized limit. At $p=1$, it reduces to the
ordinary private-mining rate $\lambda_h/(1+\lambda_h\Delta)$,
and the carrier attack can impose a lower threshold.
For $0<p<1$ and $\tau>0$, the strict inequality
in~\eqref{mixed-rate-gap} gives $r_{\mathrm{net}}(p)>\rho_p$.
The strengthened attack therefore narrows the gap without resolving
tightness for an interior mixture.

\subsection{Endpoints and numerical illustration}
\label{mixed-endpoints-examples}

At $p=1$, $d_p=\Delta$ and $h_p=\tau$, recovering the
all-optimistic security bound. At $p=0$, $a=1$ and $h_p=b_p=0$,
while $d_p=D$. Invalid carriers waste no honest work in this case,
but validation can still delay valid-chain growth.
When $\tau=0$, the mining policies coincide and
\eqref{mixed-security-conditions} reduces to the network-only bound.
The arguments also admit $\Delta=0$, and zero adversarial power
gives empty stock. The sufficient conditions retain strict rate
margins for the concentration estimates.

\begin{table}[h]
\caption{Sufficient resilience bounds for mixed honest mining with
$\lambda\Delta=0$ and $\lambda\tau=4$.}
\label{tab:mixed-security-bounds}
\centering
\begin{tabular}{cc}
\toprule
Honest SPV power fraction $p$ & Sufficient bound $B_p$\\
\midrule
$0$ & $0.190983$\\
$0.5$ & $0.292893$\\
$0.75$ & $0.333333$\\
$1$ & $0.250000$\\
\bottomrule
\end{tabular}
\end{table}

Table~\ref{tab:mixed-security-bounds} evaluates the sufficient
bound $B_p$ for $\lambda\Delta=0$ and $\lambda\tau=4$.
These values are analytical lower bounds on resilience, rather than
empirical measurements or established exact thresholds.
For $p=0.5$ in the same example, the attack in
\eqref{mixed-validation-rate} succeeds above adversarial share
approximately $0.367987$, whereas the sufficient bound is
approximately $0.292893$. The intervening range remains unresolved
by these arguments.

Network delay makes the new attack bound strictly stronger. For
$p=0.5$, $\lambda\tau=4$, and $\lambda\Delta=1$, the sufficient
bound is approximately $0.232408$. Immediate delivery gives an attack
bound of approximately $0.367987$, whereas
\eqref{mixed-network-share-bound} lowers the fully decentralized upper
bound to approximately $0.288092$. With 100 equal-power honest miners,
half optimistic and half strict, including
$\varepsilon_\Delta$ gives a finite-population attack bound of
approximately $0.292144$.

\end{document}